\documentclass[12pt]{article}
\usepackage{amsmath,amsthm,amssymb}
\usepackage{graphicx,psfrag,epsf}
\usepackage{enumerate}
\usepackage{natbib}
\usepackage{url} 
\usepackage{color}
\newcommand{\blind}{0}

\usepackage[dvipdfm]{geometry}
\usepackage{subfigure}
\usepackage{caption}
\usepackage{epsfig}
\usepackage{epstopdf}
\usepackage[center]{titlesec}
\usepackage{booktabs}
\usepackage{threeparttable}
 \usepackage{algorithm}
\usepackage{color}
\usepackage{tikz}
\usepackage{tablefootnote}
\usepackage{threeparttable}
\usepackage{bbm, dsfont}

\newtheorem{theorem}{Theorem}
\newtheorem{lemma}{Lemma}
\newtheorem{remark}{Remark}
\newtheorem{proposition}{Proposition}
\newtheorem{corollary}{Corollary}

\theoremstyle{definition}

\numberwithin{equation}{section}

\begin{document}

\def\spacingset#1{\renewcommand{\baselinestretch}%
{#1}\small\normalsize} \spacingset{1}


\if0\blind
{
  \title{\bf Inference in semiparametric formation models for directed networks}
  \author{Lianqiang Qu, Lu Chen, Ting Yan\\
School of Mathematics and Statistics, Central China Normal University,\\
Wuhan, Hubei, 430079, P.R.China\\
and Yuguo Chen\\
Department of Statistics, University of Illinois at Urbana-Champaign, \\
Champaign, IL 61820
}
  \maketitle
} \fi

\if1\blind
{
  \bigskip
  \bigskip
  \bigskip
  \begin{center}
    {\LARGE\bf Inference in semiparametric formation models for directed networks}
\end{center}
  \medskip
} \fi

\bigskip
\begin{abstract}
We propose a semiparametric model for dyadic link formations in directed networks.
The model contains a set of degree parameters that
measure different effects of popularity or outgoingness across nodes, a regression parameter vector
that reflects the homophily effect resulting from the nodal attributes or pairwise covariates associated with edges,
and a set of latent random noises with unknown distributions.
Our interest lies in inferring the unknown degree parameters and homophily parameters.
The dimension of the degree parameters increases with the number of nodes.
Under the high-dimensional regime, we develop a kernel-based least squares approach to estimate the unknown parameters.
The major advantage of our estimator is that it does not encounter the incidental parameter problem for the homophily parameters.
We prove consistency of the estimators of the degree parameters and homophily parameters.
We establish high-dimensional central limit theorems for the proposed estimators
and provide several applications of our general theory,
including testing the existence of degree heterogeneity, testing sparse signals and recovering the support.
Simulation studies and a real data application are conducted to illustrate the finite sample performance of the proposed method.
\end{abstract}

\noindent%
\vfill

\newpage
\spacingset{1.8} 
\section{Introduction}

Network data frequently arise in various fields such as genetics, sociology, finance and econometrics.
Networks consist of nodes and edges that connect one node to another.
Nodes may represent individuals in social networks, users in email networks, or countries in international trade networks.
Homophily and degree heterogeneity are two commonly observed characteristics of real-world social and economic networks.
Homophily indicates that nodes in a network are more likely to connect with others that have similar attributes,
rather than with those that possess dissimilar attributes.
Degree heterogeneity describes the variation in the number of edges among nodes;
specifically, a small number of nodes have many edges, while a large number of nodes have relatively fewer edges.
Quantifying the influence of these network features on edge formation is a critical issue in network analysis.
For a comprehensive review of network analysis, see \cite{kolaczyk2014statistical}.

The presence and extent of homophily, along with degree heterogeneity,
have significant implications for network formation (\citealp{graham2017econometric, yan2019}).
Several parametric models have been proposed to characterize these two critical network features
(e.g., \citealp{graham2017econometric,dzemski2019empirical,yan2019,de2020econometric, GRAHAM202023}),
where the estimation and inference methods depend on a specified distribution for the latent random noise.
For instance, \cite{graham2017econometric} and \cite{yan2019} assumed a logistic distribution,
while \cite{dzemski2019empirical} assumed a normal distribution.
However, network modeling based on a specific parametric distribution is vulnerable to model misspecification and potential instabilities.
When the assumed parametric distribution is inappropriate,
inference may have non-negligible bias, as demonstrated by the simulation results in Table \ref{Tab:CYan} of Section \ref{section:simulations}.

In this study, we introduce a semiparametric framework designed to model both homophily and degree heterogeneity in directed networks.
It is noteworthy that semiparametric inferences on the homophily parameter have been studied within the context of undirected networks \citep{toth2017semiparametric, zeleneev2020identification, candelaria2020semiparametric}.
These will be elaborated upon after the presentation of our main findings.
In our model, two node-specific parameters, $\alpha_i$ and $\beta_i$, are assigned to each node,
with $\alpha_i$  representing out-degree and $\beta_i$ signifying in-degree.
These are collected as two distinct sets: the out-degree parameters $\{\alpha_i\}_{i=1}^n$
and the in-degree parameters $\{\beta_j\}_{j=1}^n$,
with $n$ denoting the number of nodes in the directed graph.
Moreover, the model includes a common homophily parameter, $\gamma$,
for pairwise covariates $\{X_{ij}\}_{i,j=1}^n$ amongst the nodes.
Within the model, an edge from node $i$ to $j$ is depicted
if the sum of the degree effects $\alpha_i+\beta_j$
and the effect of covariates  $X_{ij}^\top \gamma$
surpasses a latent random noise with an undetermined distribution.
This modelling approach adopts an additive structure from the existing literature
\citep[e.g.][]{graham2017econometric,dzemski2019empirical,yan2019}.

Since edge formation is influenced by both homophily and degree heterogeneity,
it is equally important to estimate homophily parameter $\gamma$ and degree parameters $\{\alpha_i\}_{i=1}^n$ and $\{\beta_j\}_{j=1}^n$.
It is well-known that the maximum likelihood estimator (MLE) of $\gamma$
exhibits a non-negligible bias \citep{NS1948, graham2017econometric, FW2016},
thereby necessitating bias-correction procedures for valid inference \citep{yan2019,hughes2022estimating}.
A natural question arises: {\it Is it possible to derive an unbiased estimator for homophily parameter
while simultaneously obtaining the estimators for the degree parameters?}
To the best of our knowledge, this issue has not yet been addressed in the existing literature.
Moreover, estimating degree parameters allows us to test for the presence of degree heterogeneity
within a sub-network with a fixed or increasing number of nodes.
This analysis can provide insights into the extent of variation in degrees within the sub-network
and help determine whether certain nodes exhibit significantly different levels of attractiveness or popularity compared to others.

To address this problem, we develop a projection approach to estimate the unknown parameters, which involves three steps.
In the first step, we derive a kernel smoothing estimator for the conditional density of a special regressor given the other covariates.
A special regressor is defined as a continuous covariate that possesses a positive coefficient \citep{lewbel1998, lewbel2000semiparametric}.
In the following step, we obtain an unbiased estimator of $\gamma$
by projecting the covariates onto the subspace spanned by the column vectors of the design matrix of degree parameters.
The projection procedure eliminates the potential bias arising from the estimation of degree parameters.
Finally, we estimate the degree parameters by employing a constrained least squares method.
We establish consistency and asymptotic normality of the resulting estimators.
Since our asymptotic distributions for the estimators of degree parameters are high-dimensional,
it improves upon the fixed dimensional results presented in \cite{yan2019}.
Recently, \cite{shao} proposed a high-dimensional central limit theorem  for the analysis of undirected network data;
however, their method relies on the assumption of a logistic distribution for the noise.
Based on the asymptotic results, we develop  hypothesis testing methods to study three related problems:
(1) testing whether $\alpha_{i}$ and $\beta_{i}$ are equal to zero, which corresponds to testing for sparse signals;
(2) identifying which of $\alpha_{i}$ and $\beta_{i}$ are non-zero, thereby achieving support recovery;
and (3) assessing whether $\alpha_{i}=\alpha_{j}$ and $\beta_{i}=\beta_{j}$ within a sub-network,
which tests for the existence of degree heterogeneity.

As previously mentioned, inferences have been drawn in semiparametric models for undirected networks
\citep{toth2017semiparametric,candelaria2020semiparametric,zeleneev2020identification}.
While all these studies treated degree parameters as random variables, we consider them as fixed parameters.
\cite{toth2017semiparametric} employed conditional methods to eliminate the degree parameters
and proposed a tetrad inequality estimator for the homophily parameter.
\cite{zeleneev2020identification} developed estimators for homophily parameters based on a conditional pseudo-distance between two nodes,
which measures the similarity of degree heterogeneity.
However, the estimation of degree parameters is not addressed in \cite{toth2017semiparametric},
\cite{candelaria2020semiparametric} and \cite{zeleneev2020identification}.
Additionally, although \cite{gao2020nonparametric} derived identification results for nonparametric models of undirected networks,
the estimation aspect remains unexplored.
{\color{black}The work most closely related to our paper is \cite{candelaria2020semiparametric},
which also employs the special regressor method to address model identification issues.
However, the estimation strategies of \cite{candelaria2020semiparametric}
rely on the information derived from all sub-networks formed by groups of four distinct nodes.
In contrast, our estimator for the homophily parameter is based on a projection method.
For theoretical analysis, we develop
the Gaussian approximation method and derive the accurate inverse of the design matrix $V=U^\top U$
to analyze our proposed estimator with an explicit solution,
while \cite{candelaria2020semiparametric} primarily utilizes the U-statistic theory.
Consequently, the theoretical framework in our study significantly differs from that of \cite{candelaria2020semiparametric}.
Finally, our projection-based method can substantially reduce computational complexity compared to a U-statistic of order four
used in \cite{candelaria2020semiparametric}.
Further comparisons with the method of \cite{candelaria2020semiparametric} are presented in Section C.2 of the Supplementary Material.

We summarize our novel contributions as follows:
(1)  We propose a projection methodology that thoroughly addresses the identification issue of the homophily.
(2) We introduce a kernel-based least squares estimator for all unknown parameters and demonstrate their consistency and asymptotic normality.
(3) Our estimator is easy to compute since it has an explicit solution, unlike those estimators in \cite{toth2017semiparametric,candelaria2020semiparametric} and \cite{zeleneev2020identification} that require optimizing an objective function involving a complex U-statistic.
(4) The proposed semiparametric framework for  directed networks addresses
the issue of model misspecification when the parametric model in \cite{yan2019} is not correctly specified.
}

The remainder of this paper is organized as follows.
Section \ref{GNF:model} introduces the semiparametric network formation models.
Section \ref{section2-identification-estimation} provides the conditions for model identification and presents the estimation method.
Section \ref{section:Properties} provides  consistency and Gaussian approximations of the proposed estimators.
Section \ref{section:application} presents some applications of the general theory.
Section \ref{section:simulations} reports on the simulation studies and a real data analysis.
Section \ref{Sec:discussion} presents concluding remarks.
The technical details and additional numerical results are in the Online Supplementary Material.

We conclude this section by introducing some notation.
Let $e_{i}$ be a $(2n-1)$-dimensional row vector
with the $i$th element being 1 and 0 otherwise $(i=1,\dots,2n-1)$,
and $e_{2n}$ be the $(2n-1)$-dimensional zero vector.
For any $x=(x_1,\dots, x_n)^\top\in\mathbb{R}^n,$
we define the $\ell_p$-norm $\|x\|_p=(\sum_{i=1}^nx_i^p)^{1/p}$
and the $\ell_\infty$-norm $\|x\|_{\infty}=\max_{1\le i\le n}|x_i|.$
We define $\text{diag}\{d_1,\dots,d_n\}$ as a diagonal matrix, where $d_i$ is the $i$th element on the diagonal.
Let $I_n$ denote the $n\times n$ identity matrix, and $\mathbb{I}(\cdot)$ denote the indicator function.
For a matrix $D=(d_{ij})\in \mathbb{R}^{p\times q}$, define $\|D\|_{\max}=\max_{1\le l\le p,\ 1\le k\le q}|d_{ij}|.$
For the positive sequences $\{a_n\}$ and $\{b_n\},$
we write $a_n=o(b_n)$ if $a_n/b_n\rightarrow 0$ as $n\rightarrow \infty,$
and write $a_n=O(b_n)$ if there exists a constant $C$ such that $a_n\le  Cb_n$ for all $n.$
For vectors $x=(x_1,\dots,x_n)^\top$ and $y=(y_1,\dots,y_n)^\top,$
we write $x\le y$ if $x_i\le y_i$ for all $1\le i\le n.$
We denote the set $\{1, 2, \dots, n\}$ as $[n].$
For any set $B,$ denote its cardinality as $|B|.$
Denote by $\lfloor x\rfloor$ the integer part of a positive real number $x.$
The symbol $\mathcal{N}_{n}(\mu, \Sigma)$ is reserved for an
$n$-dimensional multivariate normal distribution with mean $\mu$ and covariance matrix $\Sigma$.
We use the subscript ``$0$" to denote the true parameter, under which the data are generated.

\section{Semiparametric network formation models} \label{GNF:model}
Consider a directed network composed of $n$ nodes, labelled as $``1, 2, \ldots, n"$.
Let $A=(A_{ij})_{n\times n}$ denote the adjacency matrix.
If there is a directed edge from node $i$ to $j$, we encode $A_{ij}=1$; otherwise, we set $A_{ij}=0$.
In this study, we assume that the network does not contain any self-loops (i.e., $A_{ii}=0$ for $i\in[n]$).
Let $X_{ij}=(X_{ij1},\dots,X_{ij,p+1})^\top\in \mathbb{R}^{p+1}$ denote the covariate for the node pair $(i,j)$,
which can be either a link-dependent vector or a function of the node-specific covariates.
For example, if node $i$ has a $d$-dimensional characteristic $W_i$,
the pairwise covariate $X_{ij}$ can be constructed as $X_{ij}=\|W_i-W_j\|_2.$
Under this specific choice of $X_{ij}$,  a smaller value of $X_{ij}$ indicates greater similarity between nodes  $i$ and $j$.
We assume that $p$ is fixed and that $A_{ij}$ are conditionally independent across $1\le i\neq j\le n$, given the covariates $X_{ij}$.

To capture the two aforementioned network features: homophily effects and degree heterogeneity,
we propose the following semiparametric link formation model for directed networks:
\begin{align}\label{model}
    A_{ij}=\mathbb{I}(\alpha_{i}+\beta_{j}+X_{ij}^\top\gamma-\varepsilon_{ij}>0),
\end{align}
where $\alpha_i$ represents the outgoingness parameter of node $i$,
$\beta_j$ denotes the popularity parameter of node $j$,
and $\gamma$ is the regression coefficient associated with the covariate $X_{ij}$.
In this model, $\varepsilon_{ij}$ denotes the unobserved latent noise,
for which we assume that $\mathbb{E}(\varepsilon_{ij}|X_{ij})=0$ almost surely.
The model posits that an edge from node $i$ to $j$ is formed
if the total effect, comprising the outgoingness of node $i$, the popularity of node $j$,
and the covariates effect $X_{ij}^\top\gamma$, exceeds the latent noise.

The parameter sets $\{\alpha_i\}_{i=1}^n$ and $\{\beta_i\}_{i=1}^n$ characterize the heterogeneity of nodes in their participation in network connections.
Larger values of $\alpha_{i}$ and $\beta_{i}$ indicate a higher propensity for node $i$ to establish links with other nodes in the network.
The term $X_{ij}^\top\gamma$ introduces the concept of homophily.
For instance, when $X_{ij}=\|W_i-W_j\|_2$ and $\gamma<0$,
a larger value of $X_{ij}^\top\gamma$ increases the likelihood of homophilous nodes interacting with each other.
Therefore, $\gamma$ captures the homophily effect of covariates (\citealp{graham2017econometric}).
The noise term $\varepsilon_{ij}$ accounts for the unobserved random factors
that influence the decision to form a specific interaction from node $i$ to $j$.


\section{Identification and estimation}
\label{section2-identification-estimation}

\subsection{Identification of parameters}
\label{Sec:Identification}
In this section, we discuss the conditions under which  model~\eqref{model} is identifiable.
Let $\alpha=(\alpha_1,\dots,\alpha_n)^\top$ and $\beta=(\beta_1,\dots,\beta_n)^\top.$
Clearly, model~\eqref{model} remains invariant under the transformation of the parameter vector $(\alpha, \beta, \gamma)$
to $(a\alpha + c,\ a\beta-c,\ a\gamma)$, where $a>0$ and $c\in\mathbb{R}$.
This invariance can be observed as follows:
\[
A_{ij}=\mathbb{I}(\alpha_{i}+\beta_{j}+{X}_{ij}^\top{\gamma}-\varepsilon_{ij}>0)
=\mathbb{I}(\widetilde{\alpha}_i+\widetilde{\beta}_j+{X}_{ij}^\top{\widetilde{\gamma}}-\widetilde \varepsilon_{ij}>0),
\]
where $\widetilde{\alpha}_i = a\alpha_i + c$, $\widetilde{\beta}_j = a\beta_j-c$, $\widetilde{\gamma}=a\gamma$
and $\widetilde{\varepsilon}_{ij}=a\varepsilon_{ij}$.
Thus, model~\eqref{model} is scale-shift invariant,
which requires specific constraints on the parameters $\alpha_i,~\beta_i$ and $\gamma$ for identification.
A common approach to avoid scale invariance is to set $\gamma_{k}=1$,
where $\gamma_k$ is the $k$th component of $\gamma,$
and $k$ is chosen such that $X_{ijk}$ is a continuous random variable.
In addition, one can establish $\sum_{i=1}^n \beta_i=0$ or $\beta_n=0$ to avoid shift invariance.
However, the identification of the parameters in model \eqref{model} is critically dependent on
the support of the joint distribution of $(X_{ij}, \varepsilon_{ij})$.
To illustrate this, we consider an example in which identification fails even when we set $\gamma_{k}=1$.
Let $X_{ij}$ be a random variable with support $(-4,-3)\cup(0,1)$.
We define $\gamma_{1}=\widetilde \gamma_1=1,$ and set $\alpha_{i}+\beta_{j}=1$ for $1\le i\neq j\le n,$
along with $\widetilde\alpha_{i}+\widetilde\beta_j=a(\alpha_{i}+\beta_{j})$, where $a\in[1,2].$
In addition, let $\varepsilon_{ij}$ and $\widetilde\varepsilon_{ij}$ be drawn from the uniform distribution on $(-1,1).$
In this scenario, we have
\begin{align*}
   & X_{ij}>\varepsilon_{ij}-1~~\text{if}~~X_{ij}\in (0,1)~~\text{and}~~X_{ij}<\varepsilon_{ij}-1~~\text{otherwise},\\
   &  X_{ij}>\widetilde\varepsilon_{ij}-a~~\text{if}~~X_{ij}\in (0,1) ~~\text{and}~~X_{ij}<\widetilde\varepsilon_{ij}-a~~\text{otherwise}.
\end{align*}
In other words, $A_{ij}=\widetilde A_{ij}$ almost surely if $a\in [1,2]$,
where $\widetilde A_{ij}=\mathbb{I}(\widetilde{\alpha}_i+\widetilde{\beta}_j+{X}_{ij}^\top{\widetilde{\gamma}}-\widetilde \varepsilon_{ij}>0)$.
Therefore, model \eqref{model} cannot be identified within the parameter set
$\{(\alpha, \beta): 1\le \alpha_i + \beta_j\le 2,\ 1\le i \neq j \le n \}.$
However, if we change the support of $X_{ij}$ to $(-4,4)$,
then the support of $(\alpha_i+\beta_j-\varepsilon_{ij})$ becomes a subset of $(-4,4).$
This leads to $\mathbb{P}(A_{ij}\neq \widetilde A_{ij})>0$ when $\widetilde \alpha_i+\widetilde \beta_j\neq \alpha_i+\beta_j.$
In this case, the identification issue does not arise.

Motivated by the aforementioned example, we consider the following conditions to ensure model identification.

\noindent{\bf Condition (C1).}
There exists at least one $k\in[p+1]$ such that $\gamma_{0k}>0$,
and the conditional distribution of $X_{ijk}$ given $X_{ij(-k)}$
is absolutely continuous with respect to the Lebesgue measure with nondegenerate conditional density $f(x|X_{ij(-k)}),$
where $\gamma_{0k}$ is the $k$th element of $\gamma_0$
and $X_{ij(-k)}=(X_{ij1},\dots,X_{ij(k-1)},X_{ij(k+1)},\dots,X_{ij(p+1)})^\top$.

The covariate $X_{ijk}$ that satisfies Condition (C1) is termed a special regressor \citep{lewbel1998,candelaria2020semiparametric}.
For the sake of simplicity, we assume that $X_{ij1}$ meets Condition (C1), and we denote $Z_{ij}=X_{ij(-1)}.$

\noindent{\bf Condition (C2).}
The conditional density $f(x|Z_{ij})$ of $X_{ij1}$ given $Z_{ij}$ has support $(B_L,B_U)$, where $-\infty\le B_L<0<B_U\le \infty.$
Furthermore, the support for $-(\alpha_{0i}+\beta_{0j}+Z_{ij}^\top\eta_0-\varepsilon_{ij})/\gamma_{01}$ is a subset of $(B_L,B_U),$
where $\eta_0=(\gamma_{02},\dots,\gamma_{0(p+1)})^\top.$

Condition (C2) restricts the support of $X_{ij1},$
which is mild and has been widely adopted by \cite{MANSKI1985313}, \cite{lewbel1998, lewbel2000semiparametric} and \cite{candelaria2020semiparametric}.
Conditions (C1) and (C2) do not impose any restrictions on the distribution of $Z_{ij}.$
Therefore, this identification strategy accommodates discrete covariates in $Z_{ij}.$

\noindent{\bf Condition (C3).} $\varepsilon_{ij}$ $(1\le i\neq j\le n)$ are independent of $X_{ij}$ and $\mathbb{E}(\varepsilon_{ij})=0.$

Condition (C3) is mild and can be relaxed to the scenario where $\varepsilon_{ij}$ is conditionally independent of $X_{ij1}$ given $Z_{ij}$,
as discussed in Section A of the Supplementary Material.

Following the works of \cite{lewbel1998} and \cite{candelaria2020semiparametric}, we define
\begin{align*}
Y_{ij} = \frac{A_{ij}-\mathbb{I}(X_{ij1}>0)}{f(X_{ij1}|Z_{ij})}.
\end{align*}
Additionally, we let $\eta_0=(\gamma_{02},\dots,\gamma_{0(p+1)})^\top$
and $\theta=(\alpha_1,\dots,\alpha_n,\beta_1,\dots,\beta_{n-1})^\top$.
The conditional expectation of $Y_{ij}$ is presented below.

\begin{theorem}
    \label{Theorem:LSE}
If Conditions (C1)-(C3) hold, then we have
\begin{align*}
   \mathbb{E}(Y_{ij}|Z_{ij})=(\alpha_{0i}+\beta_{0j}+Z_{ij}^\top\eta_0)/\gamma_{01}.
\end{align*}
\end{theorem}

The proofs of Theorem \ref{Theorem:LSE} and all other theoretical results are provided in the Supplementary Material.
{\color{black}Theorem \ref{Theorem:LSE} states that
under Conditions (C1)-(C3), the expectation of $Y_{ij}$ conditional on $Z_{ij}$ exhibits an additive structure
involving $\alpha_{0i}$, $\beta_{0j}$, and a homophily term $Z_{ij}^\top \eta_0$.
This indicates that if $\alpha_{0},~\beta_{0}$ and $\eta_0$ are identifiable,
then they can be recovered using $Y_{ij}$.
However, according to Theorem \ref{Theorem:LSE},
the parameters $\alpha_{0i}$ and $\beta_{0j}$ are identified only up to scale rather than shift.
This is because, for any constant $c$,
the model remains invariant when transforming $(\alpha_{0i}+\beta_{0j})$ to $(\alpha_{0i}+c)+(\beta_{0j}-c)$.
Consequently, additional constraints on the parameters are necessary to ensure their identification.

We next fix some notations. Let $N=n(n-1)$.
Define $Y=(Y_1^\top,\dots,Y_n^\top)^\top\in\mathbb{R}^{N}$ and
$Z=(Z_1^\top,\dots,Z_n^\top)^\top\in \mathbb{R}^{N\times p}$,
where for each $i\in[n]$, $Y_i=(Y_{i1},\dots,Y_{i,i-1},Y_{i,i+1},\dots,Y_{in})^\top\in\mathbb{R}^{n-1}$
and $Z_i=(Z_{i1},\dots,Z_{i,i-1},Z_{i,i+1},\dots,Z_{in})^\top\in \mathbb{R}^{(n-1)\times p}$.
Recall that $e_i\in\mathbb{R}^{2n-1}$ denotes the standard basis vector,
where the $i$th element is $1$ and all other elements are $0$.
Let $U=(u_1^\top,\dots,u_N^\top)^\top\in \mathbb{R}^{N\times (2n-1)}$ be the design matrix for the parameter vector $(\alpha_1, \ldots, \alpha_n,
$ $\beta_1, \ldots, \beta_{(n-1)})$, where for each  $0\le k\le n-1$ and $1\le j\le n-1$, $u_{k(n-1)+j}=e_{k+1}+e_{n+j+1}$ if $j\ge k+1$ and $u_{k(n-1)+j}=e_{k+1}+e_{n+j}$ otherwise.
Let $V=U^\top U\in\mathbb{R}^{(2n-1)\times(2n-1)}$,
with its explicit expression provided in equation (B.1) of the Supplementary Material.
Define $D=(I_{N}-U V^{-1}U^\top) \in\mathbb{R}^{N\times N}.$
By Theorem \ref{Theorem:LSE}, if $\gamma_{01}=1$ and $\beta_{0n}=0$, we obtain
\begin{align}\label{eq:identification}
Z^\top D \mathbb{E}(Y|Z)=Z^\top D (U \theta + Z\eta_0 )
= Z^\top D Z\eta_0,
\end{align}
where the second equality follows from the property $D U =0$.
Therefore, to guarantee the existence and uniqueness of $\eta_0$, we consider the following condition.

\noindent{\bf Condition (C4).} There exists some positive constant $\phi$ such that
$\phi_{\min}(Z^\top DZ/N)>\phi$ almost surely, where $\phi_{\min}(H)$ denotes the smallest eigenvalue of any matrix $H$.

\begin{figure}
\centering
\includegraphics[width=0.45\textwidth]{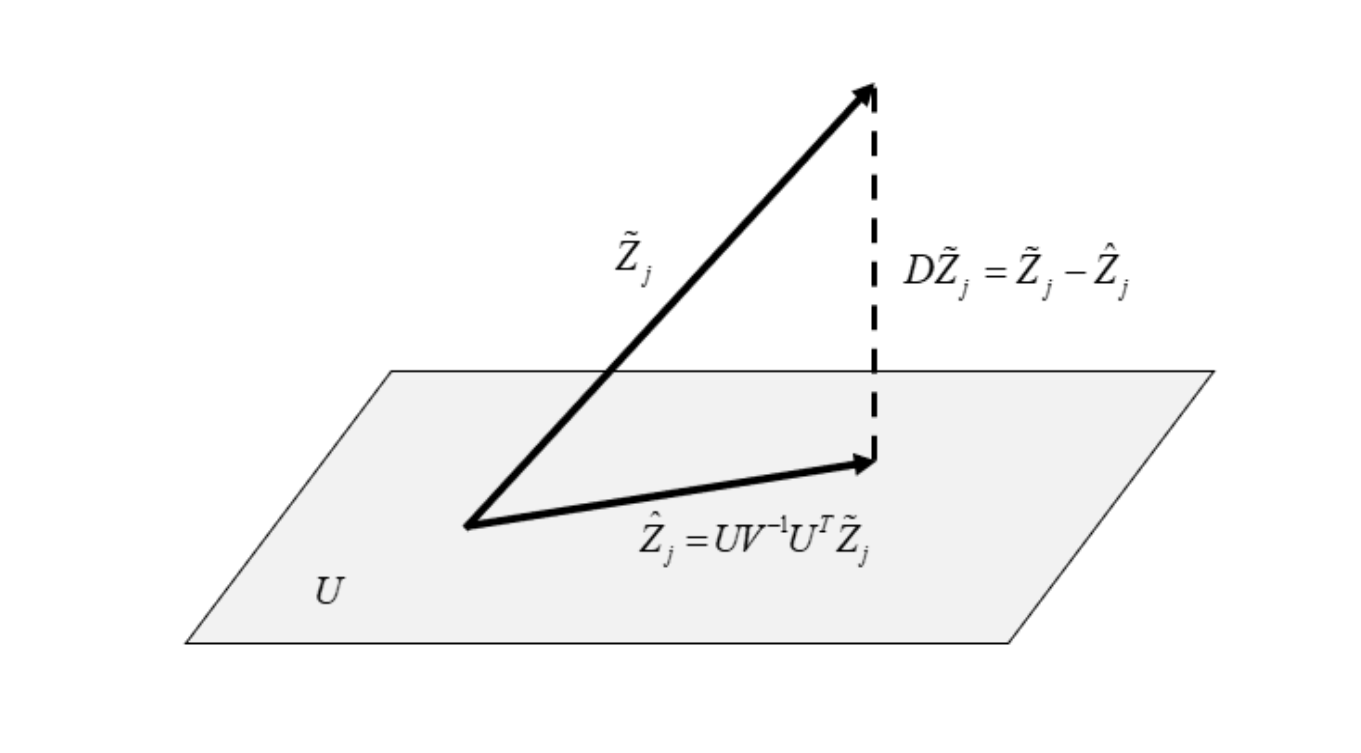}
\caption{Projection onto the linear subspace spanned by the column vectors of $U$.
Here, $\widetilde Z_j$ denotes the $j$th column vector of $Z$.}\label{Projection}
\end{figure}

Condition (C4) is essential for ensuring the existence and uniqueness of $\eta_0$.
It indicates that the effects of covariates should not be entirely captured by the degree heterogeneity.
To illustrate this intuitively, we present the projection of $Z\in \mathbb{R}^{N\times p}$
onto the linear subspace $\mathcal{U}$ spanned by the column vectors of $U$; see Figure 1.
Specifically, let $\widetilde Z_j$ denote the $j$th column of $Z$,
which can be decomposed as follows:
$$
\widetilde Z_j= \widehat Z_j+(\widetilde Z_j-\widehat Z_j).
$$
Here, $\widehat Z_j=UV^{-1}U^\top \widetilde Z_j\in \mathbb{R}^{N}$ denotes the projection of $\widetilde Z_j$
onto the linear subspace $\mathcal{U}$, while $(\widetilde Z_j-\widehat Z_j)=D\widetilde Z_j$ represents the residual vector.
If $\widetilde Z_j$ can be expressed as a linear combination of the column vectors of $U$, i.e., $\widetilde Z_j= \widehat Z_j$,
then the residual vector $D\widetilde Z_j$ is equal to zero, indicating that Condition (C4) does not hold.

Condition (C4) holds under certain straightforward assumptions.
For example, consider the following two conditions:
(i) $\widetilde Z_1,\dots, \widetilde Z_p$ are not collinear,
and (ii) they do not almost surely lie in the linear subspace $\mathcal{U}$.
Condition (i) is mild and has been widely adopted in the classical linear regression model,
while an intuitive understanding of condition (ii) has been provided above.
Under conditions (i) and (ii), we can establish that Condition (C4) holds.
The proof is provided in Section B.3 of the Supplementary Material.

Recall that $\eta_0=(\gamma_{02},\dots,\gamma_{0(p+1)})^\top$.
The following corollary states the identification of the parameters $\alpha_{0i},~\beta_{0j}$ and $\eta_0$.
The proof is provided in Section B.2 of the Supplementary material.

\begin{corollary}\label{Ident}
Under Conditions (C1)-(C4), the parameters $\alpha_{0i}, \beta_{0j}$ and $\eta_0$ are identifiable
within the subspace $\{(\alpha_{1},\dots,\alpha_{n},\beta_{1},\dots,\beta_{n},\gamma_{1},\dots,\gamma_{p+1})^\top: \beta_{n}=0~~\text{and}~~\gamma_{1}=1\}$.
\end{corollary}
}

\subsection{Estimation methods}

In this section, we develop a procedure for estimating unknown parameters.
Hereafter, we assume that $\gamma_{01}=1$ and $\beta_{0n}=0$ for the identification of model \eqref{model}.
Let $\widehat f(X_{ij1}|Z_{ij})$ be a nonparametric estimator of $f(X_{ij1}|Z_{ij})$,
which will be discussed later. Correspondingly, we define $\widehat Y_{ij}$ as
\[
\widehat Y_{ij}=\frac{A_{ij}-\mathbb{I}(X_{ij1}\ge 0)}{\widehat f(X_{ij1}|Z_{ij})},
\]
and we write $\widehat Y=(\widehat Y_1^\top,\dots,\widehat Y_n^\top)^\top\in \mathbb{R}^{N}$,
where $\widehat Y_i=(\widehat Y_{i1},\dots,\widehat Y_{i,i-1},\widehat Y_{i,i+1},\dots,\widehat Y_{in})^\top\in \mathbb{R}^{n-1}$ for $i\in[n]$.
By \eqref{eq:identification}, we can estimate $\eta_0$ using
$$
\widehat\eta=(Z^\top DZ)^{-1}Z^\top D\widehat Y.
$$
To estimate $\alpha_0$ and $\beta_0$, we employ a constrained least squares method.
Specifically, we estimate $\alpha_0$ and $\beta_0$ using
\begin{align*}
(\widehat\alpha^\top,\widehat\beta^\top)^\top=\text{arg}\min_{\alpha,\, \beta}~\mathcal{M}(\alpha,\beta,\widehat\eta)~~\text{subject~to}~\beta_n=0,
\end{align*}
where $\mathcal{M}(\alpha,\beta,\eta)=\sum_{i=1}^{n}\sum_{j\ne i}(\widehat Y_{ij}-\alpha_i-\beta_j-Z_{ij}^\top\eta)^2$.

{\color{black} We provide an intuitive insight into how the projection procedure helps eliminate the potential bias of $\widehat\eta$.
As mentioned above, the $j$th column of $Z$, denoted as $\widetilde Z_j$, can be expressed as follows:
$\widetilde Z_j = \widehat Z_j + (\widetilde Z_j - \widehat Z_j)$,
where $\widehat Z_j$ represents the projection of $\widetilde Z_j$, and $(\widetilde Z_j - \widehat Z_j)$ denotes the residual.
Therefore, the information contained in $\widetilde Z_j$ consists of two components:
the first component, $\widehat Z_j$, is characterized by degree heterogeneity,
while the second component, i.e., $(\widetilde Z_j-\widehat Z_j)$, cannot be accounted for by degree heterogeneity.
This indicates that the effects of covariates contributing to network formation are integrated into the residual.
When $\widetilde Z_j$ is directly used to estimate $\eta_0$,
bias arises from the information redundancy between $U$ and $\widehat Z_j$.
We eliminate this bias through projection.}

We now discuss the nonparametric estimator $\widehat f(x|Z_{ij})$ of $f(x|Z_{ij})$.
We divide $Z_{ij}$ into two subvectors $\widetilde{Z}_{ij1}$ and $\widetilde{Z}_{ij2}$, where
$\widetilde{Z}_{ij1}$  includes all continuous elements of $Z_{ij}$,
and $\widetilde{Z}_{ij2}$ contains the remaining discrete elements.
We estimate $f(x|Z_{ij})$ using the Nadaraya-Watson type estimator \citep{watson1964smooth, nadaraya1964estimating}:
\begin{align*}
\widehat f(x|\widetilde Z_{ij1}=z_1, \widetilde Z_{ij2}=z_2)
=\frac{\sum_{1\le i\neq j\le n}\mathcal{K}_{xz,h}(X_{ij1}-x,\widetilde Z_{ij1}-z_1)\mathbb{I}(\widetilde Z_{ij2}=z_2)}{\sum_{1\le i\neq j\le n}\mathcal{K}_{z,h}(\widetilde Z_{ij1}-z_1)\mathbb{I}(\widetilde Z_{ij2}=z_2)},
\end{align*}
where $\mathcal{K}_{xz,h}(x,z)=h^{-(p_1+1)}\mathcal{K}_{xz}(x/h, z/h)$ and $\mathcal{K}_{z,h}(x)=h^{-p_1}\mathcal{K}_{z}(z/h)$.
Here, $\mathcal{K}_{xz}(\cdot)$ and $\mathcal{K}_{z}(\cdot)$ are two kernel functions,
$h$ denotes the bandwidth parameter, and $p_1$ denotes the number of continuous covariates in $Z_{ij}.$

{\color{black}
\begin{remark}
We consider a nonparametric method to estimate the conditional density function of $X_{ij1}$ given $Z_{ij}$,
which is robust without any model assumption for $X_{ij1}$ given $Z_{ij}$.
However, this method may suffer the curse of dimensionality, especially when the dimension of $Z_{ij}$ is large.
To overcome this issue, dimension reduction techniques, such as Lasso \citep{tibshirani1996regression} and sure independent screening \citep{fan2008sure}, may be used before carrying out the proposed method.
Nevertheless, the theoretical analysis requires further investigation, and we will explore this issue in future research.
\end{remark}

\begin{remark}
Note that the residuals $(Y_{ij}-\alpha_{0i}-\beta_{0j}-Z_{ij}^\top\eta_0)$ are bounded
by a rate of at most $O(\log n)$ if $q_n=O(\log n)$,
where $q_n=\max_{1\le i\neq j\le n} |\alpha_{0i}+\beta_{0j}+Z_{ij}^\top\eta_0|$.
Therefore, for moderately large $n$, the residuals do not exhibit a heavy-tailed behavior.
In this case, a least squares method generally performs well.
This is the primary reason for our choice of the least squares method.
Further details can be found in Sections \ref{section:Properties} and \ref{section:simulations}.
\end{remark}

\begin{remark} To apply the proposed method, it is essential to identify a covariate as the special regressor and determine its sign.
We consider the following procedure to achieve this.
For a given covariate $X_{ij1}$, we partition its support into $K$ non-overlapping subintervals of equal length,
denoted by $\mathcal{I}_k=[x_k,x_{k+1})~(1\le k\le K)$.
Here, $K$ is a prespecified positive integer.
We then calculate the number of edges for which $X_{ij1}$ falls within $\mathcal{I}_k$:
$\text{Count}_k=\sum_{i=1}^n\sum_{j\neq i}A_{ij}\mathbb{I}(X_{ij1}\in\mathcal{I}_k)$.
If $\text{Count}_k$ exhibits an increasing pattern as $k$ increases,
then the effect of covariate $X_{ij1}$ on network formation is likely positive, and we can set $\gamma_{01}=1$.
Conversely, if $\text{Count}_k$  shows a decreasing pattern with increasing $k$,
then the effect of covariate $X_{ij1}$ on network formation is likely negative, and we can set $\gamma_{01}=-1$.
When there are no such increasing or decreasing trends among all covariates,
identifying a covariate as the special regressor becomes challenging. We do not study this case here and
intend to investigate it in future research.
\end{remark}}

\section{Theoretical Results}
\label{section:Properties}

In this section, we present consistency and asymptotic normality of the estimators.
To achieve this, additional conditions are required.

\noindent{\bf Condition (C5).} $\max_{i,j}\|Z_{ij}\|_\infty\le \kappa$ almost surely,
where $\kappa$ is allowed to diverge with $n$.
Here, the subscript $n$ in $\kappa$ is omitted for simplicity.

\noindent{\bf Condition (C6).}
There exists a constant $m$ such that $f(x|Z_{ij})>m>0$ on the support of $X_{ij1}$.
In addition, the $r$th order partial derivative of
the probability density function $f_Z(z)$ of $Z_{ij}$ with respect to
continuous components of $Z_{ij}$ exists and is continuous and bounded.
The $r$th order partial derivative of the joint density function $f_{XZ}(x,z)$ of $(X_{ij1}, Z_{ij})$
with respect to continuous entries of $(X_{ij1}, Z_{ij})$ also exists and is continuous and bounded.
Here, $m$ is allowed to decrease to zero as $n\rightarrow \infty $, and the subscript $n$ in $m$ is suppressed.

\noindent{\bf Condition (C7).} The kernel function $\mathcal{K}_{z}(z)$ is a symmetric and piecewise Lipschitz continuous kernel of order $r$.
That is, $\int\dots\int \mathcal{K}_z(z_1,\dots,z_{p_1})dz_1\cdots dz_{p_1}=1,$
\begin{align*}
&\int\dots\int z_1^{j_1}\dots z_{p_1}^{j_{p_1}}\mathcal{K}_z(z_1,\dots,z_{p_1})dz_1\cdots dz_{p_1}=0~~(0<j_1+\dots+j_{p_1}<r),\\
&\int\dots\int z_1^{j_1}\dots z_{p_1}^{j_{p_1}}\mathcal{K}_z(z_1,\dots,z_{p_1})dz_1\cdots dz_{p_1}\neq 0~~(0<j_1+\dots+j_{p_1}=r).
\end{align*}
In addition, it is a bounded differentiable function with absolutely integrable Fourier transforms.
All of the conditions also hold for $\mathcal{K}_{xz}(x,z)$ by replacing $z$ with $(x,z)$.

Condition (C5) assumes the boundedness of $Z_{ij}$,
which is required to simplify the proof of the following theorems.
However, this assumption can be relaxed to allow for sub-Gaussian variables.
The first part of Condition (C6), together with Theorem \ref{Theorem:LSE}, implies that
$\max_{1\le i\neq j\le n}|\alpha_{0i}+\beta_{0j}+Z_{ij}^\top\eta_0|<2/m$ holds almost surely.
Note that individual-specific parameters $\alpha_{0i}$ and $\beta_{0i}$ can be utilized
to determine the level of sparsity in a network (e.g., \citealp{yan2019}).
Recall $q_n=\max_{1\le i\neq j\le n}|\alpha_{0i}+\beta_{0j}+Z_{ij}^\top\eta_0|$.
If $q_n$ is of the order $\log(n)$,
then $m$ must be at least of the order $1/\log(n)$.
In addition, this requires that the support $[B_L, B_U]$ of the special regressor $X_{ij1}$
satisfies $B_L=O(\log n)$ and $B_U=O(\log n)$.
The second part of Condition (C6) is mild and similar conditions have been used in
\cite{andrews_1995}, \cite{honore2002}, \cite{ARADILLASLOPEZ2012120} and \cite{candelaria2020semiparametric}.
Condition (C7) requires the application of a higher-order kernel,
which is widely adopted across various contexts; see \cite{andrews_1995}, \cite{lewbel1998},
\cite{honore2002}, \cite{Qi2005} and \cite{candelaria2020semiparametric}.

In the following, we redefine $\beta=(\beta_{1},\dots,\beta_{n-1})^\top$, excluding element $\beta_{n}$.
Define $\theta_0$ as the true value of $\theta$,
and let $\widehat\theta=(\widehat\alpha^\top,\widehat\beta^\top)^\top$ represent the estimator of $\theta_0$.
The consistency of $\widehat\theta$ and $\widehat\eta$ is presented below.

\begin{theorem}\label{theorem:consistency}
Suppose that Conditions (C1)-(C7) hold.
If 
\begin{align}\label{Theorem:consistency:eq2}
\frac{(\kappa+q_n)^2}{\phi m^2}\bigg[\sqrt{\frac{\log^5(n)}{n}}+\sqrt{\frac{\log^5(n)}{n^2h^{2p_1+2}}}+h^{r}\bigg]=o(1),
\end{align}
then we have
\begin{align*}
\|\widehat{\theta}-\theta_0\|_\infty=o_p(1)~~~\text{and}~~~\|\widehat{\eta}-\eta_0\|_\infty=o_p(1).
\end{align*}
\end{theorem}

Condition \eqref{Theorem:consistency:eq2} places restrictions on the rate at which $\kappa$ increases
and on the rates at which $m$ and $\phi$ decrease.
Here, $\phi$, $\kappa$ and $m$ are defined in Conditions (C4)-(C6), respectively.
Moreover, the terms involving $h$ in \eqref{Theorem:consistency:eq2} are included
to balance the bias and variance of $\widehat Y_{ij}$ when using the kernel smoothing method.
This indicates that the bandwidth $h$ significantly affects the behavior of the estimator both theoretically and practically.
When $\kappa$, $\phi$ and $m$ are treated as constants,
it is necessary for $h\rightarrow 0$ and $n^2h^{2p_1+2}/\log(n)\rightarrow \infty$ as $n\rightarrow \infty$
to ensure the consistency of the estimators.

Next, we present a high-dimensional central limit theorem for the estimators $\widehat \alpha$ and $\widehat\beta$.
Specifically, we consider the inferences related to $\mathcal{L}_{1\alpha}\alpha_0$, $\mathcal{L}_{1\beta}\beta_0$ and $\mathcal{L}_2\eta_0$,
where $\mathcal{L}_{1\alpha}$, $\mathcal{L}_{1\beta}$ and $\mathcal{L}_2$ denote
$M_{1\alpha}\times n$, $M_{1\beta}\times (n-1)$ and $M_2\times p$ matrices, respectively.
This framework can be utilized to construct confidence intervals for the linear combinations of the parameters $\alpha_0$, $\beta_0$ and $\eta_0$.
{\color{black}
To illustrate this, we provide two examples for $\mathcal{L}_{1\alpha}$.
The matrices $\mathcal{L}_{1\beta}$ and $\mathcal{L}_2$ can be defined analogously.
The first example is given by $\mathcal{L}_{1\alpha}=\widetilde e_i,$
where $\widetilde e_i$ is an $n$-dimensional row vector
with its $i$th element equal to 1 and all other elements equal to 0.
In this case, $\mathcal{L}_{1\alpha}(\widehat\alpha-\alpha_0)=(\widehat\alpha_i-\alpha_{0i})$,
which allows us to construct a confidence interval for $\alpha_{0i}$.
The second example is $\mathcal{L}_{1\alpha}=I_{n\times n}$,
where $I_{n\times n}$ denotes the $n\times n$ identity matrix.
In this case, we have $\mathcal{L}_{1\alpha}(\widehat\alpha-\alpha_0)=(\widehat\alpha_1-\alpha_{01}, \ldots,
\widehat\alpha_n-\alpha_{0n})^\top$.
Therefore, we can perform simultaneous inference for all out-degree parameters.
}

To obtain asymptotic distributions of the estimators, we need an additional condition.

\noindent{\bf Condition (C8)}.
(i) 
$0<M_{1\alpha}\le n$ and
$0<M_{1\beta}\le n $.
(ii) There exist some constants $s_{U\alpha}$, $s_{U\beta}$, $L_{U\alpha}$ and $L_{U\beta}$ (independent of $n$) such that
$1\le \max_{k\in [M_{1\alpha}]}|S_{k\alpha}|\le s_{U\alpha},$
$1\le \max_{k\in [M_{1\beta}]}|S_{k\beta}|\le s_{U\beta},$
$\|\mathcal{L}_{1\alpha}\|_{\max}\le L_{U\alpha}$
and
$\|\mathcal{L}_{1\beta}\|_{\max}\le L_{U\beta}$,
where  $S_{k\alpha}=\{j: L_{1,kj}^\alpha\neq 0 \}$,
$S_{k\beta}=\{j: L_{1,kj}^\beta\neq 0 \}.$
Here, $L_{1,kj}^\alpha$ and $L_{1,kj}^\beta$ denote the $(k,j)$th elements of $\mathcal{L}_{1\alpha}$ and $\mathcal{L}_{1\beta},$ respectively.

The first part of Condition (C8) states that the row dimensions of the matrices $\mathcal{L}_{1\alpha}$ and $\mathcal{L}_{1\beta}$
must satisfy the conditions $M_{1\alpha}\le n$ and $M_{1\beta}\le n-1$, respectively, which can increase as $n$  increases.
This implies that we are operating in high-dimensional settings.
The second part assumes the sparsity and boundedness of the matrices $\mathcal{L}_{1\alpha}$ and $\mathcal{L}_{1\beta},$
which can be satisfied by cases such as $\mathcal{L}_{1\alpha}=\widetilde e_i$
and $\mathcal{L}_{1\alpha}=\widetilde e_{i}-\widetilde e_j$.

For $i\neq j\in[n]$, define $\epsilon_{ij}=Y_{ij}-\alpha_{0i}-\beta_{0j}-Z_{ij}^\top\eta_0.$
The high-dimensional central limit theorems for $\widehat{\alpha}$ and $\widehat\beta$ are presented below.

\begin{theorem}\label{Theorem:GA:theta}
Suppose that Conditions (C1)-(C8) hold,
and there exist some constants $0<\sigma_{\epsilon L}^2<\sigma_{\epsilon U}^2<\infty$
such that $\sigma_{\epsilon L}^2<\mathbb{E}(\epsilon_{ij}^2)=\sigma_\epsilon^2<\sigma_{\epsilon U}^2$ for all $i \neq j\in [n]$.
If
\begin{align}\label{Theorem:GA:theta:condition}
\frac{(\kappa+q_n)^2}{\phi m^2}\bigg[\sqrt{\frac{\log^5(n)}{nh^{2p_1+2}}}+\sqrt{n\log(n)}h^{r}\bigg]=o(1),
\end{align}
then we have
\begin{align*}
&\sup_{x\in\mathbb{R}^{M_{1\alpha}}}\big|\mathbb{P}\big(\sqrt{n-1}\mathcal{L}_{1\alpha}(\widehat\alpha-\alpha_0)\le x\big)-\mathbb{P}\big(\mathcal{L}_1^\alpha G_1\le x\big)\big|=o(1)\\
\text{and}~~~~&\sup_{x\in\mathbb{R}^{M_{1\beta}}}\big|\mathbb{P}\big(\sqrt{n-1}\mathcal{L}_{1\beta}(\widehat\beta-\beta_0)\le x\big)-\mathbb{P}\big(\mathcal{L}_1^\beta G_1\le x\big)\big|=o(1),
\end{align*}
where $G_{1}\sim \mathcal{N}_{2n-1}(0,(n-1)\sigma_\epsilon^2V^{-1})$, $\mathcal{L}_{1}^\alpha=(\mathcal{L}_{1\alpha},0_{M_{1\alpha}\times (n-1)})$
and $\mathcal{L}_{1}^\beta=(0_{M_{1\beta}\times n},\mathcal{L}_{1\beta}).$
Here, $0_{b_1\times b_2}$ denotes a $b_1\times b_2$ zero matrix.
\end{theorem}

When $\kappa$, $q_n$, $\phi$ and $m$ are constants,
\eqref{Theorem:GA:theta:condition} indicates that
the bandwidth $h$ must satisfy the conditions $nh^{2(p_1+1)}/\log^2(n)\rightarrow\infty$ and $n\log(n)h^{2r}\rightarrow 0$ as $n\rightarrow \infty$.
To fulfill this bandwidth condition, we can set $h = O(n^{-1/d})$ for some
integer $d>2(p_1+1)$, and choose $r$ to be the smallest even integer such that $r \ge d-(p_1+1)$.
For example, when $p_1=2$,  we can set $d=7$ and $r=4$.
In practice, the bandwidth should be carefully chosen to balance the trade-off between the bias and variance of  $\widehat Y_{ij}$.
To enhance the practicality of the proposed method,
we develop a data-driven procedure for selecting the bandwidth $h$ in Section \ref{section:simulations}.

Define $Q=(Q_1^\top,\dots,Q_{n}^\top)^\top,$
where $Q_i=(Q_{i1},\dots,Q_{i,i-1},Q_{i,i+1},\dots,Q_{in})^\top$ and
$Q_{ij}=Y_{ij}-\mathbb{E}(Y_{ij}|X_{ij1}, Z_{ij}).$

\begin{theorem}\label{Theorem:GA:eta}
Suppose that Conditions (C1)-(C8) hold, and
there exist some constants $0<\sigma_{QL}^2<\sigma_{QU}^2<\infty$ such that $\sigma_{QL}^2<\mathbb{E}(Q_{ij}^2)=\sigma_Q^2<\sigma_{QU}^2$
for all $i\neq j\in [n].$
If
\begin{align}\label{Theorem:GA:eta:condition}
 \frac{(\kappa+q_n)^2}{\phi m^2}\Bigg[\sqrt{\frac{\log^5(n)}{n^2h^{2p_1+2}}}+\sqrt{n^2\log n}h^{r}\Bigg]=o(1),
\end{align}
then we have
\begin{align*}
&\sup_{x\in\mathbb{R}^{M_2}}\big|\mathbb{P}\big(\sqrt{N}\mathcal{L}_2(\widehat\eta-\eta_0)\le x\big)-\mathbb{P}\big(\mathcal{L}_2 G_2\le x\big)\big|=o(1),
\end{align*}
where $G_2\sim \mathcal{N}_p(0,\sigma_Q^2[\mathbb{E}(Z^\top DZ/N)]^{-1}).$
\end{theorem}

\begin{remark}
Theorem \ref{Theorem:GA:theta} addresses a high-dimensional central limit theorem,
while the asymptotic distribution presented in \cite{yan2019}
is developed specifically for a fixed-dimensional subvector of $\widehat\theta$.
The asymptotic distributions of the estimators for the degree parameters have not been established in \cite{graham2017econometric}.
Furthermore, the central limit theorem for homophily parameters in both \cite{graham2017econometric} and \cite{yan2019} exhibits asymptotic bias due to the incidental parameter problem in likelihood inference \citep{NS1948, FW2016}.
However, $\widehat \gamma$ in current study is unbiased, owing to the projection technique.
\end{remark}

The variance of $G_1$ in Theorem \ref{Theorem:GA:theta} involves the inverse of $V=U^\top U$.
An explicit expression for this is provided in (B.2) of the Supplementary Material.
In addition, the variance of $G_1$ also incorporates the variance $\sigma^2_\epsilon$ of $\epsilon_{ij}$,
which is unknown but can be estimated using $\widehat{\sigma}_\epsilon^2=N^{-1}\sum_{i=1}^n\sum_{j\neq i}^n\widehat \epsilon_{ij}^2$.
Here, $\widehat \epsilon_{ij}=\widehat Y_{ij}-\widehat\alpha_i-\widehat\beta_j-Z_{ij}^\top\widehat\eta.$

We now estimate the unknown parameter $\sigma_Q^2$ in the covariance matrix of $G_2$.
We consider the following estimator:
$\widehat{\sigma}_Q^2=N^{-1}\sum_{i=1}^n\sum_{j\neq i}^n\widehat Q_{ij}^2,$
where $\widehat Q_{ij}=\widehat Y_{ij}-\widehat{\mathbb{E}}(Y_{ij}|X_{ij1},Z_{ij})$.
Here, $\widehat{\mathbb{E}}(Y_{ij}|X_{ij1},Z_{ij})$ is a nonparametric estimator of $\mathbb{E}(Y_{ij}|X_{ij1},Z_{ij})$.
We adopt the Nadaraya-Watson type estimator, which is defined as follows:
\begin{align*}
\widehat{\mathbb{E}}(Y_{ij}|X_{ij1}=x,\widetilde Z_{ij1}=z_1, \widetilde Z_{ij2}=z_2)
=\frac{\sum_{1\le i\neq j\le n}\widehat Y_{ij}\mathcal{K}_{xz,h}(X_{ij1}-x,\widetilde Z_{ij1}-z_1)\mathbb{I}(\widetilde Z_{ij2}=z_2)}{\sum_{1\le i\neq j\le n}\mathcal{K}_{xz,h}(X_{ij1}-x,\widetilde Z_{ij1}-z_1)\mathbb{I}(\widetilde Z_{ij2}=z_2)}.
\end{align*}
Here, $\widetilde Z_{ij1}$ and $\widetilde Z_{ij2}$ are defined in Section \ref{section2-identification-estimation}.

The consistency of the Gaussian approximation
when replacing $\mathbb{E}(Y_{ij}|X_{ij1}, Z_{ij})$ with $\widehat{\mathbb{E}}(Y_{ij}|X_{ij1}, Z_{ij})$
is established in the following theorem.

\begin{theorem}\label{Theorem:hatGA:theta}
If the conditions in Theorem \ref{Theorem:GA:theta} hold, then we have
\begin{align*}
&\sup_{x\in\mathbb{R}^{M_{1\alpha}}}\big|\mathbb{P}\big(\sqrt{n-1}\mathcal{L}_{1\alpha}(\widehat\alpha-\alpha_0)\le x\big)-\mathbb{P}\big(\mathcal{L}_1^\alpha \widehat G_1\le x\big)\big|=o(1)\\
\text{and}~~~&\sup_{x\in\mathbb{R}^{M_{1\beta}}}\big|\mathbb{P}\big(\sqrt{n-1}\mathcal{L}_{1\beta}(\widehat\beta-\beta_0)\le x\big)-\mathbb{P}\big(\mathcal{L}_1^\beta \widehat G_1\le x\big)\big|=o(1),
\end{align*}
where $\widehat G_1\sim \mathcal{N}_{2n-1}(0,(n-1)\widehat\sigma_{\epsilon}^2V^{-1}).$
In addition, if the conditions in Theorem \ref{Theorem:GA:eta} hold, then we have
\begin{align*}
&\sup_{x\in\mathbb{R}^{M_2}}\big|\mathbb{P}\big(\sqrt{N}\mathcal{L}_2(\widehat\eta-\eta_0)\le x\big)-\mathbb{P}\big(\mathcal{L}_2 \widehat G_2\le x\big)\big|=o(1),
\end{align*}
where $\widehat G_2\sim \mathcal{N}_p(0,\widehat\sigma_{Q}^2(Z^\top DZ/N)^{-1}).$
\end{theorem}

\begin{remark}
The above theorem can be used to construct confidence intervals for the parameters $\alpha_{0i},~\beta_{0i}$ and $\eta_0$.
For example, if we set $\mathcal{L}_{1\alpha}=\widetilde e_i~(i\in[n])$
and $M_{1\alpha}=1$, then we have $\sup_{x\in\mathbb{R}}\big|\mathbb{P}\big(\sqrt{n-1} (\widehat\alpha_i-\alpha_{0i})\le x\big)
-\mathbb{P}\big(\widehat G_{1i}\le x\big)\big|\rightarrow 0$ by Theorem \ref{Theorem:hatGA:theta},
where $\widehat G_{1i}$ represents the $i$th element of $\widehat G_{1}.$
Let $z_{i,1-\nu/2}$ denote the upper $(\nu/2)$-quantile of the distribution $\widehat G_{1i},$
We can then construct the pointwise $(1-\nu)$ confidence interval for each $\alpha_{0i}$ as
$\big[\widehat\alpha_i-z_{i,1-\nu/2}/\sqrt{n-1},\ \widehat\alpha_i+z_{i,1-\nu/2}/\sqrt{n-1}\big].$
In addition, if we are interested in constructing confidence intervals for the difference $(\alpha_{0i}-\alpha_{0j})$ for any pair $(i,j),$
we can set $\mathcal{L}_{1\alpha}=\widetilde e_{ij}~(1\le i\neq j\le n)$ and  $M_{1\alpha}=1.$
Theorem \ref{Theorem:hatGA:theta} implies that
$\sup_{x\in\mathbb{R}}\big|\mathbb{P}\big(\sqrt{n-1}[\widehat\alpha_i-\widehat\alpha_j-(\alpha_{0i}-\alpha_{0j})]\le x\big)
-\mathbb{P}\big(\widehat G_{1i}-\widehat G_{1j}\le x\big)\big|\rightarrow 0.$
Let $z_{ij,1-\nu/2}$ be the upper $(\nu/2)$-quantile of the distribution $(\widehat G_{1i}-\widehat G_{1j}).$
Then, the pairwise confidence interval for $(\alpha_{0i}-\alpha_{0j})$ is defined as
$\big[\widehat\alpha_i-\widehat\alpha_j-z_{ij,1-\nu/2}/\sqrt{n-1},\ \widehat\alpha_i-\widehat\alpha_j+z_{ij,1-\nu/2}/\sqrt{n-1}\big].$
Similarly, we can obtain the confidence intervals for $\beta_{0i}$ and for the difference $(\beta_{0i}-\beta_{0j})$
by setting $\mathcal{L}_{1\beta}=\widetilde e_{n+i}$ and $\mathcal{L}_{1\beta}=\widetilde e_{n+i,n+j},$ respectively.
\end{remark}

\section{Applications}
\label{section:application}

This section presents several concrete applications of Theorem \ref{Theorem:hatGA:theta}.
Specifically, we consider the following applications:
(i) testing for sparse signals;
(ii) support recovery and (iii) testing the existence of the degree heterogeneity.

\subsection{Testing for sparse signals}\label{section:application:TS}
In this section, we focus on testing the following hypotheses:
\begin{align*}
& H_{0\alpha,S}: \alpha_{0i}=0~~ \text{for~all}~~ i\in[n]~~~~~~~
\text{ versus }~~ H_{1\alpha,S}: \alpha_{0i}\neq 0~~\text{for some}\ i\in[n];\\
& H_{0\beta,S}: \beta_{0i}=0~~ \text{for~all}~~  i\in[n-1]~~
\text{ versus }~~ H_{1\beta,S}: \beta_{0i}\neq 0~~\text{for some}\ i\in[n-1].
\end{align*}
Let $\widehat\zeta_{i}$ denote the $i$th diagonal element of $(n-1)\widehat\sigma_{\epsilon}^2V^{-1}.$
To test $H_{0\alpha,S},$  the matrix $\mathcal{L}_{1\alpha}$ is given by $\mathcal{L}_{1\alpha,S}=\text{diag}\{1/\widehat\zeta_{1}^{1/2},\dots,1/\widehat\zeta_{n}^{1/2}\}$ with $M_{1\alpha}=n.$
Under $H_{0\alpha},$ we have $\mathcal{L}_{1\alpha,S}(\widehat\alpha-\alpha_0)
=(\widehat \alpha_{1}/\widehat\zeta_{1}^{1/2},\dots,\widehat\alpha_{n}/\widehat\zeta_{n}^{1/2})^\top$.
To test $H_{0\beta,S},$ we define $\mathcal{L}_{1\beta}$ as
$\mathcal{L}_{1\beta,S}=\text{diag}\{1/\widehat\zeta_{n+1}^{1/2},\dots,1/\widehat\zeta_{2n-1}^{1/2}\}$ with $M_{1\beta}=n-1$.
Under $H_{0\beta,S}$, we obtain $\mathcal{L}_{1\beta}(\widehat\beta-\beta_0)
=(\widehat \beta_{1}/\widehat\zeta_{n+1}^{1/2},\dots,\widehat\beta_{n-1}/\widehat\zeta_{2n-1}^{1/2})^\top$.

We consider the following test statistics for $H_{0\alpha,S}$ and $H_{0\beta,S},$ respectively:
\begin{align*}
\mathcal{T}_{\alpha,S}=\max_{1\le i\le n}
\sqrt{n-1}|\widehat \alpha_{i}|/\widehat\zeta_{i}^{1/2}~~~\text{and}~~~
\mathcal{T}_{\beta,S}=\max_{1\le i\le n-1}\sqrt{n-1}|\widehat \beta_{i}|/\widehat\zeta_{n+i}^{1/2}.
\end{align*}
The test statistics $\mathcal{T}_{\alpha,S}$ and $\mathcal{T}_{\beta,S}$ are close to zero under the null hypotheses $H_{0\alpha}$ and $H_{0\beta,S}$.
Consequently, we reject $H_{0\alpha,S}$ if $\mathcal{T}_\alpha>c_{\alpha,S}(\nu),$
and reject $H_{0\beta}$ if $\mathcal{T}_\beta>c_{\beta,S}(\nu),$
where $c_{\alpha,S}(\nu)$ and $c_{\beta,S}(\nu)$ are the critical values.
Based on Theorem \ref{Theorem:hatGA:theta}, we employ a resampling method to obtain these critical values.
Specifically, we repeatedly generate normal random samples from $\widehat G_1$
and derive $c_{\alpha,S}(\nu)$ and $c_{\beta,S}(\nu)$ using the empirical distributions of
$\|\mathcal{L}_{1S}^\alpha\widehat G_1\|_{\infty}$ and $\|\mathcal{L}_{1S}^\beta\widehat G_1\|_{\infty},$
where $\mathcal{L}_{1S}^\alpha=(\mathcal{L}_{1\alpha,S},0_{M_{1\alpha}\times (n-1)})$
and $\mathcal{L}_{1S}^\beta=(0_{M_{1\beta}\times n},\mathcal{L}_{1\beta,S}).$
By Theorem \ref{Theorem:hatGA:theta}, we conclude that the proposed test is asymptotically of level $\nu$.

\begin{sloppypar}
We now examine the asymptotic power analysis of the aforementioned procedure.
To do this, we define the separation sets as follows:
$\mathcal{U}_{\alpha,S}(c)=\{\alpha=(\alpha_1,\dots,\alpha_{n})^\top: \ \max_{i\in[n]} |\alpha_{0i}|/\zeta_{i}^{1/2}>c\sqrt{\log(n)/(n-1)}\}$
and $\mathcal{U}_{\beta,S}(c)=\{\beta=(\beta_1,\dots,\beta_{n-1},0)^\top:\ \max_{i\in[n-1]} |\beta_{0i}|/\zeta_{n+i}^{1/2}>c\sqrt{\log(n-1)/(n-1)}\},$
where $\zeta_{i}$ is the $i$th element of $(n-1)\sigma_{\epsilon}^2V^{-1}.$
\end{sloppypar}

\begin{proposition}\label{proposition:power}
  Under the conditions in Theorem \ref{Theorem:hatGA:theta},  we have for any $\lambda_0>0$,
\begin{align*}
\inf_{\alpha_0\in \mathcal{U}_{\alpha,S}(\sqrt{2}+\lambda_0)}\mathbb{P}\big(\max_{i\in[n]} \sqrt{n-1}|\widehat\alpha_i|/\widehat\zeta_{i}^{1/2}>c_{\alpha,S}(\nu) \big)&\rightarrow 1,\\
\text{and}~~\inf_{\beta_0\in \mathcal{U}_{\beta,S}(\sqrt{2}+\lambda_0)}\mathbb{P}\big(\max_{i\in[n-1]} \sqrt{n-1}|\widehat\beta_i|/\widehat\zeta_{n+i}^{1/2}>c_{\beta,S}(\nu) \big)&\rightarrow 1
~~\text{as}~~n\rightarrow \infty.
\end{align*}
\end{proposition}
Proposition \ref{proposition:power} indicates that the proposed test can be triggered
even when only a single entry of $\alpha$ exceeds the threshold of $(\sqrt{2}+\lambda_0)\sqrt{\log(n)/(n-1)}.$
Consequently, the proposed test demonstrates sensitivity in detecting sparse alternatives.

\subsection{Support recovery}\label{section:application:SR}

Denote $\mathcal{S}_{0\alpha}=\{j\in[n]: \alpha_{0j}\neq 0\}$ and $\mathcal{S}_{0\beta}=\{j\in[n-1]: \beta_{0j}\neq 0\}$ as
the supports of $\alpha_0$ and $\beta_0,$ respectively.
Let $s_{0\alpha}=|\mathcal{S}_{0\alpha}|$ and $s_{0\beta}=|\mathcal{S}_{0\beta}|.$
When the null hypothesis $H_{0\alpha,S}$ and $H_{0\beta,S}$ are rejected,
recovering the supports $\mathcal{S}_{0\alpha}$ and $\mathcal{S}_{0\beta}$ is of great interest in practice.
Our support recovery procedure employs a proper threshold $t$ defined in the sets:
\begin{align*}
 \widehat{\mathcal{S}}_{\alpha}(t)=&\bigg\{i\in[n]: |\widehat\alpha_i|>\sqrt{t\widehat\zeta_{i}\log(n)/(n-1)}\bigg\}\\
\text{and}~~\widehat{\mathcal{S}}_{\beta}(t)=&\bigg\{i\in [n-1]: |\widehat\beta_i|>\sqrt{t\widehat\zeta_{n+i}\log(n-1)/(n-1)}\bigg\}.
\end{align*}
This strategy is widely utilized for recovering sparse signals in high-dimensional data contexts,
and it is also referenced in \cite{shao} for the analysis of undirected network data.
In Proposition \ref{proposition:SR}, we demonstrate that the aforementioned support recovery procedure is consistent
when the threshold value is set to $t = 2.$
To this end, we define ${\mathcal{G}}_{\alpha}(t_0)=\{\alpha: \min_{i\in \mathcal{S}_{0\alpha}} |\alpha_i|/\zeta_{i}^{1/2}>t_0\sqrt{\log (n)/(n-1)}\}$
and $\mathcal{G}_{\beta}(t_0)=\{\beta: \min_{i\in \mathcal{S}_{0\beta}} |\beta_i|/\zeta_{n+i}^{1/2}>t_0\sqrt{\log (n-1)/(n-1)}\}.$
Let $\mathcal{G}_{\alpha}^*=\{\alpha: \sum_{i=1}^{n}\mathbb{I}\{\alpha_i\neq 0\}=s_{0\alpha}\}$
and $\mathcal{G}_{\beta}^*=\{\beta: \sum_{i=1}^{n-1}\mathbb{I}\{\beta_i\neq 0\}=s_{0\beta}\}$
denote the classes of $s_{0\alpha}$-sparse and $s_{0\beta}$-sparse vectors, respectively.

\begin{proposition}\label{proposition:SR}
Under the conditions in Theorem \ref{Theorem:hatGA:theta}, we have
\begin{align*}
\inf_{\alpha\in \mathcal{G}_{\alpha}(2\sqrt{2})}\mathbb{P}\big(\widehat{\mathcal{S}}_{\alpha}(2)=\mathcal{S}_{0\alpha}\big)\rightarrow 1~~~
\text{and}~~~\inf_{\beta\in \mathcal{G}_{\beta}(2\sqrt{2})}\mathbb{P}\big(\widehat{\mathcal{S}}_{\beta}(2)=\mathcal{S}_{0\beta}\big)\rightarrow 1
~~\text{as}~~n\rightarrow\infty.
\end{align*}
Moreover, if $s_{0\alpha}=o(n)$ and $s_{0\beta}=o(n),$
then for any $0<t<2,$
\begin{align*}
\sup_{\alpha\in \mathcal{G}_\alpha^*}\mathbb{P}\big(\widehat{\mathcal{S}}_{\alpha}(t)=\mathcal{S}_{0\alpha}\big)\rightarrow 0~~~
\text{and}~~~\sup_{\beta\in \mathcal{G}_\beta^*}\mathbb{P}\big(\widehat{\mathcal{S}}_{\beta}(t)=\mathcal{S}_{0\beta}\big)\rightarrow 0
~~\text{as}~~n\rightarrow\infty.
\end{align*}
\end{proposition}

The first part of Proposition \ref{proposition:SR} shows that
with a probability approaching 1, $\widehat{\mathcal{S}}_{\alpha}(2)$ and $\widehat{\mathcal{S}}_{\beta}(2)$
can accurately recover the supports $\mathcal{S}_{0\alpha}$
and ${\mathcal{S}}_{0\beta}$ uniformly over the collections $\mathcal{G}_\alpha(2\sqrt{2})$ and $\mathcal{G}_\beta(2\sqrt{2}),$ respectively.
The second part of Proposition \ref{proposition:SR} states the optimality of the threshold parameter $t=2$.
In other words, if the threshold parameter $t$ is set below $2,$ some zero entries of $\alpha_0$ and $\beta_0$
will be retained in $\widehat{\mathcal{S}}_{\alpha}(t)$ and $\widehat{\mathcal{S}}_{\beta}(t),$ leading to an overfitting phenomenon.

\subsection{Testing for degree heterogeneity}\label{section:application:TDH}

Degree heterogeneity is an important feature of real-world networks.
However, it is not always clear whether a subnetwork exhibits degree heterogeneity or not.
Therefore, we consider the following hypotheses:
\begin{align*}
& H_{0\alpha,D}: \alpha_{0i}=\alpha_{0j}~~\text{for~all}~i,j\in\mathcal{G}_1~\;
\text{ versus }~~H_{1\alpha,D}: \alpha_{0i}\neq \alpha_{0,j}~\;\text{for some}\ i,j\in \mathcal{G}_1;\\
& H_{0\beta,D}: \beta_{0i}=\beta_{0,j}~~\text{for~all}~i,j\in\mathcal{G}_2~
\text{ versus }~~H_{1\beta,D}: \beta_{0i}\neq \beta_{0,i+1}~~\text{for some}\ i\in \mathcal{G}_2,
\end{align*}
where $\mathcal{G}_1$ and $\mathcal{G}_2$ are prespecified subsets of $[n]$ and $[n-1]$, respectively.
Under $H_{0\alpha,D}$, the observed subnetwork does not exhibit out-degree heterogeneity,
while  under $H_{0\beta,D}$, it does not display in-degree heterogeneity.
To test $H_{0\alpha,D}$ and $H_{0\beta,D},$ we apply the following test statistics:
\begin{align*}
 \widetilde{\mathcal{T}}_{\alpha,D}=\max_{i<j\in\mathcal{G}_1} \sqrt{n-1}|\widehat\alpha_i-\widehat\alpha_{j}|/\widehat\zeta_{ij,D}^{1/2}~~~\text{and}~~~
 \widetilde{\mathcal{T}}_{\beta,D}=\max_{i<j\in\mathcal{G}_2} \sqrt{n-1}|\widehat\beta_i-\widehat\beta_{j}|/\widehat\zeta_{n+i,n+j,D}^{1/2},
\end{align*}
where $\widehat\zeta_{ij,D}=e_{ij} [(n-1)\widehat{\sigma_{\epsilon}}^2V^{-1}]e_{ij}^\top$ and $e_{ij}=e_i-e_j.$
However, the tests $\widetilde{\mathcal{T}}_{\alpha,D}$ and $\widetilde{\mathcal{T}}_{\beta,D}$ are computationally intensive.
To overcome this issue, we propose a simpler test for the null hypotheses.
Specifically, define $\mathcal{L}_{1\alpha}=\mathcal{L}_{1\alpha,D}$ and $\mathcal{L}_{1\beta}=\mathcal{L}_{1\beta,D},$
where the $i$th row of $\mathcal{L}_{1\alpha,D}$ for $i\in[n-1]$ and $\mathcal{L}_{1\beta,D}$ for $i\in[n-2]$ is $\widetilde e_i-\widetilde e_{i+1}.$
Then, we have  $\mathcal{L}_{1\alpha,D}(\widehat\alpha-\alpha_0)
=(\widehat \alpha_{1}-\widehat\alpha_{2},\dots,\widehat \alpha_{n-1}-\widehat\alpha_{n})^\top$ under $H_{0\alpha,D},$
and $\mathcal{L}_{1\beta,D}(\widehat\beta-\beta_0)=(\widehat \beta_{1}-\widehat\beta_{2},\dots,\widehat \beta_{n-2}-\widehat\beta_{n-1})^\top$
under $H_{0\beta,D}.$ Furthermore, we define
\begin{align*}
\mathcal{T}_{\alpha,D}=\max_{i\in \mathcal{G}_1}
\sqrt{n-1}|\widehat \alpha_{i}-\widehat \alpha_{i+1}|/\widehat\zeta_{i,D}^{1/2}~~~\text{and}~~~
\mathcal{T}_{\beta,D}=\max_{i\in \mathcal{G}_2}\sqrt{n-1}|\widehat \beta_{i}-\widehat \beta_{i+1}|/\widehat\zeta_{n+i,D}^{1/2},
\end{align*}
where $\widehat\zeta_{i,D}=e_{i,i+1} [(n-1)\widehat{\sigma}^2V^{-1}]e_{i,i+1}^\top.$

Based on $\mathcal{T}_{\alpha,D}$ and $\mathcal{T}_{\beta,D},$ the proposed tests are defined as follows.
First, we rearrange the indices of the nodes and recalculate the tests $\mathcal{T}_{\alpha,D}$ and $\mathcal{T}_{\beta,D}.$
This procedure is then repeated $\widetilde {M}$ times, where $\widetilde M$ is a prespecified constant.
Finally, we obtain the tests $\mathcal{T}_{\alpha,D}(\widetilde{M})$ and $\mathcal{T}_{\beta,D}(\widetilde{M})$ using
\[
\mathcal{T}_{\alpha,D}(\widetilde{M})=\max_{1\le s\le \widetilde{M}}{\mathcal{T}_{\alpha,D}^{s}}~~~\text{and}~~~\mathcal{T}_{\beta,D}(\widetilde{M})=\max_{1\le s\le \widetilde{M}}{\mathcal{T}_{\beta,D}^{s}},
\]
where $\mathcal{T}_{\alpha,D}^{s}$ and $\mathcal{T}_{\beta,D}^{s}$ denote the tests $\mathcal{T}_{\alpha,D}$
and $\mathcal{T}_{\beta,D}$ calculated at the $s$th replication, respectively.
It can be observed that $\mathcal{T}_{\alpha,D}(\widetilde{M})$ and $\mathcal{T}_{\beta,D}(\widetilde{M})$ approach zero under the null hypotheses
$H_{0\alpha,D}$ and $H_{0\beta,D},$ respectively.
Therefore, we reject $H_{0\alpha,D}$ if $\mathcal{T}_{\alpha,D}(\widetilde{M})> c_{\alpha,D}(v)$
and reject $H_{0\beta,D}$ if $\mathcal{T}_{\beta,D}(\widetilde{M})> c_{\beta,D}(v),$
where $c_{\alpha,D}(v)$ and $ c_{\beta,D}(v)$ are the critical values.
In practice, $c_{\alpha,D}(\nu)$ and $c_{\beta,D}(\nu)$ can be obtained using a resampling method discussed in Section \ref{section:application:TS}.

\begin{remark}
It is worth noting that the tests $\mathcal{T}_{\alpha,D}(\widetilde{M})$ and $\mathcal{T}_{\beta,D}(\widetilde{M})$
may exhibit reduced power compared to $\widetilde{\mathcal{T}}_{\alpha,D}$ and $\widetilde{\mathcal{T}}_{\beta,D}.$
However, our simulation studies indicate that the tests $\mathcal{T}_{\alpha,D}(\widetilde{M})$ and $\mathcal{T}_{\beta,D}(\widetilde{M})$
with $\widetilde{M}=3$ can be comparable to $\widetilde{\mathcal{T}}_{\alpha,D}$ and $\widetilde{\mathcal{T}}_{\beta,D}.$
A more detailed comparison of these tests is provided in Section \ref{section:SM:TD}.
\end{remark}

\section{Numerical studies}
\label{section:simulations}

It is well known that bandwidth significantly impacts the finite sample performance of kernel smoothing methods.
Inspired by \cite{lewbel1998}, we propose a bandwidth selection procedure that has proven effective in our simulation studies.
Specifically, let $\delta$ be an arbitrarily chosen constant.
By Theorem \ref{Theorem:LSE}, we can demonstrate that
$\delta=\mathbb{E}\big\{\big[\mathbb{I}(X_{1ij}+\delta>0)-\mathbb{I}(X_{1ij}>0)\big]/f(X_{1ij}|Z_{ij})\big\}.$
Define $\widehat\delta(h)=N^{-1}\sum_{1\le i\neq j\le n}\big[\mathbb{I}(X_{1ij}+\delta>0)-\mathbb{I}(X_{1ij}>0)\big]/\widehat f(X_{1ij}|Z_{ij}).$
We can then obtain an estimate of $h$ as
$\widehat h=\text{arg}\min_{h} \sum_{i=1}^{M_0} \big(\delta_i-\widehat\delta_i(h)\big)^2,$
where $\delta_i~(1\le i\le M_0)$ are prespecified grid points on the interval $(0,1]$,
and $M_0$ is a prespecified integer.
In subsequent simulation studies, we set $\delta_i=i/M_0$ and $M_0=10$ for the bandwidth selection procedure.
Next, we describe the kernel function utilized in our study.
We employ the biweight product kernel for parameter estimation, defined as follows:
$\mathcal{K}_{z}(z_1,,\dots,z_{p_1})=\Pi_{l=1}^{p_1}(15/16)(1-z_l^2)^2\mathbb{I}(|z_l|\le 1).$
A similar biweight product kernel is utilized for $\mathcal{K}_{xz}(x,z)$.
The rationale for this choice is its computational efficiency and generally favorable performance (see e.g., \citealp{hardle1992bandwidth}).

\subsection{Evaluating asymptotic properties}
\label{section:CN}

In this section, we conduct simulation studies to assess the finite sample performance of the proposed method.
The covariates $Z_{ij}~(1\le i\neq j\le n)$ are independently generated from a bivariate normal distribution
with mean 0 and covariance matrix $\Sigma\in\mathbb{R}^{2\times 2}.$
Denote $\tilde \sigma_{ij}$ as the $(i,j)$th entry of $\Sigma$.
Here, we set $\tilde\sigma_{11}=\tilde\sigma_{22}=1$ and $\tilde\sigma_{12}=\tilde\sigma_{21}=0.25.$
We then generate $X_{ij1}$ using $X_{ij1}=Z_{ij}^\top b+\mathcal{E}_{ij},$
where $b=(0.5,-0.5)^\top$ and $\mathcal{E}_{ij}~(1\le i\neq j\le n)$ are independently drawn from the standard normal distribution.
For each $i\in[n]$, the parameter $\alpha_{0i}$ is defined as
$\alpha_{0i}=-0.25\log(n)+(i-1)[0.25\log(n)+\rho_1\log(n)]/(n-1),$
where $\rho_1$ is a parameter that controls network sparsity.
We set $\beta_{0i}=\alpha_{0i}$ for $i\in[n-1]$
and $\beta_{0n}=0$ to ensure the identification of model \eqref{model}.
The parameter $\eta_0=(\eta_{01},\eta_{02})^\top=(-0.5,0.5)^\top.$
For the noises $\varepsilon_{ij}~(1\le i\neq j\le n),$ we consider the following three scenarios:
(i) $\varepsilon_{ij}$ are independently generated from the standard normal distribution $\mathcal{N}_1(0,1)$;
(ii) $\varepsilon_{ij}$ are independently generated from a logistic distribution
with the cumulative distribution function given by $[1+\exp\{-(x-\mu)/\upsilon\}]^{-1}$.
Here, we set $\mu=0$ and $\upsilon=1/2$, denoted by $\text{Logistic}(0,1/2)$;
and (iii) $\varepsilon_{ij}$ are independently generated from $\mathcal{N}_1(-0.3,0.91)$
with a probability of 0.75 and from $\mathcal{N}_1(0.9,0.19)$ with a probability of $0.25$, denoted as $\text{MNorm}_1.$
The first case corresponds to the probit regression model, while the second case corresponds to a generalized logistic regression model.
The last case is a mixture of normal distributions with the overall mixture having a mean of zero and a variance of one,
which is designed to yield a distribution that is both skewed and bimodal.
We set $n=50$ and $100$. For each configuration, we replicate 1000 simulations.

The value of $z_{1-\nu/2}$ is obtained using 10000 bootstrap iterations.
The results for $\alpha_{0j}~(j=1,n/2,n),$ $\alpha_{n/5,4n/5},$
$\alpha_{n/2,n/2+1}$ and $\eta_{0}$ are presented in Tables \ref{Tab:CN:50} and \ref{Tab:CN:100}.
Here, $\alpha_{i,j}=\alpha_{0i}-\alpha_{0j}.$
The results for $\widehat\beta$ are similar to those of $\widehat\alpha,$
which are omitted here to save space.
Tables \ref{Tab:CN:50} and \ref{Tab:CN:100} report the bias (Bias),
defined as the difference between the sample means of the proposed estimates and the true values,
the standard deviations (SDs) characterizing sample variations over 1000 simulations,
and the $95\%$ empirical coverage probability (CP).
We see that the proposed estimators are nearly unbiased
and the SDs decrease as the sample size increases.
The $95\%$ empirical coverage probabilities are reasonable.
{\color{black}Figure S4 in the Supplementary Material further presents the histogram of the standardized $\widehat\alpha_{n/2},$
which indicates that the distribution of the standardized estimators can be closely approximated by the standard normal distribution.
The results for the other estimators are similar and are therefore not presented here.
Additionally, we find that the performance of the proposed method is comparable
across the different distributions for $\varepsilon_{ij}$.
The primary reason is that in our simulation studies, $q_n$ is approximately $0.5\log(n)$,
which leads to that the residuals $(Y_{ij}-\alpha_i-\beta_j-Z_{ij}^\top\eta)$ exhibit light-tailed behavior.
Consequently, the least squares method generally performs well in this scenario.

We also compare the proposed method with \cite{yan2019}.
They specified the standard logistic distribution for the noise with $\mu=0$ and $\upsilon=1$,
denoted by $\text{Logistic(0,1)}$.
To ensure a fair comparison, when implementing the method developed in \cite{yan2019},
we constrain $\gamma_{01}$ to 1 while allowing $\eta_0$ to remain free.
Additionally, we utilize the bias-corrected formula for $\eta_0$.
Note that although we specify a logistic model in case (ii), the parameter $\upsilon$ is set to $1/2$.
Therefore, under the setting of case (ii), the model remains misspecified for \cite{yan2019}.
The results are presented in Table \ref{Tab:CYan}.
We observe the following: (1) When the noise distribution deviates from the standard logistic distribution,
the bias of our estimator approaches zero,
whereas the bias in their estimator for degree parameters is not negligible;
(2) When the standard logistic distribution is correctly specified,
their method is better than ours, even though
 both methods consistently estimate the unknown parameters;
(3) Their estimator for $\eta_0$ seems unbiased under the constraint $\gamma_{01}=1$,
even when the noise distribution is misspecified.
}

\begin{table}[!htpb]{\centering
{\footnotesize
\caption{The results of bias, standard deviation (SD) and $95\%$ empirical coverage probability (CP) of the estimators with $n=50$.}\label{Tab:CN:50}
\begin{tabular}{ccrcccrcccrcc}
\hline
& &\multicolumn{3}{c}{$\rho_1=0$} && \multicolumn{3}{c}{$\rho_1=0.1$} && \multicolumn{3}{c}{$\rho_1=0.2$} \\
\cline{3-5} \cline{7-9} \cline{11-13}
$\varepsilon_{ij}$ & & Bias & SD &  CP($\%$) & & Bias & SD &  CP($\%$) &  & Bias & SD &  CP($\%$)\\
\hline
$\mathcal{N}_1(0,1)$&	$\alpha_{01}$	&	0.089 	&	0.448 	&	94.0 	&	&	0.096 	&	0.446 	&	93.0 	&	&	0.076 	&	0.418 	&	93.7 	\\
	&	$\alpha_{0\frac{n}{2}}$	&	$0.014$ 	&	0.416 	&	95.2 	&	&	$-0.010$ 	&	0.396 	&	95.2 	&	&	0.011 	&	0.384 	&	94.8 	\\
	&	$\alpha_{0n}$	&	$-0.015$ 	&	0.375 	&	97.3 	&	&	$-0.040$ 	&	0.377 	&	95.2 	&	&	$-$0.048 	&	0.388 	&	94.2 	\\
	&	$\alpha_{\frac{n}{5},\frac{4n}{5}}$	&	0.085 	&	0.461 	&	94.1 	&	&	0.086 	&	0.433 	&	93.3 	&	&	0.097 	&	0.441 	&	91.9 	\\
	&	$\alpha_{\frac{n}{2},\frac{n}{2}+1}$	&	$-0.014$ 	&	0.450 	&	94.1 	&	&	$-0.015$ 	&	0.426 	&	93.1 	&	&	0.019 	&	0.390 	&	94.1 	\\
	&	$\eta_{01}$	&	$-0.003$ 	&	0.038 	&	93.8 	&	&	0.007 	&	0.038 	&	94.0 	&	&	0.006 	&	0.038 	&	93.6 	\\
	&	$\eta_{02}$	&	0.001 	&	0.040 	&	94.4 	&	&	$-$0.004 	&	0.038 	&	93.6 	&	&	$-$0.006 	&	0.036 	&	95.2 	\\
\hline
$\text{Logistic}(0,1/2)$	&	$\alpha_{01}$	&	0.092 	&	0.420 	&	92.9 	&	&	0.045 	&	0.443 	&	91.5 	&	&	0.056 	&	0.395 	&	92.5 	\\
	&	$\alpha_{0\frac{n}{2}}$	&	0.009 	&	0.388 	&	95.7 	&	&	0.009 	&	0.364 	&	95.9 	&	&	$-$0.004 	&	0.364 	&	95.1 	\\
	&	$\alpha_{0n}$	&	$-$0.027 	&	0.381 	&	96.8 	&	&	$-$0.018 	&	0.343 	&	96.1 	&	&	$-$0.047 	&	0.375 	&	93.2 	\\
	&	$\alpha_{\frac{n}{5},\frac{4n}{5}}$	&	0.095 	&	0.425 	&	94.0 	&	&	0.039 	&	0.431 	&	93.0 	&	&	0.084 	&	0.412 	&	92.5 	\\
	&	$\alpha_{\frac{n}{2},\frac{n}{2}+1}$	&	$-$0.008 	&	0.406 	&	95.2 	&	&	0.009 	&	0.369 	&	94.2 	&	&	0.008 	&	0.356 	&	95.3 	\\
	&	$\eta_{01}$	&	0.001 	&	0.039 	&	93.3 	&	&	0.005 	&	0.036 	&	93.5 	&	&	0.010 	&	0.034 	&	95.1 	\\
	&	$\eta_{02}$	&	$-$0.001 	&	0.037 	&	93.7 	&	&	$-$0.007 	&	0.034 	&	95.2 	&	&	$-$0.009 	&	0.035 	&	95.1 	\\
\hline
$\text{Mnorm}_1$&	$\alpha_{01}$	&	0.103 	&	0.489 	&	92.9 	&	&	0.061 	&	0.481 	&	91.0 	&	&	0.067 	&	0.418 	&	94.1 	\\
	&	$\alpha_{0\frac{n}{2}}$	&	0.023 	&	0.424 	&	95.4 	&	&	0.008 	&	0.395 	&	95.5 	&	&	$-$0.008 	&	0.374 	&	96.0 	\\
	&	$\alpha_{0n}$	&	$-$0.015 	&	0.402 	&	96.0 	&	&	$-$0.022 	&	0.378 	&	96.6 	&	&	$-$0.050 	&	0.417 	&	92.8 	\\
	&	$\alpha_{\frac{n}{5},\frac{4n}{5}}$	&	0.089 	&	0.448 	&	94.3 	&	&	0.070 	&	0.428 	&	94.4 	&	&	0.095 	&	0.453 	&	93.6 	\\
	&	$\alpha_{\frac{n}{2},\frac{n}{2}+1}$	&	0.004 	&	0.458 	&	95.4 	&	&	$-$0.003 	&	0.402 	&	94.9 	&	&	$-$0.009 	&	0.383 	&	95.3 	\\
	&	$\eta_{01}$	&	$-$0.002 	&	0.040 	&	94.2 	&	&	0.006 	&	0.039 	&	93.5 	&	&	0.006 	&	0.037 	&	93.9 	\\
	&	$\eta_{02}$	&	0.001 	&	0.039 	&	94.1 	&	&	$-$0.006 	&	0.035 	&	95.4 	&	&	$-$0.005 	&	0.036 	&	96.2 	\\
\hline
\end{tabular}}}
\end{table}

\begin{table}[!htpb]{\centering
{\footnotesize
\caption{The results of bias, standard deviation (SD) and $95\%$ empirical coverage probability (CP) of the estimators with $n=100$.}\label{Tab:CN:100}
\begin{tabular}{ccrcccrcccrcc}
\hline
& &\multicolumn{3}{c}{$\rho_1=0$} && \multicolumn{3}{c}{$\rho_1=0.1$} && \multicolumn{3}{c}{$\rho_1=0.2$} \\
\cline{3-5} \cline{7-9} \cline{11-13}
$\varepsilon_{ij}$ & & Bias & SD &  CP($\%$) & & Bias & SD &  CP($\%$) &  & Bias & SD &  CP($\%$)\\
\hline
$\mathcal{N}_1(0,1)$ &	$\alpha_{01}$	&	0.094 	&	0.412 	&	93.6 	&	&	0.102 	&	0.359 	&	93.0 	&	&	0.074 	&	0.359 	&	92.9 	\\
 	&	$\alpha_{0\frac{n}{2}}$	&	0.004 	&	0.352 	&	96.4 	&	&	0.017 	&	0.292 	&	97.9 	&	&	0.005 	&	0.284 	&	97.2 	\\
 	&	$\alpha_{0n}$	&	$-0.038$ 	&	0.303 	&	98.3 	&	&	$-$0.023 	&	0.284 	&	97.2 	&	&	$-$0.051 	&	0.295 	&	95.8 	\\
 	&	$\alpha_{\frac{n}{5},\frac{4n}{5}}$	&	0.080 	&	0.432 	&	93.6 	&	&	0.086 	&	0.364 	&	92.5 	&	&	0.069 	&	0.360 	&	91.7 	\\
 	&	$\alpha_{\frac{n}{2},\frac{n}{2}+1}$	&	0.002 	&	0.414 	&	94.6 	&	&	0.060 	&	0.381 	&	93.1 	&	&	0.057 	&	0.339 	&	94.1 	\\
 	&	$\eta_{01}$	&	$-$0.010 	&	0.023 	&	92.2 	&	&	$-$0.001 	&	0.022 	&	93.4 	&	&	$-$0.001 	&	0.021 	&	96.1 	\\
 	&	$\eta_{02}$	&	0.011 	&	0.023 	&	92.2 	&	&	0.000 	&	0.022 	&	94.5 	&	&	$-$0.001 	&	0.020 	&	96.1 	\\
\hline
$\text{Logistic}(0,1/2)$ &	$\alpha_{01}$	&	0.094 	&	0.401 	&	93.1 	&	&	0.067 	&	0.346 	&	93.1 	&	&	0.073 	&	0.332 	&	93.9 	\\
	&	$\alpha_{0\frac{n}{2}}$	&	0.003 	&	0.327 	&	97.1 	&	&	$-$0.002 	&	0.289 	&	97.4 	&	&	0.004 	&	0.278 	&	96.5 	\\
	&	$\alpha_{0n}$	&	$-$0.028 	&	0.284 	&	98.3 	&	&	$-$0.020 	&	0.264 	&	97.3 	&	&	$-$0.038 	&	0.272 	&	96.5 	\\
	&	$\alpha_{\frac{n}{5},\frac{4n}{5}}$	&	0.092 	&	0.401 	&	92.6 	&	&	0.056 	&	0.358 	&	91.7 	&	&	0.058 	&	0.323 	&	92.0 	\\
	&	$\alpha_{\frac{n}{2},\frac{n}{2}+1}$	&	0.066 	&	0.384 	&	94.1 	&	&	0.054 	&	0.369 	&	92.7 	&	&	0.059 	&	0.312 	&	94.5 	\\
	&	$\eta_{01}$	&	$-$0.008 	&	0.022 	&	91.8 	&	&	0.000 	&	0.020 	&	96.0 	&	&	0.002 	&	0.019 	&	95.8 	\\
	&	$\eta_{02}$	&	0.008 	&	0.022 	&	93.2 	&	&	0.000 	&	0.020 	&	95.6 	&	&	$-$0.002 	&	0.019 	&	96.1 	\\
\hline
$\text{Mnorm}_1$	&	$\alpha_{01}$	&	0.089 	&	0.461 	&	93.7 	&	&	0.102 	&	0.402 	&	90.5 	&	&	0.102 	&	0.380 	&	92.4 	\\
	&	$\alpha_{0\frac{n}{2}}$	&	0.015 	&	0.364 	&	97.2 	&	&	$-$0.001 	&	0.305 	&	96.5 	&	&	0.004 	&	0.291 	&	97.2 	\\
	&	$\alpha_{0n}$	&	$-$0.030 	&	0.310 	&	98.5 	&	&	$-$0.035 	&	0.282 	&	97.6 	&	&	$-$0.037 	&	0.327 	&	93.7 	\\
	&	$\alpha_{\frac{n}{5},\frac{4n}{5}}$	&	0.074 	&	0.394 	&	94.7 	&	&	0.093 	&	0.353 	&	93.0 	&	&	0.064 	&	0.367 	&	91.8 	\\
	&	$\alpha_{\frac{n}{2},\frac{n}{2}+1}$	&	0.001 	&	0.402 	&	94.4 	&	&	0.064 	&	0.377 	&	93.2 	&	&	0.051 	&	0.325 	&	93.3 	\\
	&	$\eta_{01}$	&	$-$0.007 	&	0.024 	&	93.9 	&	&	0.000 	&	0.022 	&	94.0 	&	&	$-$0.001 	&	0.021 	&	95.6 	\\
	&	$\eta_{02}$	&	0.008 	&	0.023 	&	93.1 	&	&	0.001 	&	0.021 	&	95.8 	&	&	0.000 	&	0.020 	&	95.9 	\\
\hline
\end{tabular}}}
\end{table}

\begin{table}[!htpb]
\begin{center}
\caption{The comparison results of the bias of the estimators obtained by our method and \cite{yan2019}'s method with $n=100.$}\label{Tab:CYan}
{\footnotesize
\begin{tabular}{ccccccccccccc}
\hline
& &\multicolumn{2}{c}{$\mathcal{N}_1(0,1)$} && \multicolumn{2}{c}{$\text{Logistic}(0,1/2)$} && \multicolumn{2}{c}{$\text{MNorm}_1$}
& &\multicolumn{2}{c}{$\text{Logistic}(0,1)$}\\
\cline{3-4} \cline{6-7} \cline{9-10} \cline{12-13}
$\rho_1$ & & Ours & Yan et al. && Ours & Yan et al. &  & Ours & Yan et al. & & Ours & Yan et al.\\
\hline
0	&	$\alpha_{01}$		&	~~0.094 	&	$-0.501 $	&&	~~0.094 	&	$-0.626$ 	&&	~~0.089 	&	$-1.094$ && ~~0.105 & $-0.007$	\\
	&	$\alpha_{0n}$		&	$-0.038$ 	&	~~0.029 	&&	$-0.028 $	&	~~0.006 	&&	$-0.030$ 	&	~~0.041 && $-0.058$ & $-0.009$	\\
	&	$\alpha_{\frac{n}{5},\frac{4n}{5}}$		&	~~0.080 	&	$-0.330$ 	&&	~~0.092 	&	$-0.367$ 	&&	~~0.074 	&	$-0.664$ && ~~0.098 & $-0.027$	\\
	&	$\alpha_{\frac{n}{2},\frac{n}{2}+1}$		&	~~0.001 	&	$-0.272$ 	&&	~~0.066 	&	$-0.308$ 	&&	~~0.001 	&	$-0.532$ && $-0.043$ &  $-0.009$\\
	&	$\eta_{01}$		&	$-0.010$ 	&	~~0.001 	&&	$-0.008$ 	&	$-0.001$ 	&&	$-0.007$ 	&	$-0.001$ && $-0.039$ & $-0.001$
	\\
	&	$\eta_{02}$		&	~~0.011 	&	~~0.001 	&&	~~0.008 	&	~~0.000 	&&	~~0.008 	&	~~0.001 && ~~0.037 & ~~0.001	\\
0.1	&	$\alpha_{01}$		&	~~0.102 	&	$-0.497$ 	&&	~~0.067 	&	$-0.601$ 	&&	~~0.102 	&	$-1.021$ & & ~~0.096 & $-0.013$	\\
	&	$\alpha_{0n}$		&	$-0.023$ 	&	~~0.192 	&&	$-0.020$ 	&	~~0.223 	&&	$-0.035$ 	&	~~0.383 && ~~0.077 & ~~0.024	\\
	&	$\alpha_{\frac{n}{5},\frac{4n}{5}}$		&	~~0.086 	&	$-0.381$ 	&&	~~0.056 	&	$-0.482$ 	&&	~~0.093 	&	$-0.802$ &&  ~~0.092 & $-0.039$	\\
	&	$\alpha_{\frac{n}{2},\frac{n}{2}+1}$		&	~~0.060 	&	$-0.335$ 	&&	~~0.054 	&	$-0.394$ 	&&	~~0.064 	&	$-0.653$ && $-0.017$ & $-0.039$		\\
	&	$\eta_{01}$		&	$-0.001$ 	&	~~0.000 	&&	~~0.000 	&	~~0.000 	&&	~~0.000 	&	~~0.000 && $-0.034$ & ~~0.001	\\
	&	$\eta_{02}$		&	~~0.000 	&	$-0.002$ 	&&	~~0.000 	&	~~0.000 	&&	~~0.001 	&	~~0.000 && ~~0.031 & ~~0.001	\\
\hline
\end{tabular}}
\end{center}
\end{table}

\begin{figure}[!htpb]
\centering
\includegraphics[width=0.9\textwidth]{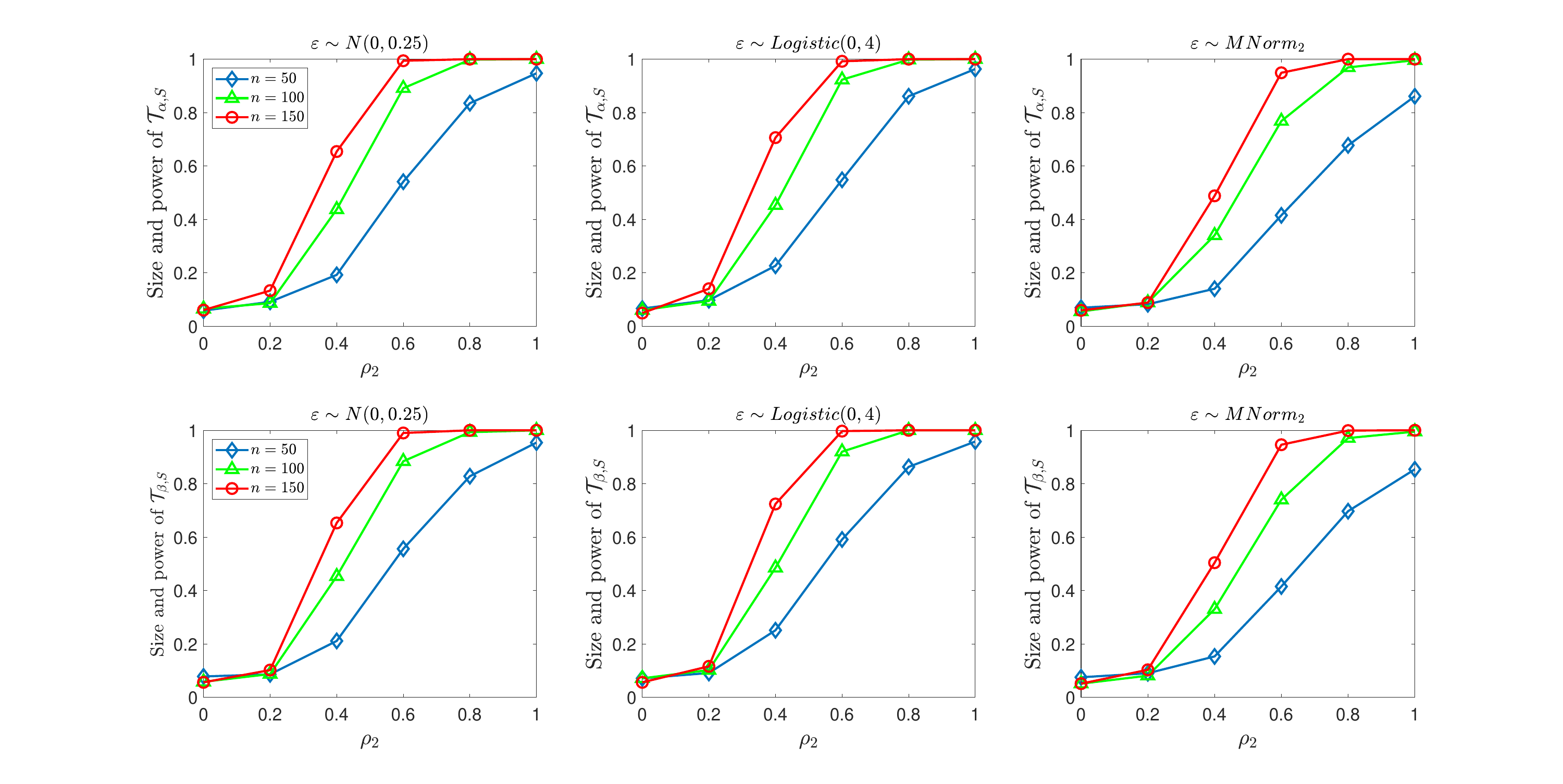}
\caption{Empirical size and power of $\mathcal{T}_{\alpha,S}$ and $\mathcal{T}_{\beta,S}.$}\label{Figure:SM:TS}
\end{figure}

\subsection{Testing for sparse signal}\label{section:SM:TS}
In this section, we evaluate the finite sample performance of $\mathcal{T}_{\alpha,S}$
and $\mathcal{T}_{\beta,S}$ in testing for sparse signals.
To do that, we set $\alpha_{0i}=-2i/n$ if $i\in [\rho_2 n]$ and $\alpha_{0i}=0$ otherwise.
For simplicity, we set $\beta_{0i}=\alpha_{0i}$ for $i\in[n-1]$ and $\beta_{0n}=0$.
The parameter $\rho_2$ is chosen from the set $\{0, 0.2, 0.4, 0.6, 0.8, 1\}.$
When $\rho_2=0,$ the null hypotheses $H_{0\alpha,S}$ and $H_{0\beta,S}$ are valid.
As $\rho_2$ increases from 0 to 1, the departures from the null hypotheses $H_{0\alpha,D}$ and $H_{0\beta,D}$ increase,
while the density of the generated network decreases from $50\%$ to $33\%$.
The sample size is set to $n=50, 100$ and $150.$
We take the significance level $\nu$ to be $0.05$.
The critical values are calculated using the bootstrap method with 10000 simulated realizations.
For the noise term $\varepsilon_{ij},$ we consider the following three settings:
(i) $\varepsilon_{ij}\sim \mathcal{N}_1(0,0.25)$;
(ii)  $\varepsilon_{ij}\sim\text{Logistic}(0,1/4)$
and (ii) $\varepsilon_{ij}$ are independently generated from $\mathcal{N}_1(-0.3,0.5)$
with a probability of 0.75 and from $\mathcal{N}_1(0.9,0.5)$ with a probability of $0.25$.
The corresponding distribution is denoted by $\text{MNorm}_2.$
All other settings are consistent with those in Section \ref{section:CN}.

Figure \ref{Figure:SM:TS} shows the empirical size and
power of the test statistics $\mathcal{T}_{\alpha,S}$ and $\mathcal{T}_{\beta,S}$ at the significance level of 0.05.
We observe that the empirical sizes of the proposed tests are close to the nominal level of $0.05,$
and the tests exhibit reasonable power in detecting deviations from the null hypothesis.
The powers of the tests increase as $\rho_2$ increases from $0.2$ to $1$.
Additionally, the powers of the tests improve as the sample size $n$ increases from $50$ to $150.$

Next, we evaluate the finite sample performance of the method developed for support recovery.
For this, we set $\alpha_{0}$ and $\beta_{0}$ as follows:
$\alpha_{0}=(-1,2,-2,1.5,-3,-1.5\mathds{1}_{d}^\top,0_{n-d-5}^\top)^\top$ and $\beta_{0}=(0_5^\top,-1,2,-2,1.5,-3,-1.5\mathds{1}_{d}^\top,0_{n-d-10}^\top)^\top,$
where $d=\lfloor n/15 \rfloor.$ We set $n=100$ and $150.$
To assess the accuracy of the support recovery, we adopt the following similarity measure \citep{cai2013two}:
$\mathcal{M}(\widehat{\mathcal{S}},\mathcal{S}_{0})=|\widehat{\mathcal{S}}\cap\mathcal{S}_{0}|/\sqrt{|\widehat{\mathcal{S}}||\mathcal{S}_0|},$
where $0\le \mathcal{M}(\widehat{\mathcal{S}},\mathcal{S}_{0})\le 1$.
A value of 0 indicates that the intersection of $\widehat{\mathcal{S}}$ and $\mathcal{S}_{0}$ is empty,
while a value of 1 indicates exact recovery.
The results are presented in Table \ref{Tab:Supp},
which summarizes the mean and the standard deviation of $\mathcal{M}(\widehat{\mathcal{S}},\mathcal{S}_{0})$,
as well as the counts of false positives (FP) and false negatives (FN) based on 1000 replications.
We observe that $\mathcal{M}(\widehat{\mathcal{S}}_{\alpha}(2),\mathcal{S}_{0\alpha})$
tends to approach one as $n$ increases.
Furthermore, the values of false positives and false negative are relatively low and decrease as $n$ increases.

\begin{table}[!htpb]
\begin{center}
\caption{The mean and standard deviation (SD) of $\mathcal{M}(\widehat{\mathcal{S}},\mathcal{S}_0)$, and the numbers of false positives (FP) and false negatives (FN).}\label{Tab:Supp}
{\small
\begin{tabular}{ccccccccccc}
\hline
& &\multicolumn{4}{c}{$\alpha_0$} && \multicolumn{4}{c}{$\beta_0$}\\
\cline{3-6} \cline{8-11}
$\varepsilon_{ij}$ & $n$ & Mean & SD &  FP & FN &  & Mean & SD &  FP & FN\\
\hline
$\mathcal{N}_1(0,0.25)$	&	$n=100$	&	0.978 	&	0.057 	&	0.229 	&	0.281 	&	&	0.977 	&	0.056 	&	0.264 	&	0.264 	\\
	&	$n=150$	&	0.989 	&	0.051 	&	0.207 	&	0.129 	&	&	0.986 	&	0.056 	&	0.230 	&	0.108 	\\
Logistic(0,1/4)	&	$n=100$	&	0.977 	&	0.062 	&	0.377 	&	0.263 	&	&	0.979 	&	0.060 	&	0.317 	&	0.255 	\\
	&	$n=150$	&	0.992 	&	0.021 	&	0.220 	&	0.048 	&	&	0.992 	&	0.019 	&	0.190 	&	0.056 	\\
$\text{Mnorm}_2$	&	$n=100$	&	0.962 	&	0.071 	&	0.350 	&	0.568 	&	&	0.957 	&	0.076 	&	0.383 	&	0.585 	\\
	&	$n=150$	&	0.984 	&	0.064 	&	0.262 	&	0.241 	&	&	0.983 	&	0.059 	&	0.233 	&	0.224 	\\
\hline
\end{tabular}}
\end{center}
\end{table}

\subsection{Testing for degree heterogeneity}\label{section:SM:TD}
In this section, we evaluate the finite sample performance of the tests
$\mathcal{T}_{\alpha,D}(\widetilde M)$ and $\mathcal{T}_{\beta,D}(\widetilde M)$ for assessing the presence of degree heterogeneity.
We consider $\mathcal{G}_1=[n]$ and $\mathcal{G}_2=[n-1]$.
Let $\alpha_{0i}=-\rho_3i/n$ for $i\in[n],$ $\beta_{0i}=\alpha_{0i}$ for $i\in[n-1]$ and $\beta_{0n}=0.$
Under these settings, we have $|\alpha_{0i}-\alpha_{0,i+1}|=\rho_3/n$, while $\max_{1\le i<j\le n}|\alpha_{0i}-\alpha_{0,j}|=\rho_3(1-1/n).$
In this analysis, we set $\widetilde M\in\{0, 1, 2, 3\}$ and $\rho_3\in\{0, 0.2, 0.4, 0.6\}.$
Note that $\mathcal{T}_{\alpha,D}(0)$ and $\mathcal{T}_{\beta,D}(0)$
are identical to the tests $\mathcal{T}_{\alpha,D}$ and $\mathcal{T}_{\beta,D},$ respectively,
and may not be sensitive to the signals $|\alpha_{0i}-\alpha_{0,i+1}|=\rho_3/n$ and $|\beta_{0i}-\beta_{0,i+1}|=\rho_3/n$.
Furthermore, the null hypotheses $H_{0\alpha,D}$ and $H_{0\beta,D}$ hold when $\rho_3=0.$
All other settings remain consistent with those used previously.

The results for $\mathcal{T}_{\alpha,D}(\widetilde{M})$ are summarized in Figure \ref{Figure:SM:compare},
while those for $\mathcal{T}_{\beta,D}(\widetilde{M})$ are similar and omitted here to save space.
We observe that the empirical sizes of all tests are close to the nominal level of $0.05.$
However, the test $\mathcal{T}_{\alpha,D}$ exhibits almost no power,
thus failing to detect deviations from the null hypothesis.
Additionally, Figure \ref{Figure:SM:compare} indicates that
the test $\mathcal{T}_{\alpha,D}(\widetilde M)$ with $\widetilde M\ge 2$
significantly enhances the power of $\mathcal{T}_{\alpha,D}$,
and the power of $\mathcal{T}_{\alpha,D}(\widetilde M)$ increases when $\widetilde M$ increases from $0$ to $3$.
Notably, $\mathcal{T}_{\alpha,D}(2)$ is comparable to $\mathcal{T}_{\alpha,D}(3)$.
Figure S5 in the Supplementary Material shows that the time required to compute the test $\mathcal{T}_{\alpha,D}(\widetilde M)$
may increase linearly with $\widetilde M$.
For instance, when $n=150$ and $\varepsilon_{ij} \sim \mathcal{N}(0,0.25)$,
the time required to calculate $\widetilde{T}_{\alpha,D}$ is approximately 450 seconds for each replication.
We also construct simulation studies to evaluate the finite sample performance of the proposed method
for the conditionally independent cases and weighted networks.
The results are summarized in Section C of the Supplementary Material,
which suggest that the proposed method works well.

\begin{figure}
\centering
\includegraphics[width=1\textwidth]{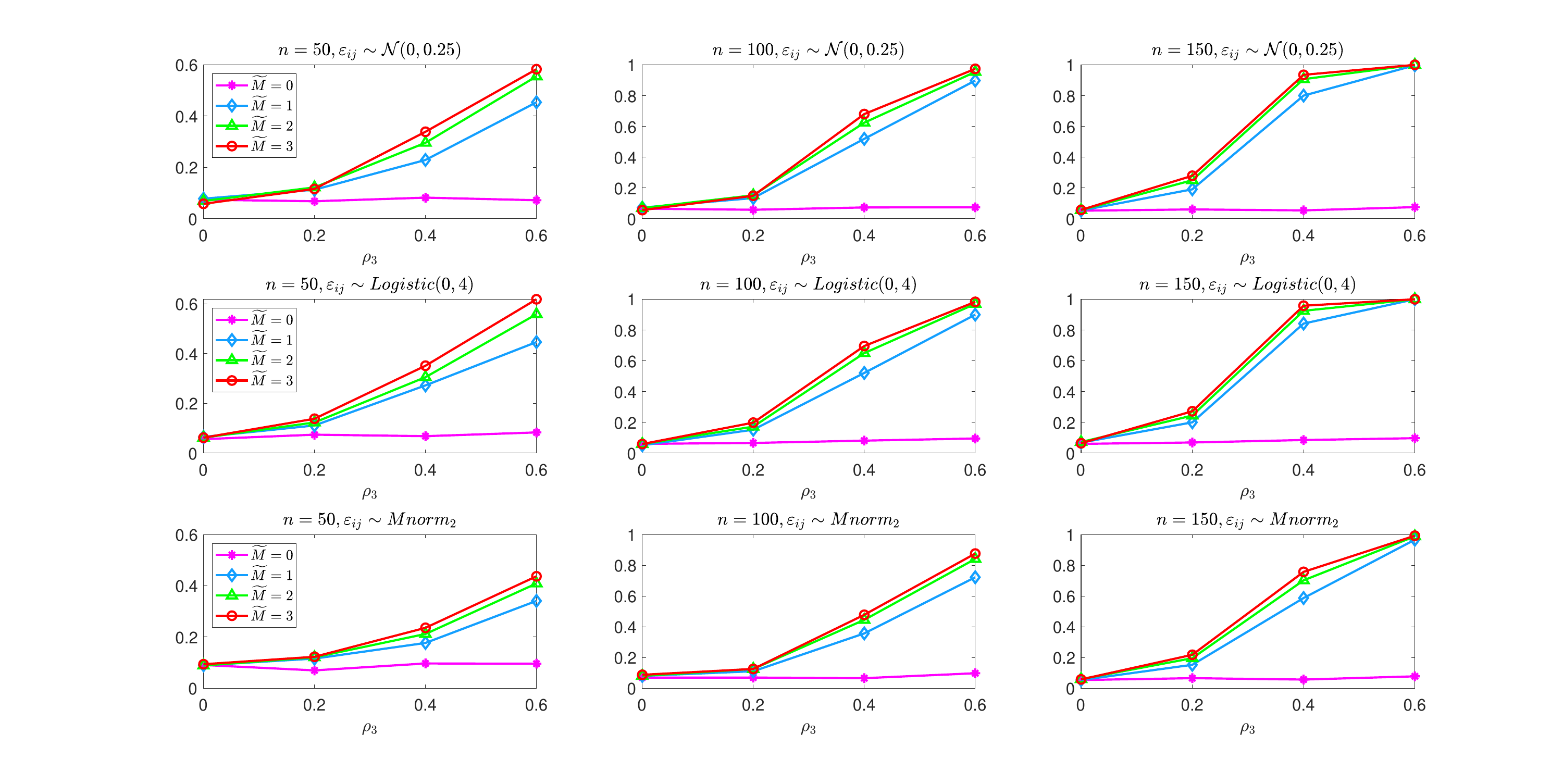}
\caption{Empirical size and power of $\mathcal{T}_{\alpha,D}(\widetilde M)$ with $\widetilde M=0,1,2$ and $3.$}\label{Figure:SM:compare}
\end{figure}

\subsection{Real data analysis}\label{Applicarion}

The data used in this study originates from Lazega's network analysis of corporate law partnerships conducted in a Northeastern U.S. corporate law firm between 1988 and 1991 in New England \citep{lazega2001collegial}.
The dataset contains 71 attorneys, including partners and associates of this firm.
In the following, we analyze the friendship network among these attorneys,
who were asked to identify their social contacts outside of work.
In addition, several covariates were collected for each attorney.
To avoid the  curse of dimensionality issue,
we focus on three specific covariates: gender (male and female), years with the firm, and age.
Prior to analysis, the variables for years with the firm and age are normalized by subtracting the mean and dividing by the standard deviation.
As in \cite{yan2019}, we define the covariates
for each pair as the absolute difference between continuous variables
and indicators to determine the equality of categorical variables.
{\color{black} We set age as the special regressor $X_{ij1}$.
To determine the sign of $\gamma_{01}$,
we partition the support of $X_{ij1}$ into 7 equal subintervals.
The values of $\text{Count}_k$ are 249, 149, 119, 22, 17, 4 and 0 for $1\le k\le 7$, respectively,
exhibiting a significant decreasing pattern as $k$ increases. Therefore, we set $\gamma_{01}=-1$.}

In the dataset, individuals are labeled from 1 to 71.
After removing individuals with zero in-degrees or out-degrees,
we analyze the remaining 63 vertices.
We set $\beta_{71} = 0$ as the reference.
The bandwidth is selected using the procedure developed above,
and the kernel employed is the biweight kernel.
As shown in Figure S6 of the Supplementary Material,
$\widehat h=0.7651$ results in the smallest prediction error.

We first conduct the testing procedures outlined in Section \ref{section:application:TS}
to assess whether $\alpha_{0i}$ and $\beta_{0j}$ are equal to zero.
The resulting p-values for testing $\alpha_{0i}=0~(i\in[n])$ are less than $0.001$,
while those for testing $\beta_{0i}=0~(i\in[n-1])$ are $0.021$.
This indicates the presence of nonzero entries in the parameters $\alpha_{0i}$ and $\beta_{0j}$ at the significance level of $0.05$.
We then apply the support recovery procedure to identify the nonzero signals.
The nonzero out-degree parameters are reported in Table \ref{Tab:NZ}.
Two nodes, ``11" and ``35", exhibit nonzero in-degree parameters with respective estimates of $0.661$ and $0.732$.
Additionally, we implement the procedure described in Section \ref{section:application:TDH} to test for the presence of degree heterogeneity.
The p-values obtained from testing $\alpha_{0i}$ and $\beta_{0i}$ are less than 0.001.
This indicates the presence of degree heterogeneity for both out-degree and in-degree at the significance level of $0.05$.

For the homophily effects, the estimated coefficient for gender is 0.112, with a confidence interval $[0.017, 0.207]$.
This indicates that gender has a positive effect on the formation of friendships among lawyers,
suggesting that they are more likely to befriend individuals of the same gender.
In addition, the estimated coefficient for years with the firm is $-0.219$, with a confidence interval $[-0.267,-0.053],$
indicating a negative effect on network formation.
This implies that a greater disparity in years with the firm between two lawyers
is associated with a lower likelihood of friendship.
{\color{black}Although these results are consistent with the findings reported by \cite{yan2019},
we develop a graphical method to compare the goodness of fit for the network data.
To save space, the details are provided in Section C.5 of the Supplementary Material,
and the results are presented in Figure S3.
We observe that the method proposed by \cite{yan2019}  underestimates both out-degree and in-degree,
suggesting a potential violation of the logistic parametric assumption.
In addition, we calculate the $L_2$-norm errors between the observed degrees and the fitted degrees.
The values for out- and in-degrees obtained using our method are $33.21$ and $32.70$, respectively,
compared to $39.84$ and $37.08$ obtained from Yan et al.'s method. This indicates an improvement in our method.}


\begin{table}[!htpb]
\begin{center}
\caption{Estimation results of the selected degree parameters in the real data analysis. Here, CI denotes the pointwise $95\%$ confidence interval.}\label{Tab:NZ}
{\small
\begin{tabular}{cccccccccccc}
\hline
ID & $\widehat{\alpha}_i$ & CI && ID & $\widehat{\alpha}_i$ & CI  && ID & $\widehat{\alpha}_i$ & CI\\
\hline
1 	&	0.641 	&	$[	0.230 	,	1.053 	]$	&&	20 	&	0.729 	&	$[	0.317 	,	1.140 	]$	&&	40 	&	1.010 	&	$[	0.598 	,	1.421 	]$	\\
2 	&	0.695 	&	$[	0.284 	,	1.106 	]$	&&	22 	&	0.872 	&	$[	0.460 	,	1.283 	]$	&&	43 	&	0.664 	&	$[	0.253 	,	1.075 	]$	\\
3 	&	1.120 	&	$[	0.709 	,	1.532 	]$	&&	25 	&	0.695 	&	$[	0.284 	,	1.106 	]$	&&	45 	&	0.864 	&	$[	0.453 	,	1.276 	]$	\\
8 	&	0.839 	&	$[	0.428 	,	1.251 	]$	&&	27 	&	0.632 	&	$[	0.221 	,	1.044 	]$	&&	47 	&	0.868 	&	$[	0.456 	,	1.279 	]$	\\
10 	&	1.058 	&	$[	0.647 	,	1.470 	]$	&&	29 	&	1.043 	&	$[	0.632 	,	1.454 	]$	&&	50 	&	0.607 	&	$[	0.196 	,	1.018 	]$	\\
11 	&	1.094 	&	$[	0.682 	,	1.505 	]$	&&	35 	&	1.037 	&	$[	0.626 	,	1.448 	]$	&&	56 	&	1.493 	&	$[	1.082 	,	1.905 	]$  \\
12 	&	1.052 	&	$[	0.640 	,	1.463 	]$	&&	37 	&	0.624 	&	$[	0.213 	,	1.036 	]$	&&	57 	&	1.117 	&	$[	0.706 	,	1.528 	]$	\\
15 	&	1.200 	&	$[	0.788 	,	1.611 	]$	&&	38 	&	0.748 	&	$[	0.337 	,	1.160 	]$	&&	58 	&	1.168 	&	$[	0.756 	,	1.579 	]$ 	\\
18 	&	0.691 	&	$[	0.280 	,	1.102 	]$	&&	39 	&	0.820 	&	$[	0.409 	,	1.232 	]$	&&		&           &                               \\
\hline
\end{tabular}}
\end{center}
\end{table}

\section{Conclusions}
\label{Sec:discussion}

We proposed a semiparametric model to analyze the effects of homophily and degree heterogeneity on network formation.
The model was practically flexible as it did not specify the distribution of the latent random variables.
We then developed a kernel-based least squares method to estimate the unknown parameters
and derived the asymptotic properties of the estimators, including consistency and asymptotic normality.
{\textcolor{black}{The extensions to the cases that $\varepsilon_{ij}$ may depend on the covariates or edges may take weighted values,
are presented in Section A of the Supplementary Material.
}}

Several topics need to be addressed in future studies.
First, many large-scale network datasets are highly sparse,
which necessitates additional research on how to adapt the proposed method for analyzing such sparse networks.
Second, we consider a nonparametric method for the conditional density function of $X_{ij1}$ given $Z_{ij}$.
It may suffer from the curse of dimensionality, especially when the dimension of $Z_{ij}$ is large.
Although we propose a linear regression model between $X_{ij1}$ and $Z_{ij}$ in Remark 1 to address this issue,
the theoretical analysis necessitates further effort.
Finally, extending our method to analyze more complex network data,
such as dynamic networks, presents an intriguing area for future research.

\begin{center}
{\large\bf Supplementary material}
\end{center}

\noindent The supplementary material contains the proofs of Theorems \ref{Theorem:LSE}-\ref{Theorem:hatGA:theta}, Corollary \ref{Ident}
and Propositions \ref{proposition:power}-\ref{proposition:SR},
and additional results in the simulation studies and real data analysis.

\spacingset{1.3}
\bibliographystyle{Chicago}
\bibliography{reference}

\newpage

\setcounter{page}{1}
\setcounter{section}{0}
\renewcommand{\thesection}{\Alph{section}}

\if0\blind
{
  \bigskip
  \bigskip
  \bigskip
  \begin{center}
    {\LARGE\bf Supplementary Material to ``Inference in semiparametric formation models for directed networks"}
\end{center}
  \medskip
} \fi

\spacingset{1.8}

The supplementary material is organized as follows.
Section \ref{extension} presents extensions to the cases that $\varepsilon_{ij}$ may depend on the covariates or edges may take weighted values.
Section \ref{section-identification} presents the proofs of Theorem 1 and Corollary 1.
Section \ref{sec-condition-c4} shows that if $\widetilde Z_1,\dots, \widetilde Z_p$ are not collinear
and do not almost surely lie within the linear subspace spanned by the column vectors of $U$, then Condition (C4) holds.
Section \ref{section-proof-theorem} presents the proofs of Theorems 2-6 and Propositions 1-2.
Section \ref{section-add-simulation} presents the additional results in the simulation studies,
 including simulation results for conditionally independent cases, comparison with Candelaria's method, asymptotic relative efficiency and simulation results for weighted networks.
Let $C_0,~C_1,~C_2\dots$ be some constants that may change from line to line.

\section{Extensions}
\label{extension}

In the main paper, we conduct semiparametric inference under the assumption
that the distribution of $\varepsilon_{ij}$ is independent of $X_{ij}$ and the edges are binary.
However, in some cases, $\varepsilon_{ij}$ may depend on the covariates, and edges may take weighted values.
In this section, we extend the proposed model to account for these scenarios.

\subsection{Extension to conditionally independent cases}

We relax Condition (C3) as follows.

\noindent{\bf Condition (C3').} The conditional density function of $\varepsilon_{ij}$ satisfies
$f_{\varepsilon}(\varepsilon_{ij}|X_{ij1},Z_{ij})=f_{\varepsilon}(\varepsilon_{ij}|Z_{ij}).$
In addition, $\mathbb{E}(\varepsilon_{ij}|Z_{ij})=0$ almost surely.

\begin{sloppypar}
Condition (C3') assumes that $\varepsilon_{ij}$ is conditionally independent of $X_{ij1}$,
but it can depend on $Z_{ij}$.
Define $\widehat{\text{Cov}(\epsilon)}=\text{diag}\{\widehat\epsilon_1^\top,\dots,\widehat\epsilon_n^\top\},$
and $\widehat{\text{Cov}(Q)}=\text{diag}\{\widehat{Q}_1^\top,\dots,\widehat{Q}_n^\top\},$
where $\widehat\epsilon_i=(\widehat\epsilon_{i1},\dots,\widehat\epsilon_{i,i-1},\widehat\epsilon_{i,i+1},\dots,\widehat\epsilon_{in})^\top$
and $\widehat{Q}_i=(\widehat Q_{i1},\dots,\widehat Q_{i,i-1},\widehat Q_{i,i+1},\dots,\widehat Q_{in})^\top.$
The following theorem extends the results of Theorem \ref{Theorem:hatGA:theta}.
\end{sloppypar}

\begin{theorem}\label{Theorem:hatGA:Here}
Suppose that Conditions (C1)-(C2), (C3') and (C4)-(C8) hold,
and there exist constants $0<\sigma_{\epsilon L}^2<\sigma_{\epsilon U}^2<\infty$ such that
$\sigma_{\epsilon L}^2<\mathbb{E}(\epsilon_{ij}^2)<\sigma_{\epsilon U}^2$ for all $i\neq j\in [n].$
If the condition \eqref{Theorem:GA:theta:condition} in Theorem \ref{Theorem:GA:theta} holds, then we have
\begin{align*}
&\sup_{x\in\mathbb{R}^{M_{1\alpha}}}\big|\mathbb{P}\big(\sqrt{n-1}\mathcal{L}_{1\alpha}(\widehat\alpha-\alpha_0)\le x\big)
-\mathbb{P}\big(\mathcal{L}_1^\alpha\widehat G_3\le x\big)\big|=o(1)\\
\text{and}~~~&\sup_{x\in\mathbb{R}^{M_{1\beta}}}\big|\mathbb{P}\big(\sqrt{n-1}\mathcal{L}_{1\beta}(\widehat\beta-\beta_0)\le x\big)
-\mathbb{P}\big(\mathcal{L}_1^\beta\widehat G_3\le x\big)\big|=o(1),
\end{align*}
where $\widehat G_3\sim \mathcal{N}_{2n-1}(0,(n-1)V^{-1}U^\top \widehat{\text{Cov}(\epsilon)}UV^{-1}).$
Moreover, if there exist constants $0<\sigma_{Q L}^2<\sigma_{Q U}^2<\infty$
such that $\sigma_{Q L}^2<\mathbb{E}(Q_{ij}^2)<\sigma_{Q U}^2$ for all $i\neq j\in [n],$
and the condition \eqref{Theorem:GA:eta:condition} in Theorem \ref{Theorem:GA:eta} holds,
then we have
\begin{align*}
&\sup_{x\in\mathbb{R}^{M_2}}\big|\mathbb{P}\big(\sqrt{N}\mathcal{L}_2(\widehat\eta-\eta_0)\le x\big)-\mathbb{P}\big(\mathcal{L}_2 \widehat G_4\le x\big)\big|=o(1),
\end{align*}
where $\widehat G_4\sim \mathcal{N}_{p}(0,N(Z^\top DZ)^{-1}Z^\top D\widehat{\text{Cov}(Q)}DZ(Z^\top DZ)^{-1}).$
\end{theorem}

Based on Theorem \ref{Theorem:hatGA:Here}, the inference methods developed in Section \ref{section:application}
can be extended to the parameters $\alpha_0$ and $\beta_0$.
Details are omitted here to save space.

\subsection{Extension to weighted networks}

We generalize the proposed method for analyzing weighted network data.
For this purpose, we consider the following model:
\begin{align} \label{model2}
    A_{ij}=\sum_{l=0}^{R-1} \pi_l \mathbb{I}(\omega_{0l}<\alpha_{0i}+\beta_{0j}+X_{ij1}+Z_{ij}^\top\eta_0-\varepsilon_{ij}\le \omega_{0,l+1}),
\end{align}
where $R$ is a specified constant, $\pi_l$ denotes the weight,
and $\omega_{00}<\omega_{01}<\dots<\omega_{0R}$ are unknown threshold parameters.
We set $\omega_{00}=-\infty$ and $\omega_{0,R}=\infty$.
Under model \eqref{model2},  we have $A_{ij}\in\{\pi_0,\pi_1,\dots,\pi_{R-1}\}$,
and $A_{ij}=\pi_l$ if $(\alpha_{0i}+\beta_{0j}+X_{ij1}+Z_{ij}^\top\eta_0-\varepsilon_{ij})\in(\omega_{0l}, \omega_{0,l+1}].$
For the identification of model \eqref{model2}, we assume $\gamma_{01}=1$ and $\alpha_{0n}=\beta_{0n}=0$.
Subsequently, using arguments similar to those in the proof of Theorem \ref{Theorem:LSE}, we obtain
\begin{align}\label{eq:se7}
\mathbb{E}(Y_{ijl}|Z_{ij})=-\omega_{0l}+\alpha_{0i}+\beta_{0j}+Z_{ij}^\top \eta_0,~~l\in[R-1],
\end{align}
where $Y_{ijl}=[\mathbb{I}(A_{ij} \ge \pi_l)-\mathbb{I}(X_{ij1}>0)]/f(X_{ij1}|Z_{ij}).$
Let $\omega=(\omega_{1},\dots,\omega_{R-1})^\top,$
$\alpha=(\alpha_{1},\dots,\alpha_{n-1})^\top$ without $\alpha_n,$
and $\beta$ and $\eta$ are defined as previously.
Furthermore, let $\omega_0,~\alpha_0,~\beta_0$ and $\eta_0$ denote true values of $\omega,~\alpha,~\beta$ and $\eta$, respectively.
Define $\widehat Y_{ijl}=[\mathbb{I}(A_{ij} \ge \pi_l)-\mathbb{I}(X_{ij1}>0)]/\widehat f(X_{ij1}|Z_{ij}).$
Let $\widetilde Y=\mathds{1}_R\otimes Y,$
$\widetilde Z=(\mathbb{I}_R \otimes \mathds{1}_{N},\mathds{1}_R\otimes Z)$, and $\widetilde U=\mathds{1}_{R} \otimes \overline U,$
where $\mathds{1}_C$ denotes a $C$-dimensional vector of ones and $\overline U$ is obtained by deleting the $i$th column of $U$.
With some abuse of notation, we define $D=I_{RN}-\widetilde{U}(\widetilde{U}^\top\widetilde{U})^{-1}\widetilde{U}^\top.$
By \eqref{eq:se7}, given the constraints $\alpha_{0n}=\beta_{0n}=0$ and $\gamma_{01}=1$,
we can show that $\widetilde Z^\top D \mathbb{E}(\widetilde Y|\widetilde Z)=\widetilde Z^\top D\widetilde Z\eta_0$.
This implies that we can estimate $\eta_0$ by $\widehat\eta=(\widetilde Z^\top D\widetilde Z)^{-1}\widetilde Z^\top D \widehat Y$.
Then, we can obtain the least squares estimators $\widehat\alpha,~\widehat\beta$ and $\widehat\omega$ for $\alpha_0,~\beta_0$ and $\omega_0$  by minimizing
\begin{align*}
\sum_{l=1}^{R-1}\sum_{i=1}^n\sum_{j\ne i}(\widehat Y_{ijl}-\alpha_{i}-\beta_{j}+\omega_l-Z_{ij}^\top\widehat \eta)^2
~~\text{subject~to}~~\alpha_{n}=0~~\text{and}~~\beta_{n}=0.
\end{align*}
Define $\widehat\sigma_{w\epsilon}^2=N^{-1}\sum_{i=1}^n\sum_{j\neq i}
(\breve Y_{ij}+\widetilde\omega-\widehat\alpha_{i}-\widehat\beta_{j}-Z_{ij}^\top \widehat\eta)^2$
and $\widehat\sigma_{wQ,l}^2=N^{-1}\sum_{i=1}^n\sum_{j\neq i}[\widehat Y_{ijl}-\widehat{\mathbb{E}}(Y_{ijl}|X_{ij1},Z_{ij1})]^2,$
where $\breve Y_{ij}=\sum_{l=1}^{R-1}\widehat Y_{ijl}/(R-1)$ and $\widetilde \omega=\sum_{l=1}^{R-1}\widehat \omega_{l}/(R-1).$
Furthermore, define $\widehat \Omega=\text{diag}\{\widehat\Omega_1^\top,\dots,\widehat\Omega_R^\top\},$
where $\widehat\Omega_l=\widehat\sigma_{wQ,l}^2\mathds{1}_{N}.$
Denote $\xi_0=(\omega_0^\top,\eta_0^\top)^\top$ as true value of $\xi$,
and let $\widehat\xi=(\widehat\omega^\top,\widehat\eta^\top)^\top$ represent its estimator.

\begin{theorem}\label{Theorem:hatGA:weight}
If the conditions in Theorem \ref{Theorem:GA:theta} hold, then we have
\begin{align*}
\sup_{x\in\mathbb{R}^{M_{1\alpha}}}\big|\mathbb{P}\big(\sqrt{n-1}\mathcal{L}_{1\alpha}(\widehat\alpha-\alpha_0)\le x\big)
-\mathbb{P}\big(\mathcal{L}_1^\alpha\widehat G_5\le x\big)\big|&=o(1)\\
\text{and}~~~\sup_{x\in\mathbb{R}^{M_{1\beta}}}\big|\mathbb{P}\big(\sqrt{n-1}\mathcal{L}_{1\beta}(\widehat\beta-\beta_0)\le x\big)
-\mathbb{P}\big(\mathcal{L}_1^\beta\widehat G_5\le x\big)\big|&=o(1),
\end{align*}
where $\widehat G_5\sim \mathcal{N}_{2n-2}(0,(n-1)\widehat\sigma_{w\epsilon}^2V^{-1})$.
In addition, if the conditions in Theorem \ref{Theorem:GA:eta} hold, then we have
\begin{align*}
\sup_{x\in\mathbb{R}^{M_2}}\big|\mathbb{P}\big(\sqrt{N}\mathcal{L}_2(\widehat\xi-\xi_0)\le x\big)-\mathbb{P}\big(\mathcal{L}_2 \widehat G_6\le x\big)\big|=o(1),
\end{align*}
where $\widehat G_6\sim \mathcal{N}_{p}(0,N(\widetilde Z^\top D\widetilde Z)^{-1}\widetilde Z^\top D \widehat\Omega D\widetilde Z(\widetilde Z^\top D\widetilde Z)^{-1})$
 and $\mathcal{L}_2$ denotes an $M_2\times (p+R-1)$ matrix.
\end{theorem}

Theorem \ref{Theorem:hatGA:weight} extends the results of Theorem \ref{Theorem:hatGA:theta} to weighted networks.
The proof of Theorem \ref{Theorem:hatGA:weight} is omitted here because it
is similar to that of Theorem \ref{Theorem:hatGA:theta}
by treating $\omega_{0l}~(1\le l\le R-1)$ as intercepts.

\section{Proofs of Theorem 1 and Corollary 1}
\label{section-identification}

\subsection{Proof of Theorem 1}
\begin{proof}
We first consider the case  $\gamma_{01}>0$.
Define $Y_{ij}^*=\alpha_{0i}+\beta_{0j}+Z_{ij}^\top\eta_0-\varepsilon_{ij}.$
Note that for any $y>0,$
we have $[I(y+x>0)-I(x>0)]=1$ if $x\in(-y,0]$ and 0 otherwise.
Thus, if $Y_{ij}^*>0$ and $\gamma_{01}>0$, by Conditions (C2) and (C3),  we have
\begin{align*}
E[Y_{ij}|Z_{ij},\varepsilon_{ij}]=&\int_{B_L}^{B_U}\frac{I(Y_{ij}^*/\gamma_{01}+x>0)-I(x>0)}{f(x|Z_{ij})}f(x|Z_{ij})dx\\
=&\int_{B_L}^{B_U}\big[I(Y_{ij}^*/\gamma_{01}+x>0)-I(x>0)\big]dx\\
=&\int_{-Y_{ij}^*/\gamma_{01}}^0dx=Y_{ij}^*/\gamma_{01}.
\end{align*}
In addition, for any $y\le 0,$
we have $[I(y+x>0)-I(x>0)]=-1$ if $x\in[0,-y)$ and 0 otherwise.
Thus, if $Y_{ij}\le 0$ and $\gamma_{01}>0$, by
Conditions (C2) and (C3),  we have
\begin{align*}
E[Y_{ij}|Z_{ij},\varepsilon_{ij}]=&\int_{B_L}^{B_U}\frac{I(Y_{ij}^*/\gamma_{01}+x>0)-I(x>0)}{f(x|Z_{ij})}f(x|Z_{ij})dx\\
=&\int_{B_L}^{B_U}\big[I(Y_{ij}^*/\gamma_{01}+x>0)-I(x>0)\big]dx\\
=&-\int_0^{-Y_{ij}^*/\gamma_{01}}dx=Y_{ij}^*/\gamma_{01}.
\end{align*}
It follows from Condition (C3) that
\begin{align*}
E[Y_{ij}|Z_{ij}]=\big(\alpha_{0i}+\beta_{0j}+Z_{ij}^\top\eta_0\big)/\gamma_{01}.
\end{align*}
The proof for the case $\gamma_{01}<0$ is similar and omitted here.
This completes the proof.
\end{proof}

\subsection{Proof of Corollary 1}

\begin{proof}
Let $\Upsilon=\{(\alpha_{1},\dots,\alpha_{n},\beta_{1},\dots,\beta_{n},\gamma_{1},\eta^\top)^\top: \beta_{n}=0~~\text{and}~~\gamma_{1}=1\}$.
By Theorem 1 and the constraints $\beta_{0n}=0$ and $\gamma_{01}=1$,
we have
\begin{align*}
Z^\top D \mathbb{E}(Y|Z)=Z^\top D (U \theta + Z\eta_0)=Z^\top D Z\eta_0.
\end{align*}
This, together with Condition (C4), implies that $\eta_0=(Z^\top DZ)^{-1}Z^\top D \mathbb{E}(Y|Z)$.
Thus, $\eta_0$ uniquely exists and is identifiable.
We next consider the identifiability of $\alpha_{0i}$ and $\beta_{0j}$.
Without loss of generality, we assume $\eta_0=0$.
Let $\alpha^*_i,\ \beta_{j}^*$ and $\gamma_{1}^*$ be some constants satisfying $\beta_{n}^*=0$, $\gamma_1^*=1$ and
$$
\alpha_{0i}+\beta_{0j}=\alpha_{i}^*+\beta_{j}^*.
$$
It suffices to show $\alpha_{i}^*=\alpha_{0i}$ for $i\in [n]$  and $\beta_{j}^*=\beta_{0j}$ for $j~ \in [n-1]$.
Since $\beta_{0n}=0$ and $\beta_n^*=0$, we have
$$
\alpha_{0i}=\alpha_{i}^*~~ \text{for}~~ i \in [n].
$$
Additionally, for each $j\in [n-1]$, we have
$$
\alpha_{0i}+\beta_{0j}=\alpha_{i}^*+\beta_{j}^*~\Rightarrow~\beta_{0j}=\beta_{j}^*.
$$
Therefore, $\alpha_{0i}$ and $\beta_{0j}$ are identifiable within the subspace $\Upsilon$.
This completes the proof.
\end{proof}

\subsection{Sufficient conditions for Condition (C4)}
\label{sec-condition-c4}
Here, we show that if $\widetilde Z_1,\dots, \widetilde Z_p$ are not collinear
and do not almost surely lie within the linear subspace spanned by the column vectors of $U$, Condition (C4) holds.

\begin{proof}
Since $D$ is an idempotent matrix,
there exists an orthogonal matrix $\Pi$ such that $D=\Pi\Lambda \Pi^\top$,
where $\Lambda$ is a diagonal matrix with diagonal elements equal to either 1 or 0.
Furthermore, the number of ones is $r=n(n-2)$, which is larger than $p$.
Then, we can write $Z^\top DZ$ as
\begin{align*}
Z^\top DZ=Z^\top \Pi\Lambda \Pi^\top Z=
Z^\top(\Pi_1,\Pi_2)
\begin{pmatrix}
I_{n(n-2)} &0\\
0 & 0
\end{pmatrix}
\begin{pmatrix}
\Pi_1^\top\\
\Pi_2^\top
\end{pmatrix}Z
=(\Pi_1^\top Z)^\top (\Pi_1^\top Z),
\end{align*}
where $\Pi_1\in \mathbb{R}^{N\times r}$ and $\Pi_2\in \mathbb{R}^{N\times n}$ represent the sub-matrices of $\Pi$.
Note that $\Pi$ is an orthogonal matrix.
This, along with the fact that $r=n(n-2)>p$, implies that
$$
\text{rank}(Z^\top DZ)=\text{rank}(\Pi_1^\top Z)=\text{rank}(Z)=p.
$$
Therefore, Condition (C4) holds.
\end{proof}
\section{Proofs of consistency and Gaussian approximation}
\label{section-proof-theorem}

This section is organized as follows.
Sections \ref{subsec-pro-th2}-\ref{subsec-pro-th5} present
the proofs of Theorems 2-5, respectively.
Sections \ref{subse-pro-pro1} and \ref{subsec-pro-pro2}
present the proofs of Propositions 1 and 2, respectively.
Section \ref{subsec-pro-th6} presents the proof of Theorem 6.

Define $F(\theta,\eta)=U^\top(Y-Z\eta)-V\theta$.
For a given $\eta$, write
\[
F_{\eta}(\theta)=(F_{1\eta}(\theta),\dots,F_{(2n-1)\eta}(\theta))^\top=F(\theta,\eta).
\]
Define $\theta_{0\eta}=(\alpha_{01,\eta},\dots,\alpha_{0n,\eta},
\beta_{01,\eta},\dots,\beta_{0,n-1,\eta})^\top$,
where $\alpha_{0i,\eta}$ and $\beta_{0j,\eta}$ are the solution to $\mathbb{E}\{F_{\eta}(\theta)\}=0$ with a given $\eta$.

By the definition of $U$, a direct calculation yields $V=(v_{ij})\in \mathbb{R}^{(2n-1)\times (2n-1)}$ with
\begin{equation}\label{V-explicit}
\begin{array}{l}
v_{ij}=v_{ji}=0, ~~ i,~j=1,\dots,n,~i\neq j;\quad
v_{ij}=v_{ji}=0, ~~ i,~j=n+1,\dots,2n-1, ~i\neq j,\\
v_{(n+i)i}=v_{i(n+i)}=0, ~~ i,~j=1,\dots,n-1; \quad
v_{ii}-\sum_{j=n+1}^{2n-1} v_{ij}=1,~~ i=1,\dots,n,\\
v_{ii}=\sum_{k=1}^{n}v_{ki}=\sum_{k=1}^{n}v_{ik}, ~~i=n+1,\dots,2n-1,\\
v_{ij}=v_{ji}=1, ~~ i=1,\ldots,n,~ j=n+1,\dots,2n-1;~ j\neq n+i.
\end{array}
\end{equation}
The inverse of $V^{-1}$ has a close form, where the $(i,j)$th element of $V^{-1}$ is
\begin{align}
\label{eq-inverse-V}
(V^{-1})_{ij}=
\begin{cases}
\frac{2n-1}{n(n-1)},&\ \ \ 1\le i=j \le n-1,\\
\frac{n^2-3n+1}{n(n-1)(n-2)},&\ \ \ 1\le i\neq j\le n-1,\\
\frac{2n-3}{(n-1)(n-2)},&\ \ \ i=n,~j=n,\\
\frac{1}{n-1},&\ \ \ i=n,~1\le j\le n-1~~\text{or}~~j=n,~1\le i\le n-1,\\
-\frac{1}{n},&\ \ \ 1\le i\le n-1,~j=n+i~~\text{or}~~n+1\le i\le 2n-1,~j=i-n\\
-\frac{1}{n-2},&\ \ \ n+1\le i\le 2n-1,~j=n~~\text{or}~~i=n,~n+1\le j\le 2n-1,\\
-\frac{n-1}{n(n-2)},&\ \ \ n+1\le i\le 2n-1,~1\le j\le n-1~~\text{or}~~1\le i\le n-1,~n+1\le j\le 2n-1,\\
\frac{2(n-1)}{n(n-2)},&\ \ \ n+1\le i=j\le 2n-1,\\
\frac{n-1}{n(n-2)},&\ \ \ n+1\le i \neq j\le 2n-1.
\end{cases}
\end{align}
To simplify calculations, we consider an approximation to $V^{-1}$, denoted by $\breve{S}=(\breve{s}_{ij}),$
where $\breve{s}_{ij}$ is defined as follows:
\begin{align*}
\breve s_{ij}=\left\{
\begin{array}{ll}
\frac{\delta_{ij}}{n-1}+\frac{1}{n-1},&\ \ \ i,j=1,\dots,n,\\
-\frac{1}{n-1},&\ \ \ i=1,\dots,n;j=n+1,\dots,2n-1,\\
-\frac{1}{n-1},&\ \ \ i=n+1,\dots,2n-1;j=1,\dots,n,\\
\frac{\delta_{ij}}{n-1}+\frac{1}{n-1},&\ \ \ i,j=n+1,\dots,2n-1.
\end{array}\right.
\end{align*}
Here, $\delta_{ij}=1$ if $i=j$ and 0 otherwise.

\begin{lemma}\label{lem:Vbound}
Under Conditions (C1)-(C7), we have
\[
\|V^{-1}-\breve{S}\|_{\max}\le C_0/(n-1)^2,
\]
where $C_0>0$ is some constant.
\end{lemma}
\begin{proof}
The proof immediately follows from the explicit expressions of $V^{-1}$ in
\eqref{eq-inverse-V}.
\end{proof}

\subsection{Proofs for Theorem 2}
\label{subsec-pro-th2}

We state two lemmas before proving Theorem 2.

\begin{lemma}\label{Lem:UpperBound}
Under Conditions (C1)-(C5), we have
\begin{align*}
\sup_{\eta\in\Theta}\|F_\eta(\theta_{0\eta})\|_{\infty}\le O_p\Big(\frac{(\kappa+q_n)}{m}\sqrt{n\log(n)}\Big).
\end{align*}
\end{lemma}

\begin{proof}
Let $\{\eta_k: k\in [J_n]\}$ be the equal grid points on $\Theta$
and $\eta_k^0$ be an interior point of $[\eta_k,\eta_{k+1})$.
Then, it can be shown that
\begin{align}\label{Lem:UpperBound:eq1}
&P\Big(\sup_{\eta \in \Theta}\max_{i\in [2n-1]}|F_{i\eta}(\theta_{0\eta})| >\delta \Big)\nonumber\\
=& P\Big(\bigcup_{k\in [J_n]}\sup_{\eta \in [\eta_{k},\eta_{k+1})}\max_{i\in [2n-1]}|F_{i\eta}(\theta_{0\eta})| > \delta\Big)\nonumber\\
\le& \sum_{k\in [J_n]}P\Big(\sup_{\eta \in [\eta_{k},\eta_{k+1})}\max_{i\in [2n-1]}|F_{i\eta}(\theta_{0\eta})| > \delta\Big)\nonumber\\
\le& \sum_{k\in [J_n]}P\Big(\sup_{\eta \in [\eta_{k},\eta_{k+1})}\max_{i\in [2n-1]}|F_{i\eta}(\theta_{0\eta_{k}^*})-F_{i\eta}(\theta_{0\eta})|
+\sup_{\eta \in [\eta_{k},\eta_{k+1})}\max_{i\in [2n-1]}|F_{i\eta}(\theta_{0\eta_k^*})| > \delta\Big)\nonumber\\
\le& \sum_{k\in [J_n]}P\Big(\sup_{\eta \in [\eta_{k},\eta_{k+1})}\max_{i\in [2n-1]}|F_{i\eta}(\theta_{0\eta_k^*})-F_{i\eta}(\theta_{0\eta})| > \frac{\delta}{2}\Big)
+\sum_{k\in [J_n]}P\Big(\|F_{i\eta}(\theta_{0\eta_k^*})\|_{\infty} > \frac{\delta}{2}\Big)\nonumber\\
:=&B_1+B_2.
\end{align}
Because $\theta_{0\eta}$ is Lipschitz continuous with respect to $\eta$,
we have
$$
|F_{i\eta}(\theta_{0\eta_k^*})-F_{i\eta}(\theta_{0\eta})|\le C_1n\|\eta_k^*-\eta\|_2 <\delta/2
$$
for a sufficiently large $J_n$. This leads to $B_1=0$.

We now bound $B_2$. A direct calculation yields that
\begin{align*}
 F_{i\eta}(\theta_{\eta})=&\sum_{j\ne i}^{n}\big[Y_{ij}-(\alpha_{i,\eta}+\beta_{j,\eta}+Z_{ij}^\top\eta)\big],\ \ \ i\in[n],\nonumber\\
 F_{(n+j)\eta}(\theta_{\eta})=&\sum_{i\ne j}^{n}\big[Y_{ij}-(\alpha_{i,\eta}+\beta_{j,\eta}+Z_{ij}^\top\eta)\big],\ \ \ j\in[n-1].
\end{align*}
By Hoeffding's inequality (Hoeffding, 1963) and choosing a sufficiently large $C_2>0$,
we have that for each given $\eta\in\Theta$,
\begin{align*}
B_2=& \sum_{k\in [J_n]}P\Big(\|F_{i\eta}(\theta_{0\eta_k^*})\|_{\infty} > \frac{C_2(\kappa+q_n)}{m}\sqrt{(n-1)\log(4n-2)}\Big)
\le n^{-c_0},
\end{align*}
where $c_0$ is some constant depending on $C_2$.
This, together with \eqref{Lem:UpperBound:eq1}, completes the proof.
\end{proof}

\begin{lemma}\label{Lem:UpperBound_fix_eta}
If Conditions (C1)-(C5) hold, then
\begin{align*}
\sup_{\eta\in\Theta}\|\widehat{\theta}_{\eta}-\theta_{0\eta}\|_\infty=O_p\bigg(\frac{(\kappa+q_n)}{m}\sqrt{\frac{\log(n)}{n}}+\sqrt{\frac{\log(n)}{m^4n^2h^{2p+2}}}+\frac{h^{r}}{m^2}\bigg).
\end{align*}
\end{lemma}

\begin{proof}
For simplicity of presentation, in what follows we assume $p_1=p$.
For each given $\eta\in\Theta,$ we have
\begin{align*}
\widehat{\theta}_\eta=\text{arg}\min_{\theta}\|\widehat Y-U\theta-Z\eta\|_2^2
=V^{-1}U^\top(\widehat Y-Z\eta).
\end{align*}
Note that
\begin{align*}
\widehat{\theta}_{\eta}-\theta_{0\eta}=&V^{-1}[U^\top(\widehat Y-Z\eta)-U^\top U\theta_{0\eta}]\\
=&V^{-1}F_\eta(\theta_{0\eta})+V^{-1}U^\top(\widehat Y-Y).
\end{align*}
It suffices to show
\begin{align}\label{Lem:UpperBound_fix_eta:eq1}
&\|V^{-1}F_\eta(\theta_{0\eta})\|_{\infty}=O_p\bigg(\frac{(\kappa+q_n)}{m}\sqrt{\frac{\log(4n-2)}{n-1}}\bigg),  \ \ \ \mbox{and}\\
\label{Lem:UpperBound_fix_eta:eq2}
&\|V^{-1}U^\top(\widehat Y-Y)\|_{\infty}=O_p\bigg(\sqrt{\frac{\log(n)}{m^4n^2h^{2p+2}}}+\frac{h^{r}}{m^2}\bigg).
\end{align}
To show \eqref{Lem:UpperBound_fix_eta:eq1}, we begin with the following decomposition
\begin{align}\label{Lem:UpperBound_fix_eta:eq3}
V^{-1}F_{\eta}(\theta_{0\eta})=
\underbrace{(V^{-1}-\breve S) F_{\eta}(\theta_{0\eta})}_{B_3} +
\underbrace{\breve S F_{\eta}(\theta_{0\eta})}_{B_4}.
\end{align}
By Lemma \ref{lem:Vbound}, we have
\begin{align}
\label{Lem:UpperBound_fix_eta:eq4}
\|B_3\|_{\infty} \le & (2n-1)\|V^{-1}-\breve S\|_{\max}\|F_{\eta}(\theta_{0\eta})\|_\infty\nonumber\\
=& \frac{C_3(2n-1)}{(n-1)^2}\|F_{\eta}(\theta_{0\eta})\|_\infty\nonumber\\
\le & \frac{C_3(\kappa+q_n)(2n-1)}{m(n-1)^2}\sqrt{(n-1)\log(4n-2)}.
\end{align}
According to the definition of $\breve S$, a direct calculation yields
\begin{align}\label{Lem:UpperBound_fix_eta:eq5}
\|B_4\|_{\infty}=\|\breve S F_{\eta}(\theta_{0\eta})\|_\infty
\le \frac{C_4(\kappa+q_n)}{m(n-1)}\sqrt{(n-1)\log(4n-2)}.
\end{align}
By \eqref{Lem:UpperBound_fix_eta:eq3}, \eqref{Lem:UpperBound_fix_eta:eq4}, \eqref{Lem:UpperBound_fix_eta:eq5} and Lemma \ref{Lem:UpperBound}, we obtain
\begin{align*}
\sup_{\eta\in\Theta}\|V^{-1}F_{\eta}(\theta_{0\eta})\|_{\infty}\le \Big[\frac{C_3(2n-1)}{(n-1)^2}+\frac{C_4}{n-1}\Big]\frac{\kappa+q_n}{m}\sqrt{(n-1)\log(4n-2)},
\end{align*}
which implies that \eqref{Lem:UpperBound_fix_eta:eq1} holds.

We next  show \eqref{Lem:UpperBound_fix_eta:eq2}.
For this, we write
$\widehat f(x|Z=z)=\widehat f_{XZ}(x,z)/\widehat f_Z(z),$ where
\begin{align*}
\widehat f_{XZ}(x,z)=&\frac{1}{n(n-1)}\sum_{1\le s\neq k\le n}\mathcal{K}_{xz,h}\big(X_{sk1}-x,Z_{sk}-z\big),\\
\widehat f_Z(z)=&\frac{1}{n(n-1)}\sum_{1\le s\neq k\le n}\mathcal{K}_{z,h}\big(Z_{sk}-z\big).
\end{align*}
For $i\in[n]$, a direct calculation gives that
\begin{align*}
|(U^\top (\widehat Y-Y))_i|=& \big|\sum_{j\neq i}(\widehat{Y}_{ij}-Y_{ij})\big|
\le  n\max_{j\neq i}\bigg|\frac{\widehat f_{Z}(Z_{ij})}{\widehat f_{XZ}(X_{ij},Z_{ij})}-\frac{f_{Z}(Z_{ij})}{f_{XZ}(X_{ij},Z_{ij})}\bigg|\\
\le & (n/m^2)O_p\bigg(\sqrt{\frac{\log(n)}{n^2h^{2p+2}}}+h^{r}\bigg),
\end{align*}
where the last inequality follows from the uniform convergence of the kernel estimators (e.g., Andrews, 1995).
With similar arguments as in the proof of \eqref{Lem:UpperBound_fix_eta:eq3}, we have
\[
\|V^{-1}U^\top(\widehat Y-Y)\|_{\infty}\le O_p\bigg(\sqrt{\frac{\log(n)}{m^4n^2h^{2p+2}}}+\frac{h^{r}}{m^2}\bigg),
\]
which implies \eqref{Lem:UpperBound_fix_eta:eq2}. It completes the proof.
\end{proof}

We are now ready to prove Theorem 2.

\begin{proof}[Proof of Theorem 2]
Define $\widehat Q_c(\eta)=-Z^\top(\widehat Y-U\widehat\theta_{\eta}-Z\eta)$,
$\widetilde Q_c(\eta)=-Z^\top(\widehat Y-U\theta_{0\eta}-Z\eta)$,
and $Q_c(\eta)=-Z^\top(Y-U\theta_{0\eta}-Z\eta)$.
Then, we have
\begin{align}\label{Theorem1:eq1}
\|\widehat Q_c(\eta)-Q_c(\eta)\|_2\le &\|\widehat Q_c(\eta)-\widetilde Q_c(\eta)\|_2+\|\widetilde Q_c(\eta)-Q_c(\eta)\|_2\nonumber\\
\le &\|Z^\top U(\widehat\theta_{\eta}-\theta_{0\eta})\|_2+\|Z^\top (\widehat Y-Y)\|_2.
\end{align}
For the first term on the right hand side of \eqref{Theorem1:eq1}, by Lemma \ref{Lem:UpperBound_fix_eta},
we have
\begin{align}\label{Theorem:eq11}
\|Z^\top U(\widehat\theta_{\eta}-\theta_{0\eta})\|_2
\le & \Big(\sum_{k=1}^p\Big(\sum_{i=1}^{2n-1}|\sum_{j\neq i}Z_{ij,k}|\Big)^2\Big)^{1/2} \|\widehat\theta_{\eta}-\theta_{0\eta}\|_{\infty}
\nonumber\\
\le &  O_p\bigg(\frac{(\kappa+q_n)}{m}\sqrt{n^3\log(n)}+\sqrt{\frac{n^2\log(n)}{m^4h^{2p+2}}}+\frac{n^2h^{r}}{m^2}\bigg).
\end{align}
For the second term on the right hand side of \eqref{Theorem1:eq1},
under Conditions (C5)-(C6), an application of Theorem 1 in Andrews (1995) yields
\begin{align}\label{Theorem:eq12}
& Z^\top (\widehat Y-Y)\nonumber\\
=&\sum_{1\le i\neq j\le n} Z_{ij}\Big[A_{ij}-\mathbb{I}(X_{ij1}>0)\Big]\Bigg[\frac{\widehat{f}(Z_{ij})}{\widehat f_{XZ}(X_{ij1},Z_{ij})}
-\frac{f_Z(Z_{ij})}{ f_{XZ}(X_{ij1},Z_{ij})}\Bigg]\nonumber\\
=&\sum_{1\le i\neq j\le n} \frac{1}{n(n-1)}\sum_{1\le s\neq k\le n} Z_{sk}Y_{sk}\Bigg[\frac{f(X_{sk1}|Z_{sk})\mathcal{K}_{z,h}\big(Z_{ij}-Z_{sk}\big)
-\mathcal{K}_{xz,h}\big(X_{ij1}-X_{sk1},Z_{ij}-Z_{sk}\big)
}{\widehat f_{XZ}(X_{sk1},Z_{sk})}\Bigg]\nonumber\\
=& \sum_{1\le i\neq j\le n} Z_{ij}\big[\mathbb{E}(Y_{ij}|Z_{ij})-\mathbb{E}(Y_{ij}|X_{ij1},Z_{ij})\big]+O_p\bigg(\frac{\kappa}{m}\Big(n^2h^{r}+\sqrt{\frac{\log n}{h^{2p+2}}}\Big)\bigg).
\end{align}
Note that $\|Z_{ij}[\mathbb{E}(Y_{ij}|Z_{ij})-\mathbb{E}(Y_{ij}|X_{ij1},Z_{ij})]\|_2\le O(\kappa/m)$ under Conditions (C4) and (C5).
Hoeffding's inequality (Hoeffding, 1963) implies that
\begin{align}\label{Theorem:eq13}
 Z^\top (\widehat Y-Y)=O_p\bigg(\frac{n\kappa}{m}\bigg).
\end{align}
Combining \eqref{Theorem1:eq1}-\eqref{Theorem:eq13}, we obtain
\begin{align}\label{Theorem1:Bound}
&\|\widehat Q_c(\eta)-Q_c(\eta)\|_2\nonumber\\
\le &O_p\bigg(\frac{\kappa}{m}\Big[(\kappa+q_n)\sqrt{n^3\log(n)}+\sqrt{\frac{n^2\log(n)}{m^2h^{2p+2}}}+\frac{n^2h^{r}}{m}+n+n^2h^{r}\Big]\bigg)\nonumber\\
\le &O_p\bigg(\frac{\kappa}{m^2}\Big[(\kappa+q_n)\sqrt{n^3\log(n)}+\sqrt{\frac{n^2\log(n)}{h^{2p+2}}}+n+n^2h^{r}\Big]\bigg).
\end{align}
In addition, by Hoeffding's inequality, we have $Q_c(\eta)=\mathbb{E}[Q_c(\eta)]+O_p(\kappa n/m)$.
This, together with \eqref{Theorem1:Bound}, gives
\begin{align}
\nonumber
\|\widehat Q_c(\eta)-\mathbb{E}[Q_c(\eta)]\|_2
\le& \|\widehat Q_c(\eta)-Q_c(\eta)\|_2+\|Q_c(\eta)-\mathbb{E}[Q_c(\eta)]\|_2\\
\label{eqA6}
=&O_p\bigg(\frac{\kappa}{m^2}\Big[(\kappa+q_n)\sqrt{n^3\log(n)}+\sqrt{\frac{n^2\log(n)}{h^{2p+2}}}+n+n^2h^{r}\Big]\bigg).
\end{align}

Let $\vartheta_0$ be any positive constant such that $\|\eta-\eta_0\|\le \vartheta_0$.
For any $w_2\in \mathbb{R}^{p}$ satisfying $\|w\|=1$,
$w^\top \mathbb{E}[Q_c(\eta_0+w\vartheta)]$ increases with $\vartheta$.
It implies that for any $\vartheta \ge \vartheta_0>0$,
$w^\top\{\mathbb{E}[Q_c(\eta_0+w\vartheta)]-\mathbb{E}[Q_c(\eta_0)]\}\ge 0$.
Then, we have
\begin{align*}
 & \|w\|_2\|\mathbb{E}[Q_c(\eta_0+w\vartheta)]-\mathbb{E}[Q_c(\eta_0)]\|_2 \\
\ge&|w^\top\{\mathbb{E}[Q_c(\eta_0+w\vartheta)]-\mathbb{E}[Q_c(\eta_0)]\}|\\
\ge&w^\top (ZDZ^\top) w \vartheta>p\phi\vartheta_0N,
\end{align*}
where $\phi>0$ is some constant.
Here, the first inequality is due to the Cauchy-Schwarz inequality,
and the last one follows from Condition (C4).
Therefore,
\begin{align}\label{eqA7}
\inf_{\|\eta-\eta_0\|>\vartheta_0}\|\mathbb{E}[Q_c(\eta)]-\mathbb{E}[Q_c(\eta_0)]\|_2>p\phi \vartheta_0N.
\end{align}
Note that $\widehat Q_c(\widehat\eta)=0$ almost surely and $\mathbb{E}[Q_c(\eta_0)]=0$.
It then follows from \eqref{eqA6} that
\begin{align}
\nonumber
\|\mathbb{E}[Q_c(\widehat\eta)]-\mathbb{E}[Q_c(\eta_0)]\|_2=&\|\widehat Q_c(\widehat\eta)-\{\widehat Q_c(\widehat \eta)-\mathbb{E}[Q_c(\widehat\eta)]\}\|_2\\
=&O_p\bigg(\frac{\kappa}{m^2}\Big[(\kappa+q_n)\sqrt{n^3\log(n)}+\sqrt{\frac{n^2\log(n)}{h^{2p+2}}}+n+n^2h^{r}\Big]\bigg), \label{eqA8}
\end{align}
and for sufficiently large $n$, $\|\mathbb{E}[Q_c(\widehat\eta)]\|_2<\phi n^2 \vartheta_0/2$.
Therefore, by \eqref{eqA7}, we obtain $\|\widehat\eta -\eta_0\|_2<\vartheta_0$ with probability tending to one.
By Taylor's expansion of $U(\widehat\eta)$ at $\eta_0$, we obtain
\begin{align}
\nonumber
\|\widehat\eta-\eta_0\|_2
=&\|\mathbb{E}(ZDZ^\top)^{-1}\{\mathbb{E}[Q_c(\widehat\eta)]-\mathbb{E}[Q_c(\eta_0)]\}\|_2 \nonumber \\
\le & \|\mathbb{E}(ZDZ^\top)^{-1}\|_2\|\mathbb{E}[Q_c(\widehat\eta)]-\mathbb{E}[Q_c(\eta_0)]\|_2.\label{eqA9}
\end{align}
By Condition (C4), there exists a positive constant $\phi$ such that
$\phi_{\min}(\mathbb{E}(ZDZ^\top))>\phi.$
Thus, it follows from \eqref{eqA8} and \eqref{eqA9} that
$$
\|\widehat\eta-\eta_0\|_2\le O_p\bigg(\frac{\kappa}{\phi m^2}\bigg[(\kappa+q_n)\sqrt{\frac{\log(n)}{n}}+\sqrt{\frac{\log(n)}{n^2h^{2p+2}}}+\frac{1}{n}+h^{r}\bigg]\bigg).
$$

We now show the second part of Theorem 2.
Note that $\widehat\theta=\widehat\theta_{\widehat\eta}$
and $\theta_0=\theta_{0,\eta_0}.$
By the first part of Theorem 2, it suffices to show that
\begin{align}\label{Theorem4:eq1}
\sup_{\eta:\|\eta-\eta_0\|\le \delta_n}\|\widehat\theta_{\eta}-\theta_{0}\|_{\infty}=O_p\bigg(\frac{\kappa}{\phi m^2}\bigg[(\kappa+q_n)\sqrt{\frac{\log(n)}{n}}+\sqrt{\frac{\log(n)}{n^2h^{2p+2}}}+h^{r}\bigg]\bigg),
\end{align}
where $\delta_n$ is a sequence of constants determined by Theorem 2.
The proof of \eqref{Theorem4:eq1} is similar to that of Lemma \ref{Lem:UpperBound_fix_eta}.
We only show the
differences here.
Note that
\begin{align*}
\hat{\theta}_{\eta}-\theta_{0}=&V^{-1}U^\top[(\widehat Y-Z\eta)-U\theta_{0}]\\
=&V^{-1}U^\top Z(\eta-\eta_0)+V^{-1}F_{\eta_0}(\theta_0)+V^{-1}U^\top(\widehat Y-Y).
\end{align*}
With similar arguments as in the proof of \eqref{Lem:UpperBound_fix_eta:eq3}, we have
\begin{align}
\|V^{-1}U^\top Z(\eta-\eta_0)\|_{\infty}\le O_p\big(\kappa\delta_n\big),
\end{align}
which, together with \eqref{Lem:UpperBound_fix_eta:eq1} and  \eqref{Lem:UpperBound_fix_eta:eq2},
implies that \eqref{Theorem4:eq1} holds by choosing
\[
\delta_n=\frac{\kappa}{\phi m^2}\bigg[(\kappa+q_n)\sqrt{\frac{\log(n)}{n}}+\sqrt{\frac{\log(n)}{n^2h^{2p+2}}}+\frac{1}{n}+h^{r}\bigg].
\]
This completes the proof.
\end{proof}

\subsection{Proofs for Theorem 3}
\label{subsec-pro-th3}

Before beginning to prove Theorem 3, we present five lemmas.

\begin{lemma}
\label{Lemma4}
Under Conditions (C1)-(C7), we have
\[
\widehat\theta-\theta_0=V^{-1}U^\top(Y-Z\eta_0-U\theta_0)+O_p\bigg(\frac{\kappa^2}{\phi m^2}\bigg[\frac{1}{n}+(\kappa+q_n)\sqrt{\frac{\log(n)}{n}}+\sqrt{\frac{\log(n)}{n^2h^{2p+2}}}+h^{r}\bigg]\bigg) .
\]
\end{lemma}

\begin{proof}
We begin with the following decomposition
\begin{align*}
\widehat{\theta}-\theta_0=&V^{-1}U^\top(\widehat Y-Z\widehat{\eta})-\theta_0\\
=& \underbrace{V^{-1}U^\top(Y-Z\eta_0-U\theta_0)}_{B_5}
  + \underbrace{ V^{-1}U^\top(\widehat Y-Y)}_{B_6} -
  \underbrace{ V^{-1}U^\top Z(\widehat{\eta}-\eta_0)}_{B_7}.
\end{align*}
It now suffices to show
\begin{align}
    B_6=&O_p\bigg(\frac{1}{m^2}\Big(h^r+\sqrt{\frac{\log(n)}{n^2h^{2p+2}}}\Big)\bigg), \ \ \ \mbox{and}\label{eqB6}\\
    B_7=&O_p\bigg(\frac{\kappa^2}{\phi m^2}\bigg[\frac{1}{n}+(\kappa+q_n)\sqrt{\frac{\log(n)}{n}}+\sqrt{\frac{\log(n)}{n^2h^{2p+2}}}+h^{r}\bigg]\bigg).  \label{eqB7}
\end{align}
By \eqref{Lem:UpperBound_fix_eta:eq2}, we have \eqref{eqB6} holds.

We next show \eqref{eqB7}. Note that
\begin{align*}
B_7=&V^{-1}U^\top Z(\widehat{\eta}-\eta_0)\\
=& \underbrace{\breve SU^\top Z(\widehat{\eta}-\eta_0) }_{ B_{71}} +
   \underbrace{ (V^{-1}-\breve S)U^\top Z(\widehat{\eta}-\eta_0) }_{ B_{72} }.
\end{align*}
Let $B_{71,i}$ as the $i$th element of $B_{71}$. Then, we have
\begin{align*}
B_{71,i}=&\frac{1}{n-1}\sum_{j\ne i}^nZ_{ij}^\top(\widehat{\eta}-\eta_0)+\frac{1}{n-1}\sum_{k\ne n}Z_{kn}^\top(\widehat{\eta}-\eta_0),~~i\in[n],\\
B_{71,n+j}=&\frac{1}{n-1}\sum_{i\ne j}^nZ_{ij}^\top(\widehat{\eta}-\eta_0)-\frac{1}{n-1}\sum_{k\ne n}Z_{kn}^\top(\widehat{\eta}-\eta_0),~~j\in[n-1],
\end{align*}
which, together with Condition (C5) and Theorem 2, implies
\begin{align*}
\|B_{71}\|_{\infty}
\le & 2\kappa p\|\widehat{\eta}-\eta_0\|_\infty\le O_p\bigg(\frac{\kappa^2}{\phi m^2}\bigg[\frac{1}{n}+(\kappa+q_n)\sqrt{\frac{\log(n)}{n}}+\sqrt{\frac{\log(n)}{n^2h^{2p+2}}}+h^{r}\bigg]\bigg).
\end{align*}
We now bound $B_{72}$.
In view of Lemma \ref{lem:Vbound} and Theorem 2, we have
\begin{align*}
\|B_{72}\|_{\infty}\le & \frac{C_{4}(2n-1)}{n-1}\kappa p\|\widehat{\eta}-\eta_0\|_\infty\\
\le & O_p\bigg(\frac{\kappa^2}{\phi m^2}\bigg[\frac{1}{n}+(\kappa+q_n)\sqrt{\frac{\log(n)}{n}}+\sqrt{\frac{\log(n)}{n^2h^{2p+2}}}+h^{r}\bigg]\bigg),
\end{align*}
where $C_4$ is some constant.
Combing the results of $B_{71}$ and $B_{72,}$ we have \eqref{eqB7} holds.
This completes the proof.
\end{proof}

\begin{lemma}\label{Lemma:linear:eta}
Under Conditions (C1)-(C7), we have
\[
\widehat\eta-\eta_0=(Z^\top DZ)^{-1}Z^\top DQ+O_p\bigg(\frac{\kappa}{\phi m}\Big(h^{r}+\sqrt{\frac{\log n}{n^4h^{2p+2}}}\Big)\bigg),
\]
where $Q=(Q_1^\top,\dots,Q_{n}^\top)^\top,$ $Q_i=(Q_{i1},\dots,Q_{i,j-1},Q_{i,j+1},\dots,Q_{in})^\top$ and
$Q_{ij}=Y_{ij}-\mathbb{E}(Y_{ij}|X_{ij1}, Z_{ij}).$
\end{lemma}
\begin{proof}
Note that $(Z^\top DZ)^{-1}Z^\top D U=0.$ Thus, we have
\begin{align}\label{lem4:eq1}
\widehat\eta-\eta_0=&(Z^\top DZ)^{-1}Z^\top D(Y-U\theta_0-Z\eta_0+\widehat{Y}-Y)\nonumber\\
=& (Z^\top DZ)^{-1}Z^\top D(Y-U\theta_0-Z\eta_0)+(Z^\top DZ)^{-1}Z^\top D(\widehat{Y}-Y).
\end{align}
With similar arguments as in the proof of \eqref{Theorem:eq12}, we have
\begin{align*}
(Z^\top DZ)^{-1}Z^\top D(\widehat{Y}-Y)=(Z^\top DZ)^{-1}Z^\top D\Delta+O_p\bigg(\frac{\kappa}{\phi m}\Big(h^{r}+\sqrt{\frac{\log n}{n^4h^{2p+2}}}\Big)\bigg),
\end{align*}
where $\Delta=(\Delta_1^\top,\dots,\Delta_{n}^\top)^\top$, $\Delta_i=(\Delta_{i1},\dots,\Delta_{i,j-1},\Delta_{i,j+1},\dots,\Delta_{in})^\top$ and
$\Delta_{ij}=\mathbb{E}(Y_{ij}|Z_{ij})-\mathbb{E}(Y_{ij}|X_{ij}, Z_{ij}).$
By \eqref{lem4:eq1} and the fact that $\mathbb{E}(Y_{ij}|Z_{ij})=\alpha_{0i}+\beta_{0j}+Z_{ij}^\top\eta_0,$  we obtain
\begin{align*}
\widehat\eta-\eta_0=&(Z^\top DZ)^{-1}Z^\top D(Y-U\theta_0-Z\eta_0+\Delta)+O_p\bigg(\frac{\kappa}{\phi m}\Big(h^{r}+\sqrt{\frac{\log n}{n^4h^{2p+2}}}\Big)\bigg)\\
=& (Z^\top DZ)^{-1}Z^\top DQ+O_p\bigg(\frac{\kappa}{\phi m}\Big(h^{r}+\sqrt{\frac{\log n}{n^4h^{2p+2}}}\Big)\bigg).
\end{align*}
This completes the proof.
\end{proof}

\begin{lemma}[Nazarov's inequality]
\label{lem:Nazarov-inequality}
Let $Y=(Y_1,\dots, Y_d)^\top$ be a centered Gaussian vector in
$\mathbb{R}^d$ such that $\mathbb{E}[Y_j^2]\ge \underline{\sigma}^2$ for all $j\in[d]$ and some constant $\underline{\sigma}>0$.
Then for every $y\in\mathbb{R}^d$ and $t>0,$
\[
|\mathbb{P}(Y\le y+t)-\mathbb{P}(Y\le y)|\le (t/\underline{\sigma})(\sqrt{2\log(d)}+2).
\]
\end{lemma}
\begin{proof}
See Lemma A.1 in Chernozhukov et al. (2017).
\end{proof}

\begin{lemma}[Gaussian-to-Gaussian Comparison]\label{lem:GGC}
If $Z_1$ and $Z_2$ are centered
Gaussian random vectors in $\mathbb{R}^d$ with covariance matrices $\Sigma_1$ and $\Sigma_2,$ respectively,
and $\Sigma_2$  is such that $\Sigma_{2,jj}\ge c$ for all $j\in [d]$ for some constant $c > 0,$ then
\begin{align*}
\sup_{y\in\mathbb{R}^d} |\mathbb{P}(Z_1 \le y)-\mathbb{P}(Z_2 \le y)| \le C(\Psi \log^2 d)^{1/2},
\end{align*}
where $C$ is a constant depending only on $c$ and $\Psi=\|\Sigma_{1}-\Sigma_{2}\|_{\max}.$
\end{lemma}
\begin{proof}
See Proposition 2.1 in Chernozhukov et al. (2022).
\end{proof}

\begin{lemma}\label{lem:GA}
Suppose Conditions (C1)-(C7) hold.
If there exist some constant $\sigma_L^2$ and $\sigma_U^2$ such that $\sigma_L^2<\mathbb{E}(\epsilon_{ij}^2)<\sigma_U^2$ for all $i,j\in [n],$
and $(1/m+q_n+\kappa\big)^2\log^5(n)/n=o(1),$
then we have
\begin{align*}
\sup_{x\in \mathbb{R}^{M_{1\alpha}}}\big|\mathbb{P}\big(\mathcal{L}_1(\alpha)\Gamma U^\top\epsilon/\sqrt{n-1}\le x\big)-\mathbb{P}\big(\mathcal{L}_1(\alpha)G_1\le x\big)\big|&=o(1),\\
\sup_{x\in \mathbb{R}^{M_{1\beta}}}\big|\mathbb{P}\big(\mathcal{L}_1(\beta)\Gamma U^\top\epsilon/\sqrt{n-1}\le x\big)-\mathbb{P}\big(\mathcal{L}_1(\beta)G_1\le x\big)\big|&=o(1),
\end{align*}
where $\Gamma=(V/(n-1))^{-1}$ and $\epsilon=Y-Z\eta_0-U\theta_0.$
\end{lemma}

\begin{proof}
We show the first part and the second part can be obtained similarly.
For simplicity, we let $\mathcal{L}_1=\mathcal{L}_1(\alpha)$
and $L_{1,sk}$ be the $(s,k)$th element of $\mathcal{L}_{1}$.
Recall that $\epsilon_{ij}=Y_{ij}-Z_{ij}^\top\eta_0-\alpha_{0i}-\beta_{0j}$.
Let $\widetilde S=(n-1)\breve S.$
 We begin with the following decomposition
\begin{align*}
&\big|\mathbb{P}\big(\mathcal{L}_1\Gamma U^\top\epsilon/\sqrt{n-1}\le x\big)-\mathbb{P}\big(\mathcal{L}_1G_1\le x\big)\big|\\
\le & \underbrace{ \big|\mathbb{P}\big(\mathcal{L}_1\Gamma U^\top\epsilon/\sqrt{n-1}\le x\big)-\mathbb{P}\big(\mathcal{L}_1\widetilde G_1\le x\big)\big|  }_{B_{81}}
+ \underbrace{ \big|\mathbb{P}\big(\mathcal{L}_1\widetilde G_1\le x\big)-\mathbb{P}\big(\mathcal{L}_1G_1\le x\big)\big|
}_{ B_{82} },
\end{align*}
where $\widetilde G_1\sim \mathcal{N}(0,\widetilde S)$.
To show $B_{81}=o(1),$  we first show
\begin{align}\label{Lemma8:eq1}
    \big|\mathbb{P}\big(\mathcal{L}_1\widetilde S U^\top\epsilon/\sqrt{n-1}\le x\big)-\mathbb{P}\big(\mathcal{L}_1\widetilde G_1\le x\big)\big|=o(1).
\end{align}
A direct calculation yields that the $k$th element of $\mathcal{L}_1\widetilde S U^\top\epsilon$
is $\sum_{i=1}^n\mathcal{L}_{1k}(\epsilon_i+\epsilon_{in}),$
where $\mathcal{L}_{1k}$ is the $k$th row of $\mathcal{L}_{1}$
and $\epsilon_i=(\epsilon_{1i},\dots,\epsilon_{ni})^\top$ with $\epsilon_{ii}=0$.
We apply Theorem 2.1 of Chernozhukov et al. (2022) to the sequence $\mathcal{L}_{1k}(\epsilon_i+\epsilon_{in})=\sum_{s=1}^{n}L_{1,ks}(\epsilon_{si}+\epsilon_{in})$.
By Conditions (C5) and (C6), we know that
\begin{align*}
|\epsilon_{ij}|\le |Y_{ij}|+|\alpha_{0i}+\beta_{0j}|+|Z_{ij}^\top \eta|
\le 1/m+q_n+\kappa.
\end{align*}
By Condition (C8), we have $|L_{1,ks}{\epsilon}_{si}|\le L_{U\alpha}|\epsilon_{si}|$
and
$$\big|\sum_{s=1}^{n}L_{1,ks}(\epsilon_{si}+\epsilon_{in})\big|\le 2L_{U\alpha}s_{U\alpha}\max_{1\le i\neq j\le n}|\epsilon_{ij}\big|.$$
This implies that
\begin{align*}
\mathbb{E}\bigg(\exp\Big\{\Big|\sum_{s=1}^{n}L_{1,sk}(\epsilon_{si}+\epsilon_{in})\Big|/B_n\Big\}\bigg)\le 2
\end{align*}
with $B_n\ge 2L_{U\alpha}s_{U\alpha}(1/m+q_n+\kappa).$

We next show
\begin{align}\label{eq:M}
\frac{1}{n-1}\sum_{k\neq i}^n\mathbb{E}({\epsilon}_{ik}^2) \ge b_1^2,&\qquad\qquad~
\frac{1}{n-1}\sum_{k\neq j}^n\mathbb{E}({\epsilon}_{kj}^2) \ge b_1^2,\nonumber\\
\frac{1}{n-1}\sum_{k\neq i}^n \mathbb{E}({\epsilon}_{ik}^4)\le B_n^2 b_2^2&~~~~
\text{and}~~~~
\frac{1}{n-1}\sum_{k\neq j}^n \mathbb{E}({\epsilon}_{kj}^4)\le B_n^2 b_2^2,
\end{align}
where $b_1$ and $b_2>0$ are some constants.
By condition $\sigma_L^2<\mathbb{E}(\epsilon_{ij}^2)<\sigma_U^2,$  we have $\mathbb{E}(\epsilon_{ij}^4)\le B_n^2  \sigma_U^2.$
Therefore, \eqref{eq:M} holds by setting $b_1=\sigma_L$ and $b_2=\sigma_U.$

By condition (C8), an application of Theorem 2.1 in Chernozhukov et al. (2022) yields that
\begin{align*}
   \big| \mathbb{P}\big(\mathcal{L}_1\widetilde S U^{\top}\epsilon/\sqrt{n-1}<x\big)-\mathbb{P}(\mathcal{L}_1\widetilde G< x)\big|\le C_5\Big(\frac{B_n^2\log^5n}{n}\Big)^{1/4},
\end{align*}
where $C_5$ is a constant depending only on $b_1$ and $b_2.$
This implies that \eqref{Lemma8:eq1} holds  with $B_n^2\log^5n/n=o(1).$

We next show
\begin{align*}
   \big| \mathbb{P}\big(\mathcal{L}_1\Gamma U^{\top}\epsilon/\sqrt{n-1}<x\big)-\mathbb{P}(\mathcal{L}_1G_1< x)\big|=o(1).
\end{align*}
Let $[(\Gamma-\widetilde S) U^\top/\sqrt{n-1}]_j$ denote the $j$th row of $(\Gamma-\widetilde S) U^\top/\sqrt{n-1}.$
Note that
\begin{align*}
\|(V^{-1}-\breve S)(I-V\breve S)\|_{\max}&=\|(V^{-1}-\breve S)-\breve S(I-V\breve S)\|_{\max}\\
& \le \|V^{-1}-\breve S\|_{\max}+\|\breve S(I-V\breve S)\|_{\max}\\
& \le \frac{C_0}{(n-1)^2}+\frac{2}{(n-1)^2}.
\end{align*}
This, together with Condition (C8), we have
\[
\big\|[\mathcal{L}_1(\Gamma-\widetilde S) U^\top/\sqrt{n-1}]_j\big\|_2^2\le O\Big(\frac{1}{n-1}\Big).
\]
Then, by the general Hoeffding’s inequality (Theorem 2.6.3, Vershynin, 2018), we have
\[
\mathbb{P}\bigg(\max_{j\in [2n-1]} \big|\frac{1}{\sqrt{n-1}}[\mathcal{L}_1(\Gamma-\widetilde S) U^\top]_j\epsilon\big|>C_6\big(1/m+q_n+\kappa\big)\sqrt{\frac{\log(n-1)}{n-1}}\bigg)\le (n-1)^{-c},
\]
where $C_6>0$ is some sufficiently large constant and $c>0$ is a constant depending on $C_6.$
This implies that
\[
\|\mathcal{L}_1(\Gamma-\widetilde S) U^\top\epsilon/\sqrt{n-1}\|_{\infty}=O_p\bigg((1/m+q_n+\kappa\big)\sqrt{\frac{\log(n-1)}{n-1}}\bigg).
\]
Let $\varrho_{1n}=C_6(1/m+q_n+\kappa\big)\sqrt{\log(n-1)/(n-1)},$
and define $R_n=(R_{ln})^\top=\mathcal{L}_1(\Gamma-\widetilde S) U^\top\epsilon.$
By Condition (C8), we have
\begin{align*}
\|\sqrt{n-1}\mathcal{L}_1R_n\|_{\infty}\le & \max_{1\le l\le 2n-1}\sum_{k\in S_1}|L_{1,lk}||\sqrt{n-1}R_{kn}|\\
\le & \max_{1\le l\le M_1, 1\le k\le 2n-1}|L_{1,lk}| \max_{1\le l\le 2n-1}\sum_{k\in S_1}|\sqrt{n-1}R_{kn}|\\
\le & \max_{1\le l\le M_1, 1\le k\le 2n-1}|L_{1,lk}| \max_{1\le l\le 2n-1}\max_{k\in S_1}|\sqrt{n-1}R_{kn}|=O_p(L_{1n}).
\end{align*}
Now for any $x\in \mathbb{R}^{2n-1},$
\begin{align*}
     \mathbb{P}\big(\mathcal{L}_1\Gamma U^\top\epsilon/\sqrt{n-1}\le x\big)
    = & \mathbb{P}\big(\mathcal{L}_1\widetilde S U^\top\epsilon/\sqrt{n-1} + \mathcal{L}_1(\Gamma-\widetilde S) U^\top\epsilon/\sqrt{n-1}\le x\big)\\
  \le & \mathbb{P}\big(\mathcal{L}_1\widetilde S U^\top\epsilon/\sqrt{n-1}\le x+\varrho_{1n}\big)+o(1) \\
  \le & \mathbb{P}\big(\mathcal{L}_1\widetilde G_1 \le x+L_{1n}\big)+o(1)\\
  \le & \mathbb{P}\big(\mathcal{L}_1\widetilde G_1 \le x\big)+O(\varrho_{1n}\sqrt{\log(2n-1)})+o(1),
\end{align*}
where second to last
inequality holds due to \eqref{Lemma8:eq1} and the last inequality holds by Lemma \ref{lem:Nazarov-inequality}.
Likewise, we have
\begin{align*}
     \mathbb{P}\big(\mathcal{L}_1\Gamma U^\top\epsilon/\sqrt{n-1}\le x\big)\ge & \mathbb{P}\big(\mathcal{L}_1\widetilde G_1 \le x\big)-O(\varrho_{1n}\sqrt{\log(2n-1)})-o(1).
\end{align*}
Thus, we obtain $B_{81}=o(1)$ with $\varrho_{1n}\sqrt{\log(2n-1)}=o(1).$

In addition,  by Lemma \ref{lem:Vbound}, we have
\begin{align*}
\|\mathcal{L}_1(\Gamma- \widetilde S)\mathcal{L}_1^\top \|_{\max} \le & O((n-1)^{-1}),
\end{align*}
which, together with Lemma \ref{lem:GGC}, implies that
\[
B_{82}=\big|\mathbb{P}\big(\mathcal{L}_1\widetilde G_1\le x\big)-\mathbb{P}\big(\mathcal{L}_1G_1\le x\big)\big|\le O(\log(n)/\sqrt{n}).
\]
This completes the proof.
\end{proof}

We are now ready to prove Theorem 3.

\begin{proof}[Proof of Theorem 3]
Define $\varrho_{2n}=\kappa^2/(\phi m^2)[1/\sqrt{n}+\sqrt{\log(n)/(nh^{2p+2})}+\sqrt{n}h^{r}]$
and $r_n=\sqrt{n-1}\mathcal{L}_1(\alpha)(\widehat\theta-\theta_0)-\mathcal{L}_1(\alpha)\Gamma U^\top\epsilon/\sqrt{n-1}$.
Then, by Lemma \ref{Lemma4} and Condition (C8),  we have
\begin{align*}
&\mathbb{P}\big(\sqrt{n-1}|\mathcal{L}_1(\alpha)r_n|>C_7\varrho_{2n}\big)\\
\le &\mathbb{P}\Big(\max_{1\le l\le M_1, 1\le k\le 2n-1}|L_{1,lk}(\alpha)| \max_{1\le l\le 2n-1}\max_{k\in S_1}|\sqrt{n-1}r_{kn}|>C_7\varrho_{2n}\Big)=o(1),
\end{align*}
where $C_7>0$ is a sufficiently large constant.
Now for any $x\in\mathbb{R}^{2n-1},$
\begin{align*}
\mathbb{P}\big(\sqrt{n-1}\mathcal{L}_1(\alpha)(\widehat\theta-\theta_0)\le x \big)
=& \mathbb{P}\big(\mathcal{L}_1(\alpha)\Gamma U^\top\epsilon/\sqrt{n-1}+r_n\le x\big)\\
\le & \mathbb{P}\big(\mathcal{L}_1(\alpha)\Gamma U^\top\epsilon/\sqrt{n-1}\le x+C_7\varrho_{2n}\big)+o(1)\\
\le & \mathbb{P}\big(\mathcal{L}_1(\alpha)G_1\le x+C_7\varrho_{2n}\big)+o(1) \qquad\qquad\qquad~~~~~(\text{by~Lemma}~\ref{lem:GA})\\
\le & \mathbb{P}\big(\mathcal{L}_1(\alpha)G_1\le x\big)+O(\varrho_{2n}\sqrt{\log(2n-1)})+o(1)~~~~(\text{by~Lemma}~\ref{lem:Nazarov-inequality}).
\end{align*}
Likewise, we have
\begin{align*}
\mathbb{P}\big(\sqrt{n-1}\mathcal{L}_1(\alpha)(\widehat\theta-\theta_0)\le x \big)\ge & \mathbb{P}\big(\mathcal{L}_1(\alpha)G_1\le x\big)-O(\varrho_{2n}\sqrt{\log(2n-1)})-o(1).
\end{align*}
Since $\varrho_{2n}\sqrt{\log(2n-1)}=o(1),$ we obtain the conclusion of the theorem.
\end{proof}

\subsection{Proof of Theorem 4}
\label{subsec-pro-th4}

\begin{proof}
By Lemma \ref{Lemma:linear:eta},
the proof follows the same arguments as Lemma \ref{lem:GA} and Theorem 3,
and we omit the details here.
\end{proof}

\subsection{Proof of Theorem 5}
\label{subsec-pro-th5}

\begin{proof}
We only present the proof of the first part,
and the proof of the second part is similar and omitted.
By Lemma \ref{lem:GGC} and Condition (C8), it suffices to show $\Psi=o_p(\log^{-2}n/n),$
where
\begin{align}\label{Theorem5:eq1}
    \Psi=\|\sigma_{\epsilon}^2V^{-1}-\widehat{\sigma}_{\epsilon}^2 V^{-1}\|_{\max}.
\end{align}
Note that
\begin{align*}
\widehat{\sigma}_{\epsilon}^2-\sigma_{\epsilon}^2
=&\frac{1}{n(n-1)}\sum_{i=1}^n\sum_{j\neq i}(\widehat\epsilon_{ij}-\epsilon_{ij})^2
+\frac{2}{n(n-1)}\sum_{i=1}^n\sum_{j\neq i}\epsilon_{ij}(\widehat\epsilon_{ij}-\epsilon_{ij})\\
&+\frac{1}{n(n-1)}\sum_{i=1}^n\sum_{j\neq i}(\epsilon_{ij}^2-\sigma_{\epsilon}^2).
\end{align*}
Then, by Theorem 2 and \eqref{eqB6}, we can show
\begin{align*}
    |\widehat{\sigma}_{\epsilon}^2-\sigma_{\epsilon}^2|
    = & O_p\big(\|\widehat Y-Y\|_{\infty}+2\|\widehat\theta-\theta_0\|_{\infty}+\kappa\|\widehat\eta-\eta_0\|+1/n\big).\\
    = & O_p\Big(\frac{\kappa^2}{\phi m^2}\Big[h^r+(\kappa+q_n)\sqrt{\frac{\log(n)}{n}}+\sqrt{\frac{\log(n)}{n^2h^{2p+2}}}\Big]\Big).
\end{align*}
This, together with $\|V^{-1}\|_{\max}=O(1/(n-1)),$ gives
\begin{align*}
|\widehat{\sigma}_{\epsilon}^2-\sigma_{\epsilon}^2|\times \|V^{-1}\|_{\max}= O_p\Big(\frac{\kappa^2}{n\phi m^2}\Big(h^r+(\kappa+q_n)\sqrt{\frac{\log n}{n}}+\sqrt{\frac{\log n}{n^2h^{2p+2}}}\Big)\Big).
\end{align*}
Because $(\kappa^2/\phi m^2)[h^r+(\kappa+q_n)\sqrt{(\kappa+q_n)\log(n)/n}+\sqrt{\log(n)/(n^2h^{2p+2})}]=o(\log^{-2}n)$ under the condition (4.3),
we have $\Psi=o_p(\log^{-2}n/n)$.
This completes the proof.
\end{proof}

\subsection{Proof of Proposition 1}
\label{subse-pro-pro1}

\begin{proof}
We only show the conclusion for $\mathcal{T}_{\alpha,S}.$
Let $\widetilde{\mathcal{T}}_{\alpha,S}=\sqrt{n-1}|\max_{i\in[n]}\widehat\alpha_i-\alpha_{0i}|/\widehat\zeta_{i}^{1/2}$
and $\Psi=\text{diag}\{\zeta_{1}^{1/2},\dots,\zeta_{n}^{1/2}\}.$
Since $|\widehat\zeta_{i}-\zeta_{i}|\le (n-1)\|\sigma_{\epsilon}^2V^{-1}-\widehat{\sigma}_{\epsilon}^2V^{-1}\|_{\max}=o_p(1)$ as in the proof of Theorem 5, we have
\begin{align}\label{Proposition:5.1:eq0}
\max_{i\in[n]}|1/\widehat\zeta_{i}-1/\zeta_{i}|=\max_{i\in[n]}|\widehat\zeta_{i}-\zeta_{i}|/(\widehat\zeta_{i}\zeta_{i})=o_p(1).
\end{align}
Under the conditions of Theorem 5,
we can show that the distribution of $\widetilde{\mathcal{T}}_{\alpha,S}$
can be approximated by $\|G_{1\alpha}\|_{\infty}$
with $G_{1\alpha}=(G_{11},\dots,G_{1n})^\top\sim N(0,\Psi^{-1/2}\Omega_{\alpha}\Psi^{-1/2}),$
where $\Omega_{\alpha}=(\omega_{\alpha,ij})$ is the upper left $n\times n$ block of $(n-1)\sigma^2V^{-1}.$
Note that $|\omega_{\alpha,ij}|/(\zeta_{i}\zeta_{j})<1$ and $V$ is positive definite.
By Lemma 6 of Cai et al. (2014), we have for any $x\in\mathbb{R}$ and as $n\rightarrow \infty,$
\[
\mathbb{P}\big(\max_{i\in[n]}|G_{1i}|^2-2\log(n)+\log(\log n)\le x\big)\rightarrow F(x)=\exp\{-\exp(-x/2)/\sqrt{\pi}\}.
\]
In other word, $\max_{i\in[n]}|G_{1i}|^2$ converges in distribution to a Gumbel extreme value distribution.
It further implies that
\[
\mathbb{P}\big(\widetilde{\mathcal{T}}_{\alpha,S}^2\le 2\log(n)-\log(\log n)/2\big)\rightarrow 1.
\]
This, together with the result of Theorem 5, implies that
\begin{align}\label{proposition:eq1}
|c_{\alpha,S}(\nu)^2-2\log(n)+\log(\log n)/2-q_{1-\alpha}|=o(1),
\end{align}
where $q_{1-\alpha}$ denotes the upper $\alpha$-quantile of the distribution $F(x).$
Let $i^*$ be the index such that $|\widehat\alpha_{i^*}-\alpha_{0i^*}|/\zeta_{i^*}^{1/2}>(\sqrt{2}+\lambda_0)\sqrt{\log(n)/(n-1)}.$
Then, for any $\varsigma>0,$ we have
\begin{align}\label{proposition:eq2}
|\alpha_{0i^*}|^2/\widehat\zeta_{i^*}=&|\widehat\alpha_{i^*}-(\widehat\alpha_{i^*}-\alpha_{0i^*})|^2/\widehat\zeta_{i^*}\nonumber\\
\le & (1+\varsigma^{-1})|\widehat\alpha_{i^*}-\alpha_{0i^*}|^2/\widehat\zeta_{i^*} +(1+\varsigma)|\widehat \alpha_{i^*}|^2/\widehat\zeta_{i^*}.
\end{align}
Note that $(n-1)|\widehat\alpha_{i^*}-\alpha_{0i^*}|^2/\widehat\zeta_{i^*}=o_p(\log n)$ if $i^*$ is fixed.
Thus, by \eqref{proposition:eq2} and the fact that $\alpha\in \mathcal{U}_{\alpha}(\sqrt{2}+\lambda_0),$
we have
\begin{align}\label{proposition:eq3}
(n-1)|\widehat\alpha_{i^*}|^2/\widehat\zeta_{i^*}\ge \frac{1}{1+\varsigma}\Big\{(\sqrt{2}+\lambda_0)^2(\log n-o_p(\log n))\Big\}.
\end{align}
By \eqref{proposition:eq1}, \eqref{proposition:eq3} and choosing a small $\varsigma,$ we have
\[
\inf_{\alpha_0\in \mathcal{U}_{\alpha}(\sqrt{2}+\lambda_0)}\mathbb{P}\big(\max_{i\in[n]} \sqrt{n-1}|\widehat\alpha_i|/\widehat\zeta_{i}^{1/2}>c_{\alpha,S}(\nu) \big)\rightarrow 1.
\]
This completes the proof.
\end{proof}

\subsection{Proof of Proposition 2}
\label{subsec-pro-pro2}

\begin{proof}
We first show consistency of $\widehat{\mathcal{S}}_{\alpha}(2)$.
With similar arguments as in the proof of Proposition 1, we can show that
\begin{align}\label{proposition:SR:eq1}
\mathbb{P}\big(\max_{i\in\mathcal{S}_{0\alpha}^{c}}\sqrt{n-1}|\widehat\alpha_i|^2/\widehat\zeta_{i}\le 2\log(n)-\log(\log n)/2\big)\rightarrow 1,
\end{align}
where $\mathcal{S}_{0\alpha}^{c}$ denotes the complementary set of $\mathcal{S}_{0\alpha}.$

Note that
\begin{align*}
\min_{i\in\mathcal{S}_{0\alpha}}|\alpha_{0i}|^2/\widehat\zeta_{i}
=&\min_{i\in\mathcal{S}_{0\alpha}}|\widehat\alpha_{i}-(\widehat\alpha_{i}-\alpha_{0i})|^2/\widehat\zeta_{i}\nonumber\\
\le & 2\max_{i\in\mathcal{S}_{0\alpha}}|\widehat\alpha_{i}-\alpha_{0i}|^2/\widehat\zeta_{i} +2\min_{i\in\mathcal{S}_{0\alpha}}|\widehat \alpha_{i}|^2/\widehat\zeta_{i}.
\end{align*}
In addition, because $(n-1)|\alpha_{0i}|^2/\widehat\zeta_{i}=(n-1)|\alpha_{0i}|^2/\zeta_{i}+o_p(1)$ and
\[
\mathbb{P}\big(\max_{i\in\mathcal{S}_{0\alpha}}(n-1)|\widehat\alpha_{i}-\alpha_{0i}|^2/\widehat\zeta_{i}\le 2\log(n)-\log(\log n)/2\big)\rightarrow 1,
\]
we have
\begin{align}\label{proposition:SR:eq2}
&\mathbb{P}\big(\min_{i\in\mathcal{S}_{0\alpha}}(n-1)|\widehat\alpha_{i}|^2/\widehat\zeta_{i}\le 2\log(n)\big)\nonumber\\
\ge  &\mathbb{P}\big(2\min_{i\in\mathcal{S}_{0\alpha}}(n-1)|\widehat\alpha_{i}|^2/\widehat\zeta_{i}-4\log(n)+\log(\log n)\le 8\log(n)\big)
\rightarrow 1.
\end{align}
Combining the results of \eqref{proposition:SR:eq1} and \eqref{proposition:SR:eq2}, we have
\[\inf_{\alpha_0\in \mathcal{G}_{\alpha}(2\sqrt{2})}\mathbb{P}\big(\widehat{\mathcal{S}}_{\alpha}(2)=\mathcal{S}_{0\alpha}\big)\rightarrow 1.\]

We next show the optimality of $t=2$.
Note that with a sufficiently large $n,$ we can
choose a subset $\widetilde{\mathcal{G}}_{\alpha}$ of $\{i\in[n]: \alpha_{0i}=0\}$
such that $n^{t/2}<|\widetilde{\mathcal{G}}_{\alpha}|< n.$ For simplicity, we denote $|\widetilde{\mathcal{G}}_{\alpha}|=\lfloor n^{t_1}\rfloor$
with $t/2<t_1<1,$ where $\lfloor x\rfloor$ denotes the maximum integer that is not larger than $x.$
Then, using the above arguments, we can show that the distribution of $\max_{i\in \widetilde{\mathcal{G}}_{\alpha}}\sqrt{n-1}|\widehat\alpha_{i}|/\widehat{\zeta}_{i}^{1/2}$
can be approximated by $\max_{i\in \widetilde{\mathcal{G}}_{\alpha}} |G_{1i}|$
with $G_{1\alpha}=(G_{11},\dots,G_{1n})^\top\sim N(0,\Omega_{\alpha}),$
and for any $x\in\mathbb{R}$ and as $n\rightarrow \infty,$
\[
\mathbb{P}\big(\max_{i\in\widetilde{\mathcal{G}}_{\alpha}}|G_{1i}|^2-2\log(\lfloor n^{t_1}\rfloor)+\log(\log \lfloor n^{t_1}\rfloor)\le x\big)\rightarrow \exp\{-\exp(-x/2)/\sqrt{\pi}\}.
\]
Then, for any $t_2$ such that $t<t_2<2t_1<2,$ we have
\[
\mathbb{P}\big(\max_{i\in \widetilde{\mathcal{G}}_{\alpha}} (n-1)|\widehat\alpha_{i}|^2/\widehat{\zeta}_{i}>t_2\log(n)\big)\rightarrow 1.
\]
In other words, some elements in $\widetilde{\mathcal{G}}_{\alpha}$ are retained in $\widehat{\mathcal{S}}_{\alpha}(t)$ with $0<t<2.$
The conclusion thus follows from the definition of $\widetilde{\mathcal{G}}_{\alpha}.$
Similarly, we can show that it also holds for $\widehat{\mathcal{S}}_{\beta}(t).$
This completes the proof.
\end{proof}

\subsection{Proof of Theorem 6}
\label{subsec-pro-th6}

\begin{proof}
By some arguments similar to the proof of \eqref{Theorem5:eq1}, we have
\begin{align*}
\Psi=&\|V^{-1}U^\top(\widehat{\text{Cov}(\epsilon)}-{\text{Cov}(\epsilon)})UV^{-1}\|_{\max}\\
= &O_p\Big(\frac{\kappa^2}{n\phi m^2}\Big(h^r+(\kappa+q_n)\sqrt{\frac{\log n}{n}}+\sqrt{\frac{\log n}{n^2h^{2p+2}}}\Big)\Big).
\end{align*}
which implies that $\Psi=o_p(\log^{-2}n/n)$ under the condition (4.3).
This completes the proof.
\end{proof}

\section{Additional simulation studies}
\label{section-add-simulation}
In this section, we conduct additional simulation studies to evaluate the finite sample performance of the proposed method.

\subsection{Simulation studies for conditionally independent cases}
We independently generate $\varepsilon_{ij}~(i,j\in[n])$ from  $U(-0.5,0.5)$ when $Z_{ij1}Z_{ij2}>0$
and from $U(-1,1)$ otherwise. We set $n=100$ and $150.$
All other settings remain consistent with those in Section 5 of the main paper.
The results, presented in Tables S1 and S2 and Figures S1 and S2,
indicate that the proposed methods still work well
even when $\varepsilon_{ij}$ is conditionally independent of $X_{ij1}$ given $Z_{ij}$.

\begin{table}[h]\centering
{\small
\caption*{Table S1: The results of bias, standard deviation (SD) and $95\%$ empirical coverage probability (CP) of the estimators for conditionally-independent cases.}
\label{Tab:CN:Herenorm}
\begin{tabular}{ccrcccrcc}
\hline
& &\multicolumn{3}{c}{$\rho_1=0.1$} && \multicolumn{3}{c}{$\rho_1=0.2$} \\
\cline{3-5} \cline{7-9}
$n$ & & Bias & SD &  CP($\%$) &  & Bias & SD &  CP($\%$)\\
\hline
$n=50$	&	$\alpha_{01}$	&	0.017 	&	0.312 	&	91.2 	&&	0.017 	&	0.292 	&	92.8 	\\
	&	$\alpha_{0\frac{n}{2}}$	&	$-$0.002 	&	0.259 	&	96.3 	&&	0.000 	&	0.253 	&	96.1 	\\
	&	$\alpha_{0n}$	&	$-$0.009 	&	0.240 	&	97.2 	&&	0.000 	&	0.260 	&	95.0 	\\
	&	$\alpha_{\frac{n}{5},\frac{4n}{5}}$	&	0.015 	&	0.280 	&	95.1 	&&	0.023 	&	0.290 	&	92.4 	\\
	&	$\alpha_{\frac{n}{2},\frac{n}{2}+1}$	&	0.006 	&	0.265 	&	95.8 	&&	$-$0.007 	&	0.257 	&	95.7 	\\
	&	$\eta_{01}$	&	0.015 	&	0.027 	&	92.3 	&&	0.016 	&	0.027 	&	92.1 	\\
	&	$\eta_{02}$	&	$-$0.014 	&	0.026 	&	92.5 	&&	$-$0.018 	&	0.026 	&	92.9 	\\
$n=100$	&	$\alpha_{01}$	&	0.034 	&	0.257 	&	91.2 	&&	0.028 	&	0.234 	&	91.6 	\\
	&	$\alpha_{0\frac{n}{2}}$	&	$-$0.006 	&	0.192 	&	97.0 	&&	$-$0.002 	&	0.184 	&	97.2 	\\
	&	$\alpha_{0n}$	&	0.000 	&	0.191 	&	97.9 	&&	$-$0.004 	&	0.207 	&	94.9 	\\
	&	$\alpha_{\frac{n}{5},\frac{4n}{5}}$	&	0.022 	&	0.262 	&	91.1 	&&	0.017 	&	0.238 	&	90.7 	\\
	&	$\alpha_{\frac{n}{2},\frac{n}{2}+1}$	&	0.025 	&	0.241 	&	92.2 	&&	0.009 	&	0.232 	&	92.2 	\\
	&	$\eta_{01}$	&	0.006 	&	0.014 	&	94.4 	&&	0.009 	&	0.014 	&	94.0 	\\
	&	$\eta_{02}$	&	$-$0.005 	&	0.014 	&	94.4 	&&	$-$0.009 	&	0.014 	&	93.8 	\\
\hline
\end{tabular}}
\end{table}

\begin{table}[h]\centering
\caption*{Table S2: The mean and standard deviation (SD) of $\mathcal{M}(\mathcal{S}_0, \widehat{\mathcal{S}})$, and the numbers of false positives (FP) and false negatives (FN).}\label{Tab:Supp:Here}
{\small
\begin{tabular}{cccccccccc}
\hline
&\multicolumn{4}{c}{$\alpha_0$} && \multicolumn{4}{c}{$\beta_0$}\\
\cline{2-5} \cline{7-10}
 $n$ & Mean & SD &  FP & FN &  & Mean & SD &  FP & FN\\
\hline
$100$	&	0.959 	&	0.094 	&	0.430 	&	0.533 	&&	0.958 	&	0.099 	&	0.463 	&	0.522 	\\
$150$	&	0.981 	&	0.058 	&	0.327 	&	0.117 	&&	0.981 	&	0.061 	&	0.332 	&	0.130 	\\
\hline
\end{tabular}}
\end{table}

\begin{figure}
\centering
\includegraphics[width=0.75\textwidth]{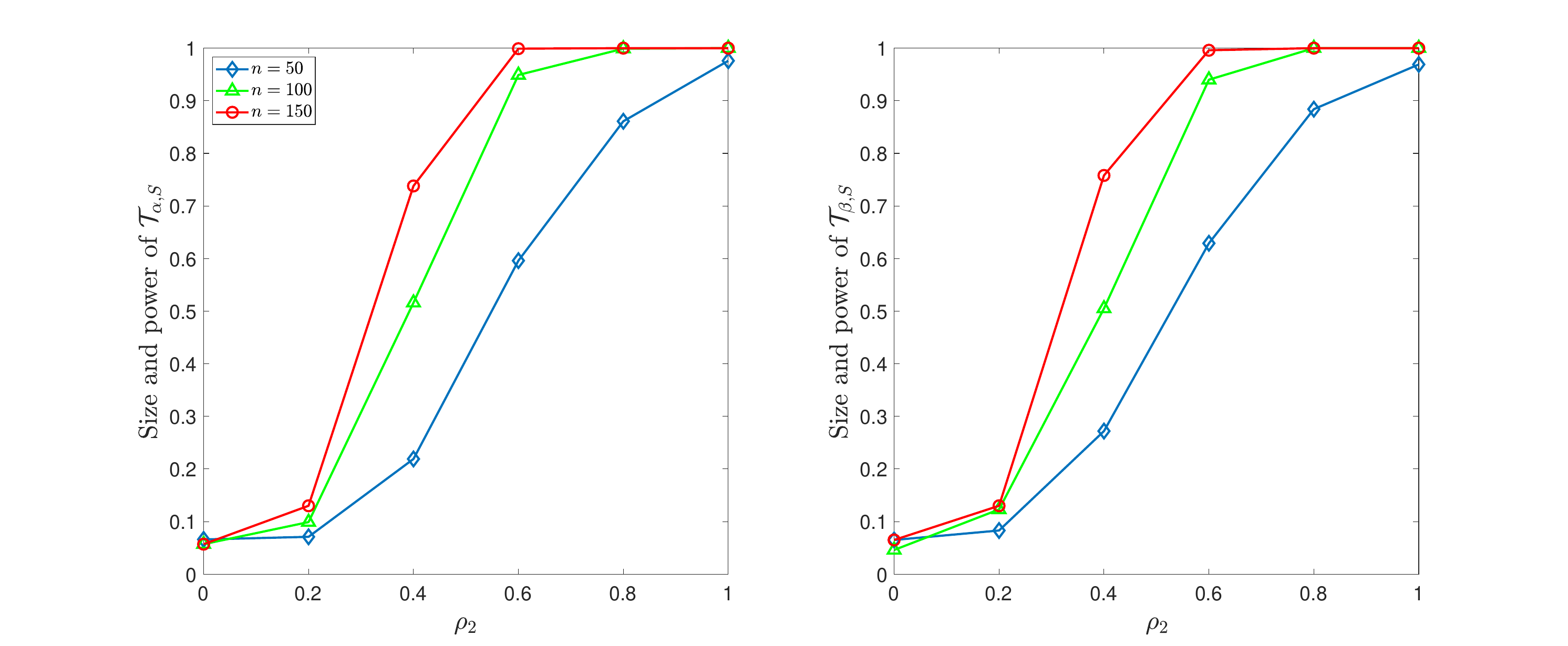}
\caption*{Figure S1: Empirical size and power of $\mathcal{T}_{\alpha,S}$
and $\mathcal{T}_{\beta,S}$ for conditionally independent cases.}\label{Figure:Supp:TS}
\end{figure}

\begin{figure}
\centering
\includegraphics[width=0.75\textwidth]{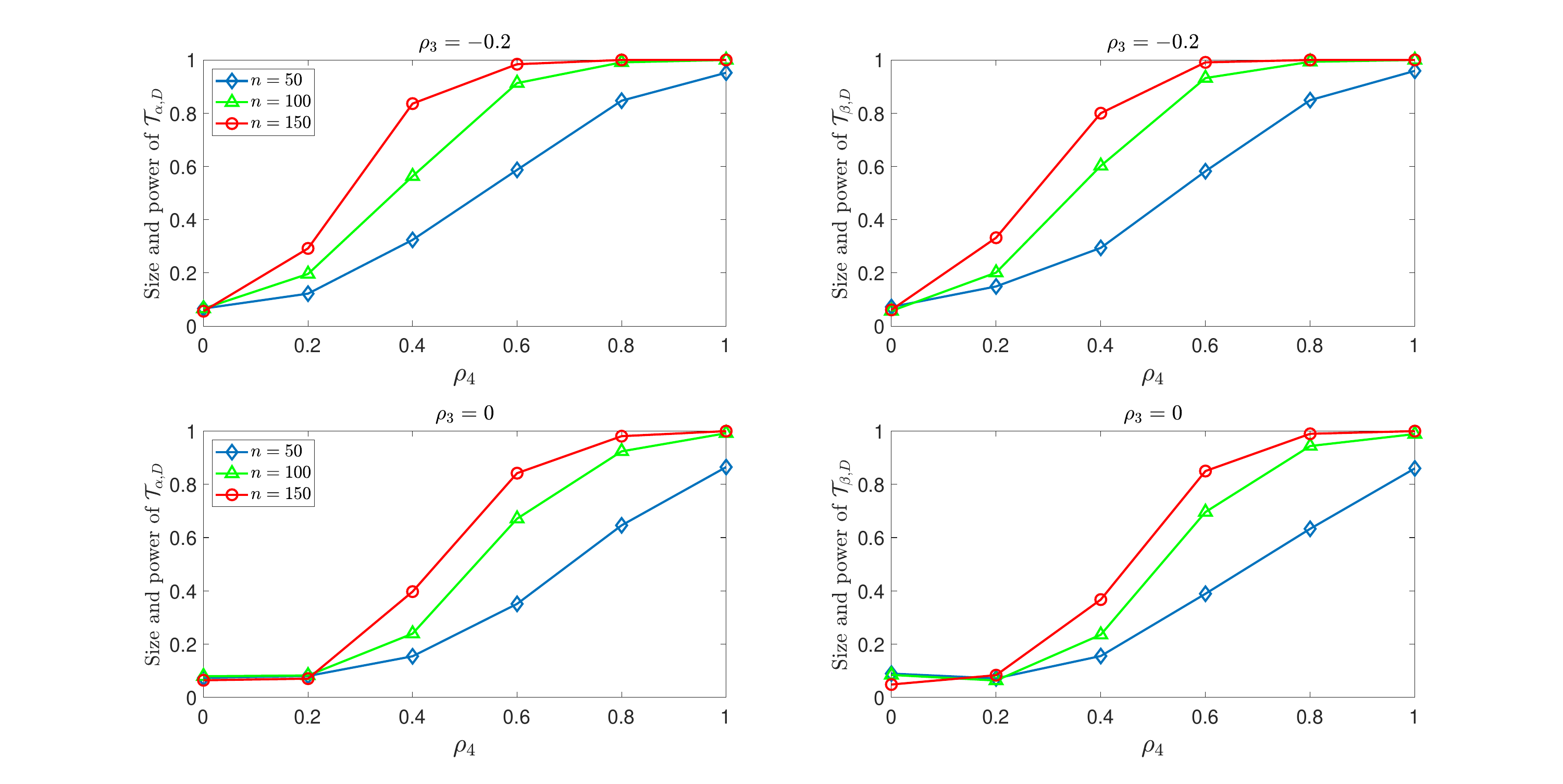}
\caption*{Figure S2: Empirical size and power of $\mathcal{T}_{\alpha,D}$
and $\mathcal{T}_{\beta,D}$ for conditionally independent cases.}\label{Figure:Supp:TD}
\end{figure}

\subsection{Comparison with Candelaria's method}

In this subsection, we conduct simulation studies to compare our method with Candelaria's approach.
Since Candelaria's method focuses on modeling undirected networks,
we adjust our method to accommodate this scenario.
In addition, we also compare with the conditional maximum likelihood method presented in Graham (2017).
To ensure a fair comparison, we set $\gamma_{01}=1$ when implementing the method of Graham (2017).
We employ the simulation settings used in Candelaria (2020).
Specifically, let $\breve{X}_i$ be drawn from a beta distribution $\text{Beta}(2, 2)-1/2$,
where $\text{Beta}(a_1, a_2)$ denotes a beta distribution with parameters $a_1$ and $a_2$.
The pair-specific covariate $Z_{ij}$ is defined as $Z_{ij}=\breve{X}_i\breve{X}_j$.
The agent-specific unobserved factor $\alpha_{0i}$ is generated as follows:
$$
\alpha_{0i}=0.75\breve{X}_i - 0.25C_n\times\text{Beta}(0.5, 0.5),
$$
where $C_n\in\{\log n, \log\log n, (\log n)^{1/2}\}$ controls the sparsity level of the network.
The special regressor $X_{ij1}=X_{ji1}$ is simulated as $X_{ij1}\sim \mathcal{N}_1 (0, 2)$ for $i<j$.
The noise term $\varepsilon_{ij}=\varepsilon_{ji}$ is generated from $\text{Beta}(2,2)-1/2$ for $i<j$.
We set $\eta_0=1.5$ and the network size $n=50$.
Then, we generate $A_{ij}$ using the following formula:
$A_{ij}=\mathbb{I}(\alpha_{0i}+\alpha_{0j}+X_{ij1}+Z_{ij}^\top \eta_0>\varepsilon_{ij}).$
The results are presented in Table S3.
We see that our method exhibits a smaller bias in the estimators compared to Candelaria's method,
while both methods yield comparable standard errors.
Since the noise does not follow a logistic distribution, Graham's method exhibits a non-negligible bias in the estimator.
Moreover, our method is computationally more efficient than the methods in Graham (2017) and Candelaria (2020).
This is because their methods involve a U-statistic of order four, which leads to a high computational cost.
In contrast, our projection method with a least squares solution effectively addresses this issue.

\begin{table}[h]\centering
\caption*{Table S3: The comparison results of our method, Graham's method and Candelaria's method.}\label{Tab:Supp:CB}
{\small
\begin{tabular}{cccccccccccc}
\hline
&\multicolumn{3}{c}{Graham's method} && \multicolumn{3}{c}{Candelaria's method} && \multicolumn{3}{c}{Our method}\\
\cline{2-4} \cline{6-8} \cline{10-12}
 $C_n$ & Bias & SD & Time &  & Bias & SD & Time && Bias & SD & Time\\
\hline
$\log\log(n)$	&	0.519	&	0.856	&	9.154 	&&	0.151 	&	0.847 	&	8.181 	&&	$-0.060$ 	&	0.776 	&	0.081 	\\
$\log(n)^{1/2}$	&	0.536	&	0.885	&	9.160 	&&	0.116 	&	0.896 	&	8.212 	&&	$-0.064 $	&	0.936 	&	0.080 	\\
$\log(n)$	&	0.562	&	1.017	&	9.259 	&&	0.168 	&	1.129 	&	8.244 	&&	$-0.019$ 	&	1.378 	&	0.079 	\\
\hline
\end{tabular}}
\end{table}

\subsection{Asymptotic relative efficiency}
In this subsection, we study the asymptotic relative efficiency of the proposed method.
For comparison, we consider the following model:
$$
\mathbb{P}(A_{ij}=1|X_{ij1},Z_{ij})=\frac{e^{(\alpha_{0i}+\beta_{0j}+X_{ij1}+Z_{ij}^\top\eta_0)/\upsilon}}
{1+e^{(\alpha_{0i}+\beta_{0j}+X_{ij1}+Z_{ij}^\top\eta_0)/\upsilon}}.
$$
In this context, the noise $\varepsilon_{ij}$ in model (2.1) is assumed to follow a generalized logistic distribution,
with distribution function given by $(1+e^{-x/\upsilon})^{-1}$.
When obtaining the MLE, we assume that $\upsilon$ is known a priori to ensure model identification and focus on estimating $\alpha_{0i},~\beta_{0j}$ and $\eta_0$.
To estimate these unknown parameters, we employ the maximum likelihood method developed in Yan et al. (2019).
Using arguments similar to those in the proofs of Theorems 3.1-3.3 from Yan et al. (2019),
we can establish that the asymptotic variance of the maximum likelihood estimator (MLE) for $\alpha_{0i}$ under the generalized logistic model is given by
$\breve\zeta_{i}=\upsilon^{-1}[1/{\breve{v}_{ii}}+1/\breve{v}_{(2n)(2n)}]$, where
$$
\breve{v}_{ii}=\sum_{j\neq i} \frac{ e^{(\alpha_{0i}+\beta_{0j}+X_{ij1}+Z_{ij}^\top\eta_0)/\upsilon}}
{[1+e^{(\alpha_{0i}+\beta_{0j}+X_{ij1}+Z_{ij}^\top\eta_0)/\upsilon}]^2},
$$
for $i=1, \ldots, n-1$ and
$$
\breve{v}_{(2n)(2n)}=\sum_{i=1}^{n-1} \frac{ e^{(\alpha_{0i}+\beta_{0n}+X_{in1}+Z_{in}^\top\eta_0)/\upsilon}}
{[1+e^{(\alpha_{0i}+\beta_{0n}+X_{in1}+Z_{in}^\top\eta_0)/\upsilon}]^2}.
$$
Furthermore, according to Theorem 3 in the main paper,
the asymptotic variance of $\widehat\alpha_i$ is given by $\zeta_{i}=(2n-1)\sigma^2_{\epsilon}/[n(n-1)]$.
Therefore, the asymptotic relative efficiency (ARE) between $\zeta_i$ and $\breve\zeta_i$ is defined as
\begin{align*}
\text{ARE}(\zeta_i,\breve\zeta_i)
=\frac{\zeta_i}{\breve\zeta_i}
=\frac{\sigma^2_{\epsilon}}{\upsilon}\frac{(2n-1)}{n(n-1)}\frac{1}{\breve{v}_{ii}^{-1}+\breve{v}_{(2n)(2n)}^{-1}}.
\end{align*}
Since the above ratio is highly dependent on the configuration of parameters,
providing a direct comparison is challenging.
Therefore, we conduct simulation studies to evaluate it.
Specifically, we set $\upsilon\in\{0.5, 0.6, 0.7, 0.8, 0.9, 1\}.$
For each $i\in[n]$, the parameter $\alpha_{0i}$ is taken as
$$\alpha_{0i}=-0.25\log(n)+\frac{i-1}{n-1}\big[0.25\log(n)+\rho_1\log(n)\big],$$
where $\rho_1\in\{0, 0.1\}$ controls the sparsity level of the network.
In addition, we set $\beta_{0i}=\alpha_{0i}$ for $i\in[n-1]$
and $\beta_{0n}=0$ for model identifiability.
All other settings remain consistent with those outlined in Section 7.1.
We estimate $\text{ARE}(\zeta_i,\breve \zeta_i)$ by replacing $\zeta_i$ and $\breve \zeta_i$ with their estimators,
denoted by $\widehat{\text{ARE}}(\zeta_i,\breve \zeta_i)$.
We then calculate the average of AREs as follows: $\text{ave(ARE)}=n^{-1}\sum_{i=1}^n\widehat{\text{ARE}}(\zeta_i,\breve \zeta_i)$.
Furthermore, we examine a scenario where the model is misspecified for the MLE by incorrectly setting $\upsilon$ to 1 for various values of $\upsilon$.
The results are presented in Table S4.
We observe that the values of $\text{ave(ARE)}$ are approximately 1.3 for various values of $\upsilon$
when compared to the MLE obtained under correct model specification.
This indicates that our method may exhibit some loss of efficiency,
as it does not specify a parametric model for the noise distribution.
However, in cases where the model is misspecified for the MLE,
the values of $\text{ave(ARE)}$ decrease as $\upsilon$ varies from 1 to 0.5, even falling below $1$ when $\upsilon=0.5$ and $0.6$.
This suggests that our method can be more efficient than the MLE obtained under model misspecification.
This is understandable, as parametric methods are highly sensitive to model assumptions.

\begin{table}[h]\centering
\caption*{Table S4: The results for asymptotic relative efficiency.
In Case I, the model is correctly specified with $\upsilon$ known a priori,
whereas Case II illustrates a misspecified model in which $\upsilon$ is incorrectly set to 1 for various values of $\upsilon$.}\label{Tab:Supp:ARE}
{\small
\begin{tabular}{ccccccccccccccccccc}
\hline
&\multicolumn{2}{c}{$\upsilon=0.5$} && \multicolumn{2}{c}{$\upsilon=0.6$} && \multicolumn{2}{c}{$\upsilon=0.7$}
&& \multicolumn{2}{c}{$\upsilon=0.8$} && \multicolumn{2}{c}{$\upsilon=0.9$} && \multicolumn{2}{c}{$\upsilon=1$}\\
\cline{2-3} \cline{5-6} \cline{8-9} \cline{11-12} \cline{14-15} \cline{17-18}
Case & $\rho_1=0$ & 0.1 && 0 & 0.1 && 0 & 0.1 && 0 & 0.1 && 0 & 0.1 && 0 & 0.1\\
\hline
I	&	1.31  & 1.31 &&  1.32  & 1.30 &&  1.32 &  1.30 && 1.43 & 1.28  &&  1.27 & 1.31 && 1.38 & 1.28  \\
II	&	0.83  &	0.85 &&  0.94  & 0.93 &&  1.02 &  1.02 && 1.22 & 1.10  &&  1.18 & 1.21 && 1.38 & 1.28  \\
\hline
\end{tabular}}
\end{table}

\subsection{Simulation studies with weighted networks}
In this section, we evaluate the finite-sample performance of the proposed method for weighted networks.
Specifically, we generate $A_{ij}$ using
$
    A_{ij}=\sum_{l=0}^{R-1} \pi_l \mathbb{I}(\omega_{0l}<\alpha_{0i}+\beta_{0j}+X_{ij1}+Z_{ij}^\top\eta_0-\varepsilon_{ij}\le \omega_{0,l+1}).
$
Here, we set $R=7$ and $\pi_l=l-1.$
Under this setting, $A_{ij}\in\{0,1,\dots, 6\}.$
In addition, let $\omega_{0l}=0.25(l-1)$ for $l\in[R-1],$ $\alpha_{0i}=\beta_{0i}=-0.3$ for $i\in[n-1],$ and
$\alpha_{0n}=\beta_{0n}=0.$
We choose the sample size to be $50$ and $100.$
All other settings remain consistent with those outlined in Section 7.
The results, presented in Table S5, indicate that the proposed method still works well for weighted networks.

\begin{table}[h]\centering
\caption*{Table S5: The results of bias, standard deviation (SD) and $95\%$ empirical coverage probability (CP) of the estimators for weighted networks.}
{\small
\begin{tabular}{ccrcccrcccrcc}
\hline
& &\multicolumn{3}{c}{$\varepsilon_{ij}\sim N(0,1)$} && \multicolumn{3}{c}{$\varepsilon_{ij}\sim \text{Logistic(0,1/2)}$} && \multicolumn{3}{c}{$\varepsilon_{ij}\sim\text{Mnorm}_1$} \\
\cline{3-5} \cline{7-9} \cline{11-13}
$n$ &   &  Bias & SD &  CP($\%$) &  & Bias & SD &  CP($\%$) &  & Bias & SD &  CP($\%$)\\
\hline
$50$	&	$\alpha_{01}$	&	0.036 	&	0.406 	&	95.6 	&	&	0.015 	&	0.390 	&	94.2 	&	&	0.035 	&	0.395 	&	95.5 	\\
	&	$\alpha_{0\frac{n}{2}}$	&	0.021 	&	0.413 	&	94.7 	&	&	0.038 	&	0.386 	&	95.0 	&	&	0.019 	&	0.400 	&	95.6 	\\
	&	$\alpha_{0n}$	&	0.047 	&	0.403 	&	95.4 	&	&	0.039 	&	0.373 	&	95.0 	&	&	0.030 	&	0.404 	&	96.3 	\\
	&	$\alpha_{\frac{n}{5},\frac{4n}{5}}$	&	0.017 	&	0.417 	&	93.7 	&	&	0.020 	&	0.403 	&	94.0 	&	&	0.016 	&	0.452 	&	92.7 	\\
	&	$\alpha_{\frac{n}{2},\frac{n}{2}+1}$	&	0.013 	&	0.400 	&	95.4 	&	&	$-$0.022 	&	0.421 	&	93.4 	&	&	$-$0.007 	&	0.415 	&	94.7 	\\
	&	$\eta_{01}$	&	$-$0.015 	&	0.034 	&	96.5 	&	&	$-$0.014 	&	0.032 	&	95.6 	&	&	$-$0.015 	&	0.037 	&	94.7 	\\
	&	$\eta_{02}$	&	0.015 	&	0.035 	&	95.2 	&	&	0.014 	&	0.032 	&	95.7 	&	&	0.016 	&	0.036 	&	94.5 	\\
	&	$\omega_1$	&	0.040 	&	0.368 	&	95.6 	&	&	0.054 	&	0.348 	&	95.0 	&	&	0.034 	&	0.375 	&	95.7 	\\
	&	$\omega_2$	&	0.025 	&	0.368 	&	95.6 	&	&	0.042 	&	0.346 	&	95.4 	&	&	0.020 	&	0.374 	&	95.7 	\\
	&	$\omega_3$	&	0.001 	&	0.368 	&	95.6 	&	&	0.023 	&	0.346 	&	95.4 	&	&	$-$0.005 	&	0.376 	&	95.5 	\\
	&	$\omega_4$	&	$-$0.032 	&	0.369 	&	95.2 	&	&	$-$0.006 	&	0.346 	&	95.5 	&	&	$-$0.039 	&	0.377 	&	95.6 	\\
	&	$\omega_5$	&	$-$0.079 	&	0.370 	&	95.0 	&	&	$-$0.047 	&	0.346 	&	95.5 	&	&	$-$0.085 	&	0.378 	&	93.8 	\\
	&	$\omega_6$	&	$-$0.138 	&	0.369 	&	94.1 	&	&	$-$0.102 	&	0.346 	&	94.7 	&	&	$-$0.146 	&	0.378 	&	93.0 	\\
$100$	&	$\alpha_{01}$	&	0.025 	&	0.190 	&	97.1 	&	&	0.010 	&	0.327 	&	94.7 	&	&	0.018 	&	0.315 	&	96.2 	\\
	&	$\alpha_{0\frac{n}{2}}$	&	0.007 	&	0.204 	&	96.5 	&	&	0.009 	&	0.318 	&	95.3 	&	&	0.007 	&	0.345 	&	95.7 	\\
	&	$\alpha_{0n}$	&	0.024 	&	0.193 	&	97.0 	&	&	0.018 	&	0.314 	&	96.4 	&	&	0.019 	&	0.319 	&	95.2 	\\
	&	$\alpha_{\frac{n}{5},\frac{4n}{5}}$	&	0.002 	&	0.211 	&	94.6 	&	&	0.000 	&	0.311 	&	93.6 	&	&	0.008 	&	0.344 	&	94.4 	\\
	&	$\alpha_{\frac{n}{2},\frac{n}{2}+1}$	&	0.012 	&	0.208 	&	95.5 	&	&	0.002 	&	0.324 	&	93.7 	&	&	0.006 	&	0.391 	&	93.6 	\\
	&	$\eta_{01}$	&	$-$0.010 	&	0.011 	&	95.7 	&	&	$-$0.015 	&	0.018 	&	92.6 	&	&	$-$0.016 	&	0.020 	&	92.5 	\\
	&	$\eta_{02}$	&	0.010 	&	0.011 	&	96.6 	&	&	0.015 	&	0.018 	&	92.3 	&	&	0.015 	&	0.020 	&	91.9 	\\
	&	$\omega_1$	&	0.039 	&	0.186 	&	95.8 	&	&	0.037 	&	0.297 	&	96.7 	&	&	0.024 	&	0.307 	&	95.0 	\\
	&	$\omega_2$	&	0.040 	&	0.186 	&	95.8 	&	&	0.029 	&	0.297 	&	96.7 	&	&	0.014 	&	0.307 	&	95.4 	\\
	&	$\omega_3$	&	0.035 	&	0.186 	&	96.1 	&	&	0.015 	&	0.298 	&	96.9 	&	&	$-$0.003 	&	0.307 	&	95.3 	\\
	&	$\omega_4$	&	0.023 	&	0.185 	&	96.7 	&	&	$-$0.006 	&	0.298 	&	97.2 	&	&	$-$0.029 	&	0.307 	&	95.2 	\\
	&	$\omega_5$	&	0.003 	&	0.185 	&	96.7 	&	&	$-$0.037 	&	0.298 	&	96.9 	&	&	$-$0.064 	&	0.308 	&	95.5 	\\
	&	$\omega_6$	&	$-$0.028 	&	0.186 	&	96.2 	&	&	$-$0.079 	&	0.299 	&	96.8 	&	&	$-$0.111 	&	0.308 	&	95.0 	\\
\hline
\end{tabular}}
\label{Tab:CN:WN}
\end{table}

\clearpage
\newpage

\subsection{Evaluating logistic parametric assumptions}

In this subsection, we develop a graphical method to compare the goodness of fit for the network data
and to investigate the logistic parametric assumptions (Yan et al., 2019).
Specifically, we compare the observed degrees with the fitted degrees based on estimated graphs constructed from our method and that of Yan et al.(2019).
To do that, define $d_i^O=(n-1)^{-1}\sum_{j\neq i}A_{ij}$ and $d_i^I=(n-1)^{-1}\sum_{i\neq j}A_{ij}$ for $i\in[n]$.
To examine the logistic parametric assumption, we consider the following model presented in Yan et al. (2019):
\[
\mathbb{P}(A_{ij}|X_{ij})=p_{ij}=\frac{e^{\alpha_{0i}+\beta_{0j}+X_{ij}^\top\gamma_0}}{1+e^{\alpha_{0i}+\beta_{0j}+X_{ij}^\top\gamma_0}}.
\]
To estimate the unknown parameters $\alpha_{0i},~\beta_{0j},$ and $\gamma_0$,
we apply the maximum likelihood method developed in Yan et al. (2019).
The corresponding estimators are denoted as $\breve{\alpha}_i,~\breve{\beta}_j,$ and $\breve{\gamma}$, respectively.
For each pair $(i,j)$, we then estimate $p_{ij}$ using $\breve p_{ij}$,
which is obtained by substituting $\alpha_{0i},~\beta_{0j},$ and $\gamma_0$ in the above model
with $\breve{\alpha}_i, \breve{\beta}_j,$ and $\breve{\gamma}$, respectively.
Subsequently, we calculate the estimates of $d_i^O$ and $d_i^I$ using
$\breve{d}_i^O=(n-1)^{-1}\sum_{j\neq i}\breve A_{ij}$ \text{and} $\breve{d}_i^I=(n-1)^{-1}\sum_{j\neq i}\breve A_{ji},$ respectively,
where $\breve A_{ij}=1$ if $\breve p_{ij}>0.5$ and 0 otherwise.
Finally, we display the scatter plots of $(d_i^O, \breve{d}_i^O)$ and $(d_i^I, \breve{d}_i^I)$.
For the proposed method, we can estimate $d_i^O$ and $d_i^I$ in a similar manner.
Denote the corresponding estimators as $\widehat{d}_i^I$ and $\widehat{d}_i^O$, respectively.

\begin{figure}
\centering
\includegraphics[width=0.8\textwidth]{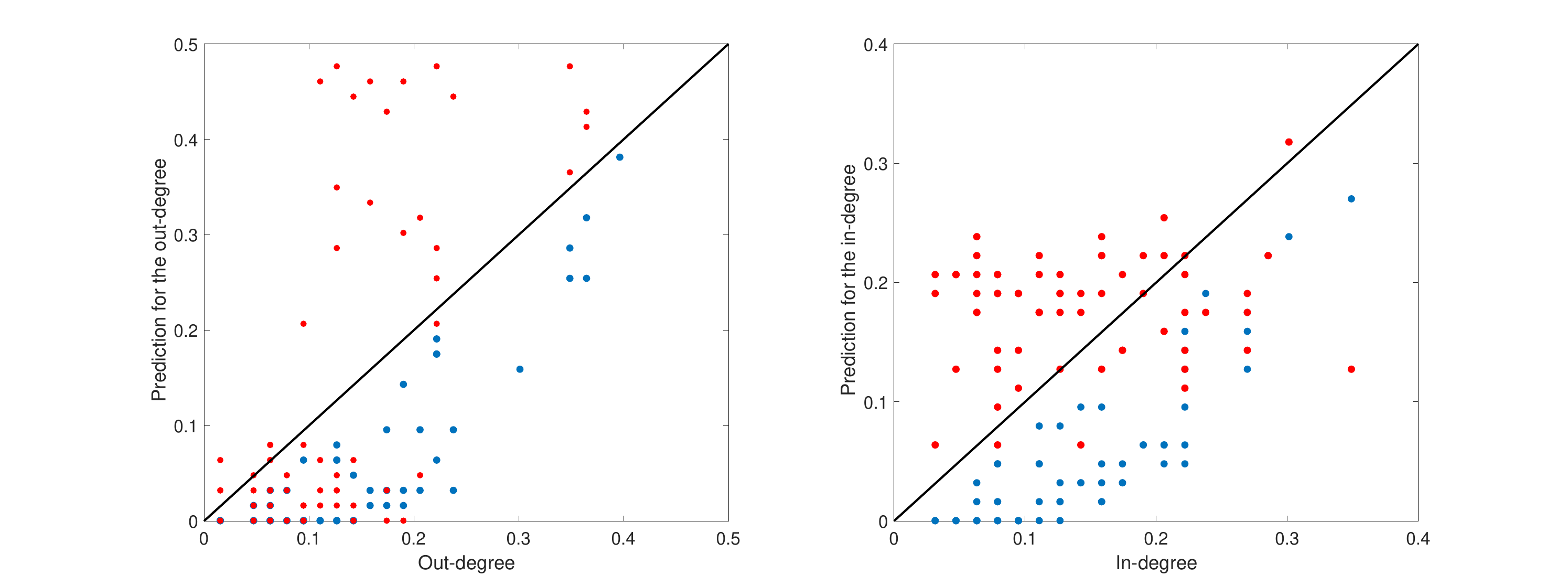}
\caption*{Figure S3: The left panel presents the scatter plots of $(d_i^O,\widehat d_i^O)$ (red dot) and $(d_i^O,\breve d_i^O)$ (blue dot),
while the right panel displays the scatter plots of $(d_i^I,\widehat d_i^I)$ (red dot) and $(d_i^I,\breve d_i^I)$ (blue dot).}
\label{figure-A}
\end{figure}

The results are presented in Figure S3.
We observe that $\breve{d}_i^O$ and $\breve{d}_i^I$ significantly underestimate $d_i^O$ and $d_i^I$,
indicating that the logistic parametric assumption may be violated.
In addition, we calculate the $L_2$-norm errors, i.e., $\| d_i^O - \breve{d}_i^O \|_2$ and $\| d_i^I - \breve{d}_i^I\|_2$.
The results are 0.738 and 0.732 for our method while they are 0.808 and 0.780 in Yan et al.'s method.
This indicates an improvement with our proposed method.

\begin{center}
{\large \textsc{References}}
\end{center}
\begin{description}
\itemsep=-\parsep
\itemindent=-1cm

\medskip\item
Andrews, D. W. (1995). Nonparametric kernel estimation for semiparametric models.
{\it Econometric Theory}, 11, 560-586.
\medskip\item
Cai, T. T., Liu, W. and Xia, Y. (2014).
Two-sample test of high dimensional means under dependence.
{\it Journal of the Royal Statistical Society Series B: Statistical Methodology, 76}, 349-372.
\medskip\item
Candelaria, L. E. (2020). A semiparametric network formation model with unobserved linear
heterogeneity. {\it arXiv preprint arXiv:2203.15603}.

\medskip\item
Chernozhukov, V., Chetverikov, D. and Kato, K. (2017).
Central limit theorems and bootstrap in high dimensions.
{\it The Annals of Probability, 45}, 2309-2352.
\medskip\item
Chernozhuokov, V., Chetverikov, D., Kato, K. and Koike, Y. (2022).
Improved central limit theorem and bootstrap approximations in high dimensions.
{\it The Annals of Statistics, 50}, 2562-2586.
\medskip\item
Graham, B. S. (2017). An econometric model of network formation with degree heterogeneity.
{\it Econometrica}, {\bf 85}, 1033-1063.
\medskip\item
Hoeffding , W. (1963).
Probability inequalities for sums of bounded random variables.
{\it Journal of the American Statistical Association, 58}, 13-30.
\medskip\item
Vershynin, R. (2018). {\it High-Dimensional Probability: An Introduction with Applications in Data Science}.
Cambridge university press.
\medskip\item
Yan, T., Leng, C. and Zhu, J. (2016).
Asymptotics in directed exponential random graph models with an increasing bi-degree sequence.
{\it The Annals of Statistics, 44}, 31-57,
\medskip\item
Yan, T., Jiang, B., Fienberg, S. E. and Leng, C. (2019).
Statistical inference in a directed network model with covariates.
{\it Journal of the American Statistical Association}, {\bf 114}, 857-868.
\end{description}

\clearpage
\newpage

\begin{figure}
\centering
\includegraphics[width=1\textwidth]{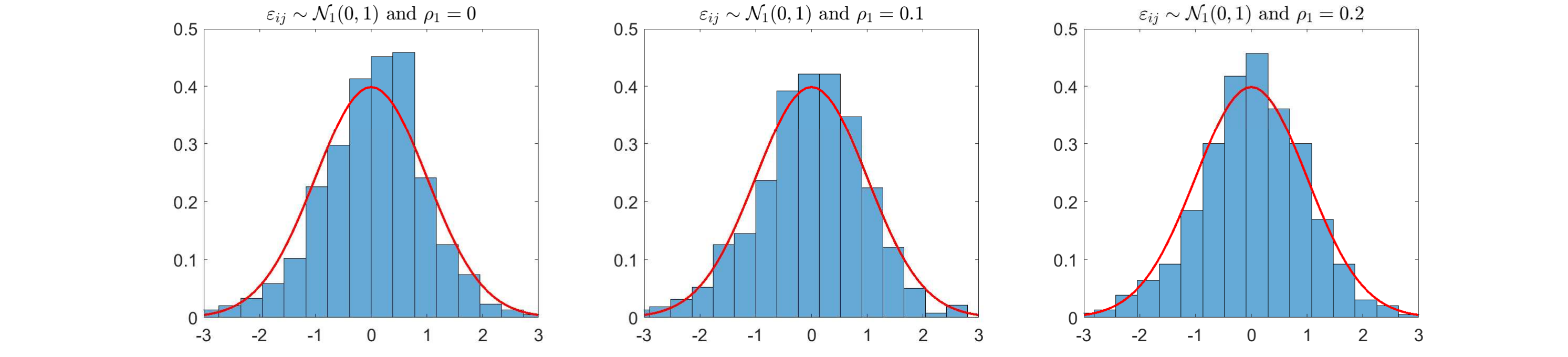}
\includegraphics[width=1\textwidth]{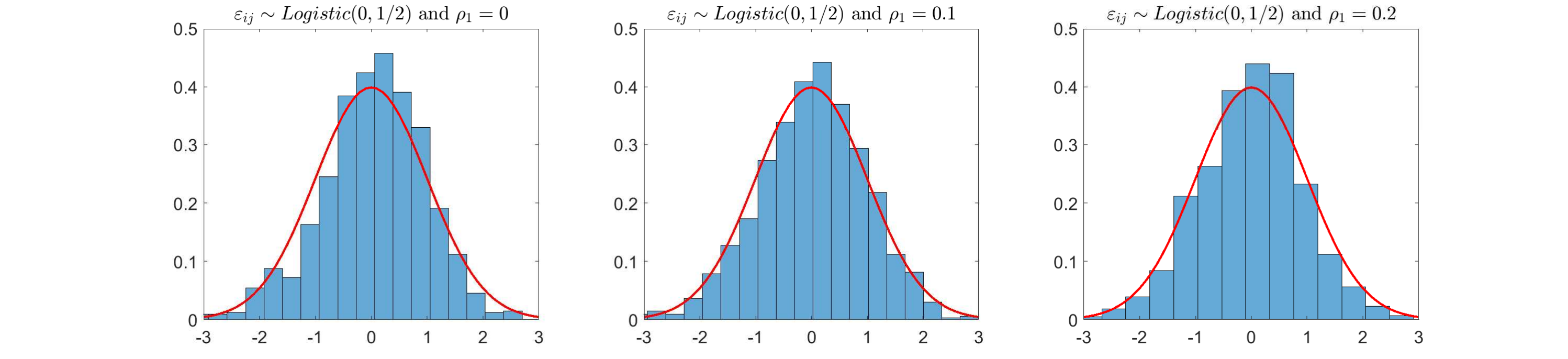}
\includegraphics[width=1\textwidth]{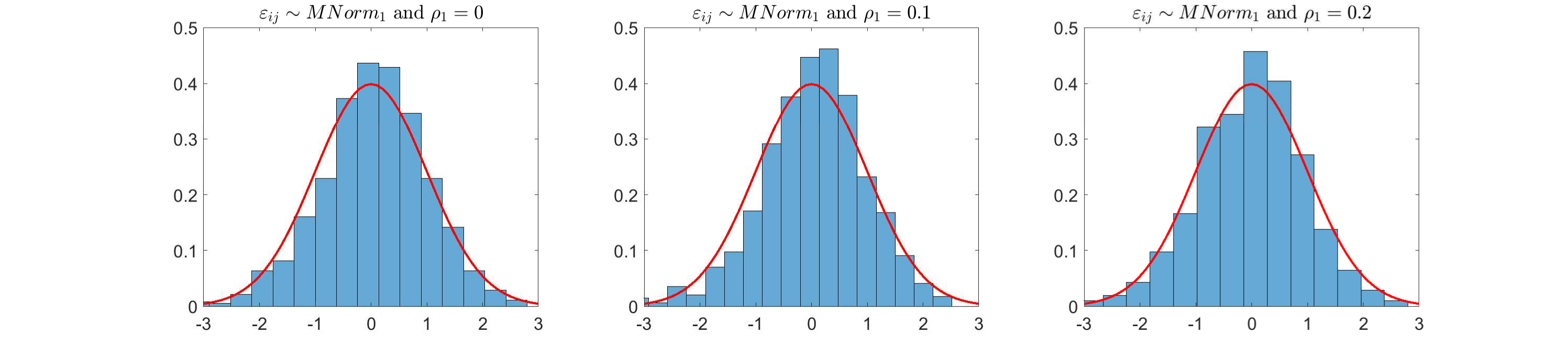}
\caption*{Figure S4: Asymptotical null distribution of the standardized estimator $\widehat\alpha_{n/2}$.
The number of nodes is $50$. The red solid lines denote the density function of the standard normal distribution.}\label{Figure:Normal}
\end{figure}

\begin{figure}
\centering
\includegraphics[width=0.8\textwidth]{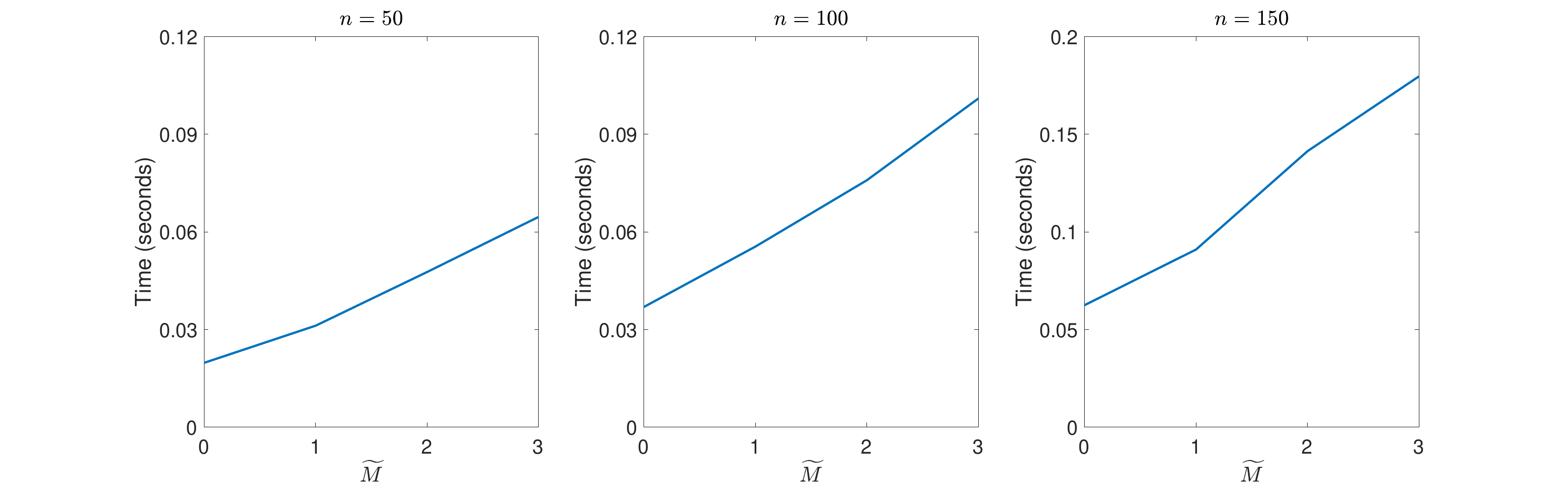}
\caption*{Figure S5: Time (seconds) for calculating $\mathcal{T}_{\alpha,D}(\widetilde M)$ with $\widetilde M\in\{0, 1, 2, 3\}$.
The noise $\varepsilon_{ij}$ follows $\mathcal{N}(0,0.25)$.}
\label{figure-A}
\end{figure}

\begin{figure}
\centering
\includegraphics[width=0.35\textwidth]{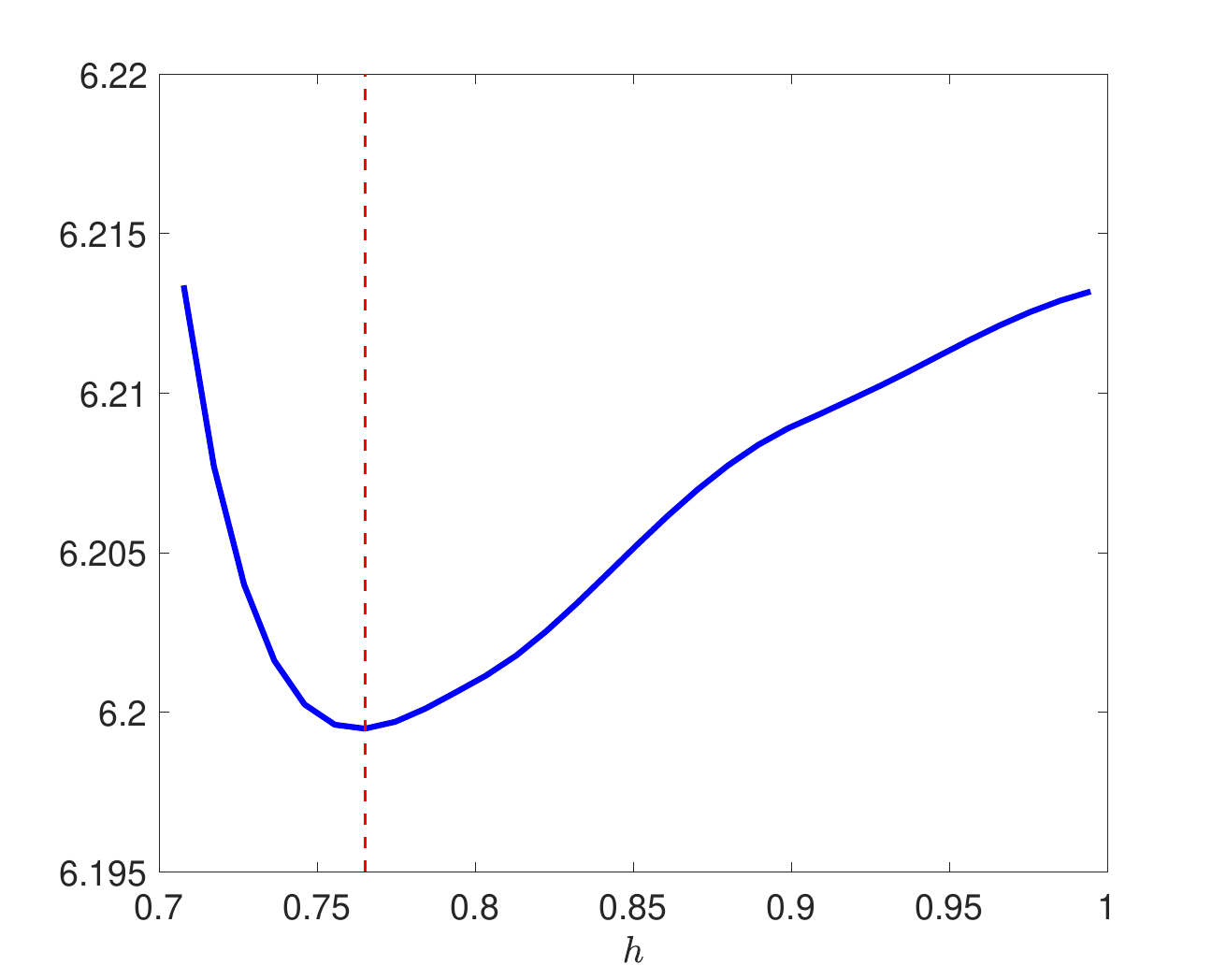}
\caption*{Figure S6: Prediction errors versus bandwidths.
The optimal bandwidth is $\widehat h=0.7651.$}\label{opth}
\end{figure}

\end{document}